\documentclass[%
  reprint,            %
nofootinbib,
 amsmath,amssymb,
 aps,
]{revtex4-2}

\usepackage{graphicx}%
\usepackage{dcolumn}%
\usepackage[caption=false]{subfig}
\usepackage{threeparttable}
\usepackage{booktabs}%
\usepackage{bm}%
\usepackage{braket}
\usepackage[dvipsnames]{xcolor}
\usepackage{amsfonts}
\usepackage{xfrac} %
\usepackage{hyperref}%
\usepackage{cleveref}
\usepackage{comment}
\usepackage{array}
\usepackage{xspace}

\newcommand\Tstrut{\rule{0pt}{2.6ex}}         %

\DeclareMathOperator{\tr}{tr}

\DeclareMathOperator{\Var}{Var}

\usepackage{amsthm}
\newtheorem{theorem}{Theorem}
\newtheorem{lemma}{Lemma}

\newtheorem{remark}{Remark}
\crefname{equation}{Eq.}{Eqs.}
\Crefname{equation}{Eq.}{Eqs.}
\crefname{figure}{Fig.}{Figs.}
\Crefname{figure}{Fig.}{Figs.}
\crefname{subfigure}{Fig.}{Figs.}
\Crefname{subfigure}{Fig.}{Figs.}
\crefname{table}{Table}{Tables}
\Crefname{table}{Table}{Tables}
\crefname{section}{Sec.}{Secs.}
\Crefname{section}{Sec.}{Secs.}
\crefname{subsection}{Sec.}{Secs.}
\Crefname{subsection}{Sec.}{Secs.}
\crefname{subsubsection}{Sec.}{Secs.}
\Crefname{subsubsection}{Sec.}{Secs.}
\crefname{appendix}{Appendix}{Appendices}
\Crefname{appendix}{Appendix}{Appendices}
\crefname{theorem}{Theorem}{Theorems}
\Crefname{theorem}{Theorem}{Theorems}
\crefname{lemma}{Lemma}{Lemmas}
\Crefname{lemma}{Lemma}{Lemmas}
\crefname{proposition}{Proposition}{Propositions}
\Crefname{proposition}{Proposition}{Propositions}
\crefname{corollary}{Corollary}{Corollaries}
\Crefname{corollary}{Corollary}{Corollaries}
\crefname{definition}{Definition}{Definitions}
\Crefname{definition}{Definition}{Definitions}
\crefname{example}{Example}{Examples}
\Crefname{example}{Example}{Examples}
\crefname{remark}{Remark}{Remarks}
\Crefname{remark}{Remark}{Remarks}
\newtheorem{protocol}{Protocol}
\crefname{protocol}{Protocol}{Protocols}
\Crefname{protocol}{Protocol}{Protocols}
\newtheorem{procedure}{Procedure}
\crefname{procedure}{Procedure}{Procedures}
\Crefname{procedure}{Procedure}{Procedures}

\newcommand{\seqAA}{\ifmmode\mathrm{SeqAA}\else{\emph{SeqAA}}\xspace\fi}
\newcommand{\renAA}{\ifmmode\mathrm{RenAA}\else{\emph{RenAA}}\xspace\fi}
\newcommand{\aao}{\ifmmode\mathrm{A@O}\else{A@O}\xspace\fi}
\newcommand{\snc}{\ifmmode\mathrm{S\&C}\else{S\&C}\xspace\fi}
\providecommand{\dket}[1]{|#1\rangle\!\rangle}
\providecommand{\dbra}[1]{\langle\!\langle #1|}
\providecommand{\dbraket}[2]{\langle\!\langle #1|#2\rangle\!\rangle}
\begin{document}

\title{Exact log-depth preparation of highly entangled matrix product states}

\author{Keisuke Murota}
\affiliation{Quantinuum K.K., Otemachi Financial City Grand Cube 3F, 1-9-2 Otemachi, Chiyoda-ku, Tokyo, Japan}

\author{Fr\'ed\'eric Sauvage}
\affiliation{Quantinuum, Partnership House, Carlisle Place, London SW1P 1BX, United Kingdom}

\author{Marco Ballarin}
\affiliation{Quantinuum, Partnership House, Carlisle Place, London SW1P 1BX, United Kingdom}

\author{Gabriel Matos}
\affiliation{Quantinuum, Partnership House, Carlisle Place, London SW1P 1BX, United Kingdom}

\author{Enrico Rinaldi}
\affiliation{Quantinuum, Partnership House, Carlisle Place, London SW1P 1BX, United Kingdom}

\begin{abstract}
    Preparing matrix product states (MPS) on a quantum device is a key subroutine in many quantum algorithms.
    The most competitive methods, based on the renormalisation group, prepare translationally invariant MPS of size $L$ and bond dimension $\chi$, up to an error $\varepsilon$, in circuit depth $\tilde O(\chi^{4}\log(L/\varepsilon))$ or $\tilde O(\chi^{6}\log\log(L/\varepsilon))$.
    We improve multiple aspects of these methods.
    First, using block-encoded correction maps, whose post-selection succeeds with constant probability, we render the preparation exact without sacrificing the scaling in $L$.
    Second, through a generalisation of oblivious amplitude amplification to isometries, we reduce the bond-dimension dependence, improving the depth to $\tilde O(\chi^{2}\log L + \chi^{4})$ or $\tilde O(\chi^{2}\log\log L + \chi^{4})$, and even to $\tilde O(\chi^{3}\log L)$ for incoherent preparations.
    Finally, we extend the framework to non-translationally invariant MPS and prove logarithmic-depth exact preparation for independent and identically distributed random tensor sequences.
    Confirmed by numerical studies, these results constitute, to the best of our knowledge, the most efficient exact MPS preparation protocols in the relevant parameter regimes.

\end{abstract}

\maketitle

\section{Introduction}

Matrix product states (MPS) provide efficient classical representations of quantum many-body states, and form the basis of tensor-network methods widely used in computational physics~\cite{white1992,fannes1992,vidal_2003,schollwock2011,orus2014, cirac2021mpspeps}. 
They represent the ground states of gapped one-dimensional systems accurately~\cite{verstraete_2006, hastings_2007} and are employed in fields ranging from quantum chemistry~\cite{ stoudenmire_2012, iqbal_2022} to quantum machine learning~\cite{rieser_2023, Berezutskii2025,murota2025}. 
Building on the success of these techniques, it is of great interest to efficiently prepare states represented by MPS on a quantum device. 
A prepared MPS can serve as a warm start for variational optimisation~\cite{Rudolph2023, Huggins_2019, Dborin_2022}, as a structured initial state for quantum simulations~\cite{martin2024,iqbal_2022, lubasch_2020, robertson_2025, jaderberg_2026}, 
or for loading classical data into a quantum register~\cite{Melnikov_2023, jumade2023dataloadableshortdepth, yuichi_2025, ballarin_2025}
The efficiency of this preparation therefore determines the practicality of these applications. 

Several approaches aiming at preparing generic MPS through quantum circuits have been proposed. 
For instance, an MPS of bond dimension $\chi$ and system size $L$ can always be prepared by a sequential circuit with depth $O(\chi^2 L)$~\cite{schon2005}.
Alternatively, by interpreting an MPS as a balanced tree tensor network (TTN),
one can obtain a quantum circuit preparing it in depth $O(\chi^{6}\log L)$~\cite{shi2006ttn, cervero_2023}.
Heuristic algorithms for MPS preparation have also been devised~\cite{ran_2020,jumade2023dataloadableshortdepth, mansuroglu2025}.
While these can provide shallower circuits, they typically involve numerical optimisations of unitary gates, known to become difficult for large system sizes and lead to approximations.

When focusing on specific families of states, the circuit depths of the MPS preparations can be improved. For instance, if an MPS is efficiently expressible using the multi-scale entanglement renormalisation ansatz (MERA)~\cite{vidal_2007, vidal_2008}, it can be prepared in depth $O(\chi^{4} \log L)$, although this condition does not apply generally. 
Other strategies leverage mid-circuit measurement and feed-forward (MCM-FF)~\cite{raussendorf2001oneway, smith_2023, smith2024, stephen_2025}. 
For instance, \citet{smith_2023, smith2024} 
achieve a circuit depth that is constant in the size of the system, at the cost of a post-selection overhead that scales as $O(\chi^4)$ for open boundary conditions and $O(\chi^5)$ for periodic ones. However, their method can only be applied to a restricted class of MPS, which includes the AKLT state~\cite{aklt} or the GHZ state~\cite{Greenberger1989}.

Recently, methods inspired by the renormalisation group (RG) have been explored~\cite{piroli2021,cirac2021mpspeps, malz2024, scheer2025}. 
These circuit constructions rely on a grouping of the tensors constituting the MPS into blocks of $q$ sites $A^{(q)}$, with polar decomposition $A^{(q)} = V^{(q)}P^{(q)}$, 
and leverage the fact that $P^{(q)}$ quickly approaches an RG fixed point $P^{(\infty)}$ that can be prepared in $O(\chi^2)$ depth.
Together with improved implementations of the isometries $V^{(q)}$,
\citet{malz2024}
achieve circuit depths scaling as $\tilde O(\chi^{4}\log(L/\varepsilon))$, or $\tilde O(\chi^{6}\log\log(L/\varepsilon))$ on all-to-all connected quantum devices or when resorting to MCM-FF
for target precision $\varepsilon$. 
(Note that the $\tilde{O}$ notation hides a $\log(\chi)$ dependence and factors depending on the correlation length $\xi$ of the MPS.)
Such schemes have been demonstrated across a transition between a symmetry-protected topological phase and a trivial phase on an 80-qubit superconducting platform~\cite{scheer2025}.

While these circuit constructions have competitive scalings, 
they only achieve approximate preparation and still carry a large bond-dimension dependence, quickly becoming prohibitive when preparing highly entangled states at scale.
Moreover, although the framework of \citet{malz2024} extends to non-TI MPS, logarithmic-depth preparation was only established numerically, without analytical guarantees.  In this work we improve on these limitations. Concretely, we make three contributions:

\begin{table}[!b]
\centering
\setlength{\tabcolsep}{0pt}
\begin{threeparttable}
\begin{tabular}{l c c}
\hline
Method & Depth of $V^{(q)}$ & Depth of MPS preparation \\ \hline
\multicolumn{3}{c}{\textbf{\citet{malz2024}}} \\ \hline
Sequential & $\tilde{O}(q\,\chi^{4})$ & $\tilde{O}(\chi^{4}\log(L/\varepsilon))$ \Tstrut \\[0.25em]
Tree & $\tilde{O}(\chi^{6}\log q)$ & $\tilde{O}(\chi^{6}\log\log(L/\varepsilon))$ \\[0.25em]\hline
\multicolumn{3}{c}{\textbf{Ours} } \\ \hline 
\seqAA & $\tilde{O}(q\chi^{2} + \chi^{4})$ & $\tilde{O}(\chi^{2}\log L + \chi^{4})$ \Tstrut \\[0.25em]
\renAA & $\tilde{O}(\chi^{2}\log q + \chi^{4})$ & $\tilde{O}(\chi^{2}\log\log L + \chi^{4})$ \\[0.25em]
Quasi-probabilistic\tnote{*} & $\tilde{O}(q\,\chi^{3})$ & $\tilde{O}(\chi^{3}\log L)$ \\ \hline\hline
\end{tabular}
\begin{tablenotes}
\item[*] Incoherent implementation.
\end{tablenotes}
\caption{
  Comparison between the circuit depth scalings for the constructions introduced in this work and those in \citet{malz2024}. 
  The left column shows the depth to implement the isometry $V^{(q)}$ for a renormalised block of size $q$ and bond dimension $\chi$ (see \cref{sec:background} for definitions of these concepts).
  The right column shows the overall MPS preparation depth for a system of size $L$, obtained by substituting $q = \tilde{O}(\log L)$.
  The $\tilde{O}$ notation hides $\log(\chi)$ and factors depending on the correlation length $\xi$ of the MPS.
  The methods of~\citet{malz2024} are approximate, with controllable error $\varepsilon$, while ours achieve exact preparation.
}
\label{tab:Vq-depth-summary}
\end{threeparttable}
\end{table}

\begin{enumerate}
    \item For TI MPS, we introduce a \emph{correction map} that converts the RG fixed point $P^{(\infty)}$ into $P^{(q)}$, and identify conditions under which its block-encoded implementation succeeds with \emph{constant} post-selection probability. This makes the preparation exact at constant overhead while retaining the system-size scalings of~\citet{malz2024} (\cref{sec:exact-prep}). This exact preparation can even be made deterministic through amplitude amplification (\cref{subsec:improve_success}).

    \item Employing the correction maps together with a generalisation of oblivious amplitude amplification~\cite{brassard2002, berry2015taylor, gilyen2019qsvt, styliaris_2025} to isometries (OAAI in~\cref{thm:OAAI}), we propose several implementations of $V^{(q)}$ with improved bond-dimension scalings, which we refer to as \seqAA, \renAA, and quasi-probabilistic (\cref{sec:implementation-Vq}).

    \item Leveraging the theory of ergodic quantum processes~\cite{Movassagh}, we rigorously extend the framework to non-TI MPS (\cref{sec:nonTI}) by proving logarithmic-depth exact preparation for independent and identically distributed (IID) random tensor sequences~\cite{garnerone2010typicality, garnerone2010rmps}. 
\end{enumerate}

A representative of our MPS-preparation circuit is depicted in~\cref{fig:intro-overview} and our main circuit-depth scalings are summarised in~\cref{tab:Vq-depth-summary}.
In the \emph{coherent} setting, where the goal is to 
prepare the TI MPS, our best construction achieves exact preparation in depth $\tilde{O}(\chi^{2}\log L + \chi^{4})$ in practically relevant regimes, improving on the approximate $\tilde{O}(\chi^{4}\log(L/\varepsilon))$ of~\citet{malz2024}.
We confirm these results numerically in~\cref{sec:numerics}. 
The numerics further support that this logarithmic scaling persists for non-TI MPS beyond the IID ensembles covered by our theorem, as showcased for ground states of a disordered Heisenberg model (\cref{sec:dmrg_impurity}).
Additionally, when the goal is not to prepare the state itself but rather to estimate expectation values of observables, this does not need to be done using a single circuit. 
We can instead perform a weighted average over an ensemble of states prepared by shorter-depth circuits, arising from a quasi-probabilistic decomposition, at the cost of a multiplicative sampling overhead. In this \emph{incoherent} setting, our circuits achieve depth $\tilde{O}(\chi^{3}\log L)$ while still ensuring that the sampling overhead remains constant.

To the best of our knowledge, these circuit constructions constitute the most efficient exact MPS preparation protocols in the relevant parameter regimes, advancing MPS preparation, especially for highly entangled states.

The manuscript is structured as follows.
In~\cref{sec:background}, we formalise the MPS preparation task and introduce the RG-inspired circuit constructions that we base ourselves on~\cite{malz2024}. \Cref{sec:exact-prep} details our exact preparation protocols, proving a constant post-selection success probability
with logarithmic depth (or double-logarithmic in the all-to-all model, or when resorting to MCM-FF).~\Cref{sec:implementation-Vq} introduces improved implementations of $V^{(q)}$ that leverage OAAI. In \cref{sec:nonTI}, we extend the reach of our methods to families of MPS that do not satisfy translational invariance.
Finally, in~\cref{sec:numerics}, we report numerical results comparing the different methods studied.

\begin{figure}[!ht]
\centering
    \includegraphics[width=\columnwidth]{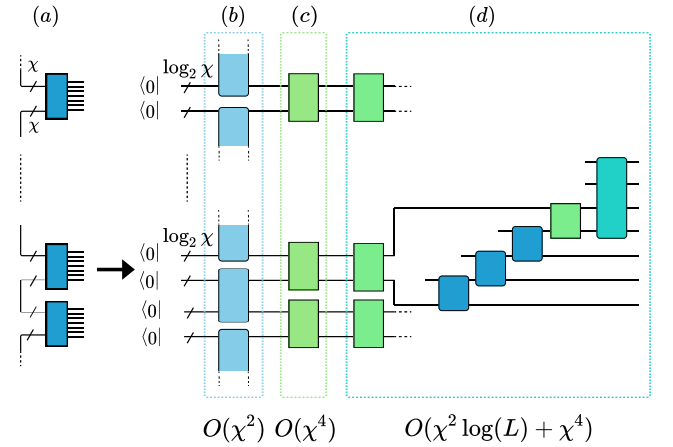}
    \caption{
    Circuit construction for the preparation of a TI MPS with bond dimension $\chi$ (a). %
    The representative circuit (b-d) sketched here corresponds to the method \emph{SeqAA} of~\cref{tab:Vq-depth-summary}.
    As detailed in~\cref{sec:background,sec:exact-prep}, the circuit consists of the preparation of an RG fixed point $P^{(\infty)}$ (b), a layer of correction maps $C^{(q)}$ (c), and a layer of isometries $V^{(q)}$ (d).
    Key to our circuit constructions are the correction maps (green rectangles) implemented through block-encoding. These ensure exactness of the circuit preparations and are used to reduce the circuit depths for the isometries $V^{(q)}$.
    }
    \label{fig:intro-overview}
\end{figure}
\section{Background} \label{sec:background}

We begin by establishing the notations and the problem setting in~\cref{sec:bg-notation}. 
We then review the RG based approach that forms the backbone of our algorithm in~\cref{sec:bg-rg}.

\subsection{Problem setting} \label{sec:bg-notation}

Consider a 3-leg tensor $A \in \mathbb{C}^{\chi \times d \times \chi}$, 
with $\chi$ the virtual (bond) dimension and $d$ the local physical dimension. 
In most of what follows, our goal will be to prepare an $L$-site translationally invariant MPS defined as
\begin{equation}
  \label{eq:MPS-def-indices}
  \ket{{A}(L)} = \sum_{i_1,\hdots,i_L=1}^d 
  {\rm tr}\bigl[A^{i_1} \hdots A^{i_L}\bigr]  \ket{i_1 \hdots i_L},
\end{equation}
where we denote as $A^i\in \mathbb{C}^{\chi \times \chi}$ the matrix corresponding to the $i$-th physical dimension of the tensor $A$.
(Throughout this paper, we will use Latin letters for the physical indices and Greek letters for the virtual indices.)
Diagrammatically, the TI MPS of~\cref{eq:MPS-def-indices} is represented as 
\begin{equation}
  \label{eq:MPS-def}
  \ket{{A}(L)} =
  \raisebox{-0.25\height}{\hspace{1em}\includegraphics[width=0.33\textwidth]{figs/TIMPS.pdf}\hspace{3pt}.  \\[6pt] }
\end{equation}

By leveraging a gauge freedom, every TI MPS can be brought to a block-diagonal canonical form with $r$ blocks (see~\cite[Section 2.3]{cirac_2017} or~\cref{app:normal-basic} for a precise definition).
We call an MPS \emph{normal}~\cite{perez_garcia_2007} when $r = 1$. When the tensor $A$ satisfies the left-canonical condition $\sum_i A^{i\dag} A^i  = \mathbb{I}_{\chi}$, where $\mathbb{I}_{\chi}$ is the identity on $\mathbb{C}^\chi$, the state is asymptotically normalised  (see~\cref{app:normal-basic}):
\begin{equation}
  \label{eq:norm}
  \bigg| 1-\| \ket{A(L)} \|^2 \bigg| = O(e^{-L / \xi}),
\end{equation}
with $\|\cdot\|$ the $2$-norm of a vector, 
and $\xi>0$ the correlation length which is set by the MPS transfer matrix
\begin{equation}\label{eq:transfer-matrix-main}
  E := \sum_i A^i \otimes \bar{A}^i \in \mathbb{C}^{\chi^2\times\chi^2}.
\end{equation}
For a normal MPS, $E$ has a unique leading eigenvalue $\lambda_1=1$ and we define $\xi := -1/\log|\lambda_2|$ with $\lambda_2$ the subdominant eigenvalue ($|\lambda_2|<1$).
In the main text, we restrict our focus to \emph{normal} MPS (see~\cref{app:non-normal} for a treatment of non-normal TI MPS), 
such that for the large system sizes $L\gg 1$, we can assume $\| \ket{A(L)} \| = 1$.

Two primitive manipulations of tensors will be repeatedly used,  the \emph{blocking} of tensors and their \emph{polar decompositions}, that are now detailed.
The blocking of $q$ consecutive tensors $A$ gives $A^{(q)} \in \mathbb{C}^{\chi \times d^{q} \times \chi}$ defined through its entries as
\begin{equation}
  \begin{split}
  &\bigl[A^{(q)}\bigr]^{(i_1\dots i_q)}_{\alpha,\beta}
  =
  \sum_{\gamma_1,\dots,\gamma_{q-1} = 1}^{\chi}
  \bigl[A\bigr]^{i_1}_{\alpha, \gamma_1}
  \bigl[A\bigr]^{i_2}_{\gamma_1, \gamma_2}
  \cdots
  \bigl[A\bigr]^{i_q}_{\gamma_{q-1}, \beta},%
  \end{split}
\end{equation}
and we refer to $q$ as the \textit{block size}. It is represented as
\begin{equation}\label{eq:blocking}
  \raisebox{-0.25\height}{\hspace{1em}\includegraphics[width=0.4\textwidth]{figs/block.pdf}\hspace{3pt}. \\[6pt]}
\end{equation}
Interpreting $A^{(q)}$ as a map from the virtual space ($\mathbb{C}^{\chi ^2}$) to the physical space ($\mathbb{C}^{d^q}$), we choose the block size large enough such that this map is injective. This requires $d^q \geq \chi^2$ and is possible for normal TI MPS after finite blocking~\cite{perez_garcia_2007}.
For such block sizes, one can perform a polar decomposition of $A^{(q)}$, yielding
\begin{equation}
  \raisebox{-0.25\height}{\hspace{1em}\includegraphics[width=0.22\textwidth]{figs/polar.pdf}\hspace{3pt}, \\[6pt] }
  \label{eq:polar}
\end{equation}
with $P^{(q)} = \sqrt{\bigl(A^{(q)}\bigr)^{\dag} A^{(q)}} \in\mathbb{C}^{\chi^{2} \times \chi^{2}}$ a
positive-definite matrix
and $V^{(q)}=A^{(q)}(P^{(q)})^{-1}\in\mathbb{C}^{d^{q}\times \chi^2}$ an isometry, i.e. satisfying
$V^{(q)\dag}V^{(q)}=\mathbb{I}_{\chi^2}$. 
Splitting the $L$-site chain~\eqref{eq:MPS-def}
into $L' := L/q$ blocks, and performing $q$-site blockings, yields the identity $\ket{A(L)} = \ket{A^{(q)}(L')}$. The tensor obtained after polar decomposition of each of the blocks is depicted in~\cref{fig:exact-prep}~(a). This will serve us as a starting point for all the circuit constructions presented here.

Finally, as our aim is to prepare TI MPS as quantum circuits, we need to define a measure of circuit complexity. 
Throughout this paper, we measure circuit complexity using CNOT-depth, and refer to it simply as the circuit depth. It is defined as the number of layers of CNOT gates (or equivalent two-qubit entangling operations) required to implement a circuit.
When assessing scaling of circuit depths, we will repeatedly resort to the fact that a state $\ket{\psi}\in \mathbb{C}^M$ can be implemented in depth $O(M)$ while an isometry $V \in \mathbb{C}^{M \times N}$  can be implemented in depth $O(MN)$~\cite{iten2016}.

\subsection{RG approach} \label{sec:bg-rg}

The decomposition of~\cref{eq:polar} can be interpreted as a factorisation of the blocked tensor $A^{(q)}$ into an RG component $P^{(q)}$, which can be approximated by a local state, and an isometry component $V^{(q)}$  handling short-range interactions.
This forms the backbone of an efficient, but approximate, MPS preparation of $\ket{{A}(L)}$~\cite{malz2024, piroli2021}, reviewed here. 

First, consider $\ket{{P^{(q)}}(L')} \in \mathbb{C}^{\chi^{2L'}}$, the MPS associated with a chain of $L'$ tensors $P^{(q)}$.
We refer to this state as the \emph{renormalised state}, and refer to each group of $q$ original sites as a \emph{renormalised block}.
In what follows, $\ket{{P^{(q)}}(L')}$ is approximated by $\ket{{P^{(\infty)}}(L')}$ defined as
\begin{equation}
P^{(\infty)} := \lim_{q \to \infty} P^{(q)},
\end{equation}
the fixed point of $P^{(q)}$ as $q$ grows. As such, we refer to $\ket{{P^{(\infty)}}(L')}$ as the \emph{RG fixed-point state}. As we now detail, this approximation has the merit of having a shallow circuit implementation at the cost of an approximation error that can be controlled by varying the block size.

For normal MPS, $P^{(\infty)}$ admits the factorisation
\begin{equation}\label{eq:pinf_product}
   P^{(\infty)} = \sigma \otimes \mathbb{I}_{\chi},
\end{equation}
where, by injectivity of the normal canonical tensor, $\sigma$ is positive-definite ($\sigma \succ 0$) and satisfies $\tr[\sigma \sigma^\dag] = 1$~\cite{malz2024, piroli2021}. 
In practice, $\sigma$ is obtained from diagonalisation of the transfer matrix $E$ in Eq.~\eqref{eq:transfer-matrix-main} associated with $A$ (see \cref{app:compute-sigma}). Diagrammatically,~\cref{eq:pinf_product} is depicted as
\begin{equation}
  \label{eq:infty-sigma}
  \raisebox{-0.25\height}{\includegraphics[width=0.19\textwidth]{figs/infty_sigma.pdf}} \hspace{3pt},
\end{equation}
showing that a chain of repeated tensors becomes a tensor product of identical locally entangled states $\sigma$:
\begin{equation}\label{eq:p_infinity_line}
  \ket{{P^{(\infty)}}(L')} = \ket{\sigma}^{\otimes L'} \in \mathbb{C}^{\chi^{2L'}},
\end{equation}
or diagrammatically as
\begin{equation}\label{eq:p_infinity}
  \raisebox{-0.05\height}{\includegraphics[width=0.35\textwidth]{figs/infty_chain.pdf}} \hspace{3pt}.
\end{equation}
Note that each $\ket{\sigma}$ is a normalised state of dimension $\chi^2$ since $\tr[\sigma \sigma^\dag] = \|\ket{\sigma}\|= 1$, and 
they can all be synthesised in parallel by a circuit of depth
\begin{equation}
\label{eq:D_inf}
 D_\infty = O(\chi^{2}).
\end{equation}

To capture discrepancies between $P^{(q)}$ and its fixed point $P^{(\infty)}$, let us define 
the residual matrix as
\begin{equation}
\label{eq:residual}
    R^{(q)} = P^{(q)} - P^{(\infty)}.
\end{equation} 
Notably, and assuming that the transfer matrix of $A$ is diagonalisable, we prove that the norm of the residual decays exponentially with the block size: 
\begin{equation}
\|R^{(q)}\|_{2} \leq \Gamma\, e^{-q/\xi}.
\label{residual}
\end{equation}
As seen before, $\xi$ is the correlation length of $A$ while $\|\cdot\|_2$ denotes the Frobenius norm. The constant $\Gamma > 0$ depends only on the MPS tensor and is defined as
\begin{equation}\label{def:gamma}
  \Gamma := \chi\,\|S\|_2\,\|S^{-1}\|_2\,\|(P^{(\infty)})^{-1}\|_2,
\end{equation}
where $S$ is the matrix that diagonalises the transfer matrix $E$ (see \cref{app:residual-bound} for the derivation and a treatment of the non-diagonalisable case).

Define the trace distance (for pure states) as
$D_{\rm tr}(\ket{\psi},\ket{\phi}):=\frac12\|\ket{\psi}\!\bra{\psi}-\ket{\phi}\!\bra{\phi}\|_1 
= \sqrt{1-|\braket{\psi|\phi}|^2}
$.
It follows from~\cref{residual} that a choice of block size
\begin{equation}\label{eq:scaling-q}
    q=O \left( \xi \log \left (\frac{\Gamma L}{\varepsilon} \right )\right)
\end{equation}
ensures $\varepsilon$-error in trace distance: 
\begin{equation}\label{eq:trace-distance-error}
  D_{\rm tr}\!\left(\ket{P^{(q)}(L')},\ket{P^{(\infty)}(L')}\right)
  \le \varepsilon .
\end{equation}
The derivation, together with a treatment of the non-diagonalisable case,
is provided in~\cref{app:trace-distance}.

\begin{figure}[htp]
  \centering
  \includegraphics[width=0.45\textwidth]{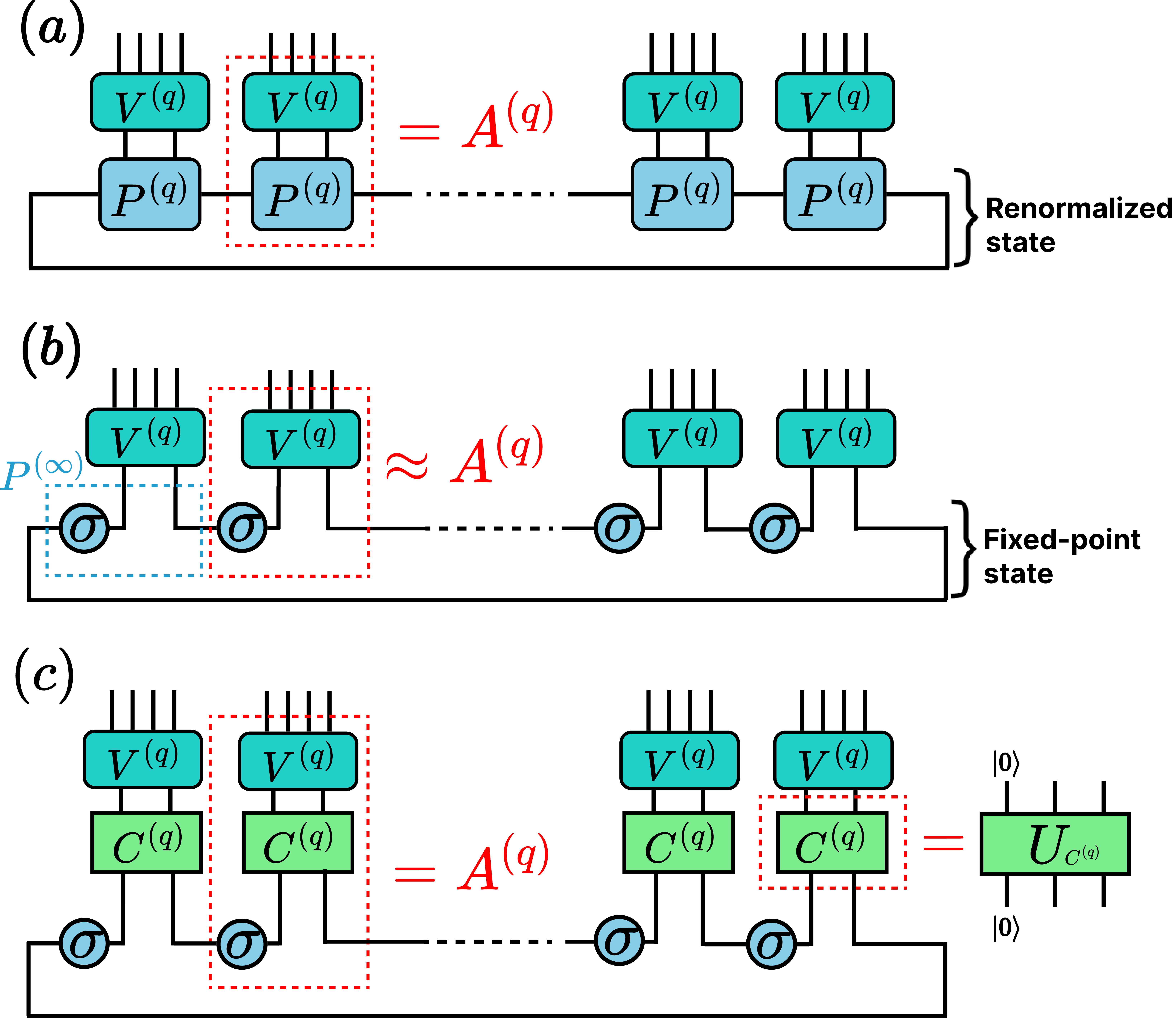}
  \caption{
  (a)~The TI MPS to be prepared~\eqref{eq:MPS-def} after the steps of blocking~\eqref{eq:blocking} and polar decomposition~\eqref{eq:polar}.
  (b)~For normal TI MPS, the tensor $P^{(q)}$ can be approximated through its fixed-point $P^{(\infty)}$~\eqref{eq:infty-sigma} that can be prepared in $O(\chi^2)$ depth~\eqref{eq:p_infinity}.
  Together with the implementation of $V^{(q)}$~\eqref{eq:depth-costs-bg}, this forms the basis of the efficient approximate MPS preparation of \citet{malz2024}, with scalings reported in~\cref{tab:Vq-depth-summary}.
  (c)~Introducing a correction map $C^{(q)}$, implemented through block-encoding, one can achieve similar scalings for exact state preparation.
  }
  \label{fig:exact-prep}
\end{figure}

After preparing the RG fixed-point state $\ket{{P^{(\infty)}}(L')}$, it remains to apply the isometries $V^{(q)}$ to each of the $L'$ renormalised blocks. This can be done in parallel, and 
given that the trace distance is isometry-invariant,
this prepares the desired TI MPS with $\varepsilon$-error. The resulting protocol~\cite{malz2024} is sketched in~\cref{fig:exact-prep}(b) and summarised in the following.

\begin{protocol}[RG-based approximate TI MPS preparation with $\varepsilon$-error]
  \label{protocol:rg-approx-prep}
\leavevmode
\begin{enumerate}
  \item \textbf{Choose a block size.}
        Choose $q$ as per~\cref{eq:scaling-q} to guarantee $\varepsilon$-error~\eqref{eq:trace-distance-error}, and define \(L' = L/q\).

  \item \textbf{Classical preprocessing.}
        Perform the polar decomposition \(A^{(q)} = V^{(q)}P^{(q)}\)~\eqref{eq:polar}, and obtain the RG fixed-point operator $P^{(\infty)}$ (equivalently, $\sigma$)~\eqref{eq:pinf_product}.
        
  \item \textbf{Prepare the RG fixed-point state.}
        Prepare \(\ket{P^{(\infty)}(L')}\) via~\cref{eq:p_infinity} (normal case). 

  \item \textbf{Implement \(V^{(q)}\).}
        Apply \(V^{(q)}\) in parallel to all \(L'\) renormalised blocks.
\end{enumerate}
\end{protocol}

\paragraph*{Circuit depth}
The depth of~\cref{protocol:rg-approx-prep} is given by
\begin{equation}
  \label{eq:def-Dmalz}
  D_{\mathrm{Malz}}(q) := D_\infty + D_V(q),
\end{equation}
with \(D_V(q)\) the depth of the circuit implementing \(V^{(q)}\).
~\citet{malz2024} presents two constructions, referred to here as the \emph{sequential} and \emph{tree} constructions, with depth:
\begin{equation}
  \begin{split}
    D_{V}(q)   &=
      \begin{cases}
        D_{V}^{\rm seq}(q) = O(\chi^{4} q) , \\[0.25em]
        D_{V}^{\rm tree}(q) = O(\chi^{6}\log q) .
      \end{cases}
  \end{split}
  \label{eq:depth-costs-bg}
\end{equation}
The second one can be achieved either on architectures with all-to-all connectivity, 
or by using MCM-FF to mediate non-local interactions.

As per~\cref{eq:trace-distance-error}, by selecting $q = O\left(\xi\log(\frac{\Gamma L}{\varepsilon})\right)$, 
one can prepare the desired TI MPS~\eqref{eq:MPS-def} with $\varepsilon$-error in depth  
\begin{equation}
  \begin{split}
    D_{\mathrm{Malz}}  &=
      \begin{cases}
        D^{\rm seq}_{\mathrm{Malz}} = O(\xi\chi^{4} \log(\frac{\Gamma L}{\varepsilon})), \\[0.25em]
        D^{\rm tree}_{\mathrm{Malz}} = O\!\left(\chi^{6}\log\!\left(\xi \log(\tfrac{\Gamma L}{\varepsilon})\right)\right).
      \end{cases}
  \end{split}
  \label{eq:depth-costs-bg-overall}
\end{equation}
These depths are reported in~\cref{tab:Vq-depth-summary}.

In~\cref{sec:exact-prep}, we show that the preparation of $\ket{P^{(q)}(L')}$ can be corrected, eliminating any approximation error while still
retaining the scalings in $\chi$ and $L$ in~\cref{eq:depth-costs-bg-overall}. 
In~\cref{sec:implementation-Vq}, we propose improved circuit implementations for the isometries, 
further reducing the scaling in $\chi$.

\section{Exact preparation of TI MPS} \label{sec:exact-prep}
In this section, we devise MPS preparation algorithms that are exact ($\varepsilon=0$). 
To that end, in~\cref{subsec:prob-correction}, we define a correction operation that maps $P^{(\infty)}$ to $P^{(q)}$, and
detail its implementation through block-encoding.
Notably, we show that for an appropriate choice of the block size $q$, effectively scaling as before~\eqref{eq:scaling-q}, the probability of failure of the algorithm (due to the post-selection probability inherent to the block-encoding) can be made \emph{constant}.
This yields the first algorithm that achieves exact (but probabilistic) preparation of TI MPS. We study its average complexity in~\cref{subsec:exact-protocol}, and proceed by detailing how the probabilistic correction can be 
improved
or be made deterministic through amplitude amplification in~\cref{subsec:improve_success}.
While these improvements do not change the scalings of the circuit depths, in practice, they can lead to significant circuit depth reductions that are later probed in the numerical studies of~\cref{sec:numerics}.

\subsection{Probabilistic correction} \label{subsec:prob-correction}
Given a block size $q$, let us define the \emph{correction map}
\(C^{(q)} \in \mathbb{C}^{\chi^2\times \chi^2}\)
through
\begin{align}\label{eq:def-Cq-imp}
  P^{(q)} = C^{(q)} P^{(\infty)}.
\end{align}
For normal MPS, 
recall that $P^{(\infty)}=\sigma\otimes\mathbb{I}_{\chi}$ with $\sigma \succ 0$, and thus is invertible (see~\cref{app:non-normal} for the non-normal case).
Hence, the correction map is well defined and can be expressed as
\begin{align}
  \label{eq:def-Cq}
  C^{(q)} := P^{(q)}\,(P^{(\infty)})^{-1}
  = \mathbb{I}_{\chi^2} + R^{(q)} (P^{(\infty)})^{-1},
\end{align}
where \(R^{(q)} := P^{(q)} - P^{(\infty)}\)~\eqref{eq:residual}. 
Assuming that the RG fixed-point state has been prepared,
the following procedure uses block-encoding (see~\cref{app:block-encoding} for a brief review) to  map $\ket{{P^{(\infty)}}(L')}$ into $\ket{{P^{(q)}}(L')}$.
\begin{procedure}[Probabilistic correction]
\label{procedure:prob-mapping}
\leavevmode
\begin{enumerate}
  \item \textbf{Block-encode the correction map.}
        Construct suitable matrices \(M_1,M_2,M_3\)~\cite{nibbi2024}, such that the correction map, rescaled by its spectral norm $\|C^{(q)}\|_\infty$, can be embedded into the unitary
        \begin{align}
          \label{eq:block-encoding-Cq}
          U_{C^{(q)}}=
          \begin{pmatrix}
            C^{(q)}/\|C^{(q)}\|_\infty & M_{1} \\
            M_{2} & M_{3}
          \end{pmatrix} \in \mathbb{C}^{2\chi^2\times 2\chi^2}.
        \end{align}
  \item \textbf{Apply and post-select.}
        For each of the \(L'\) renormalised blocks, add an ancilla qubit initialised in \(\ket{0}\),
        apply \(U_{C^{(q)}}\), and post-select the ancilla on \(\ket{0}\). 
\end{enumerate}
\end{procedure}

Since each $U_{C^{(q)}}$ acts independently on each block and has dimension $O(\chi^2 \times \chi^2)$, they can all be implemented
 in parallel, resulting in a circuit depth for the correction 
\begin{equation}\label{eq:DC}
    D_C = O(\chi^4).
\end{equation}
As for ancillary resources, the $L'$ ancillas required in~\cref{procedure:prob-mapping} 
are reused later to store the physical MPS (i.e. when implementing the isometries). 
That is, the procedure does not necessitate qubits beyond those needed to store the target MPS (see \cref{app:qubit-count} for an extended discussion).

\cref{procedure:prob-mapping} is inherently probabilistic due to the post-selection.
Its \emph{success probability} $p_{\mathrm{succ}}$ is the probability of measuring all ancillas in $\ket{0}$.
Nonetheless, the following theorem shows that by choosing a block size with similar scaling in $L$ as in \cref{eq:scaling-q},
this post-selection succeeds with probability independent of $L$.

\begin{theorem}[Constant-success probability for probabilistic correction]
\label{thm:constant-success}
When applied to the RG fixed-point state $\ket{{P^{(\infty)}}(L')}$, \cref{procedure:prob-mapping} prepares the renormalised state $\ket{P^{(q)}(L')}$.
Choosing a block size scaling as
\begin{equation}
  \label{eq:q-choice-q-2}
    q
    =
    O\!\left(
    \xi
    \log\!\left(
    \frac{\Gamma\, L}{\delta}
    \right)
    \right),
\end{equation}
with $\delta >0 $ a small constant, the post-selection succeeds
with constant probability
\begin{equation}
  \label{eq:succ-unified}
  p_{\mathrm{succ}} \;\ge\; 1 - \bigl(\delta + O(\delta^2)\bigr).
\end{equation}
\end{theorem}

\begin{proof}
By definition of the correction map~\eqref{eq:def-Cq-imp}, the state prepared upon measuring all ancillas in $\ket{0}$ is proportional to $\ket{P^{(q)}(L')}$.
The post-selection probability is
\begin{equation}
  \label{eq:p-success}
  p_{\mathrm{succ}} =  \frac{1}{\lVert C^{(q)} \rVert_{\infty}^{2L^\prime}} \bigl\lVert \ket{{P^{(q)}}(L')} \bigr\rVert^2,
\end{equation}
that we wish to bound.
From~\cref{residual,def:gamma} and our choice of block size~\eqref{eq:q-choice-q-2}, the norm of the residual scales as
\begin{equation}
\|R^{(q)}\|_{2} = O\!\left(
    \frac{\delta}{L\, \lVert (P^{(\infty)})^{-1} \rVert_{\infty}}
    \right).
\end{equation}
In turn, from~\cref{eq:def-Cq}, we see that
\begin{align}\label{eq:spect_Cq}
\begin{split}
  \lVert C^{(q)} - \mathbb{I} \rVert_{\infty}
  & \leq \lVert R^{(q)}\rVert_{2} \lVert (P^{(\infty)})^{-1} \rVert_{\infty}
  \leq \frac{\delta}{2L} \leq \frac{\delta}{2L'}.
\end{split}
\end{align}
The first inequality combines $C^{(q)}-\mathbb{I}=R^{(q)}(P^{(\infty)})^{-1}$ with submultiplicativity of the spectral norm and $\|R^{(q)}\|_{\infty}\le\|R^{(q)}\|_{2}$.
Hence, the spectral norm of the correction map is bounded through $\lVert C^{(q)} \rVert_{\infty} \leq 1 + \frac{\delta}{2L'}$, and for small $\delta \ll 1$ we get that \(\lVert C^{(q)} \rVert_{\infty}^{2L^\prime} \leq 1 + \delta + O(\delta^2)\).
Finally, using that $\|\ket{{P^{(q)}}(L^\prime)}\| = \| \ket{{A}(L)}\|$ 
as they are related by an isometry~\eqref{eq:polar} and recalling that $\|\ket{A(L)}\|$ is exponentially close to 1~\eqref{eq:norm}, we recover~\cref{eq:succ-unified} up to a correction 
exponentially small in $L$.
\end{proof}
In essence, the probability of success of each individual correction is controlled by the norm of the residual $R^{(q)}$ that decays exponentially with $q$~\eqref{residual}. When compounded over $L'=O(L)$ blocks, it is sufficient to take a block size of the order of $\log(L)$ to guarantee an overall constant probability of success.

Note that for small $\delta \ll 1$, ~\cref{eq:spect_Cq} would also hold for $(C^{(q)})^{-1}$. That is, one can implement the inverse of the correction map through block-encoding with a similar probability of success (we will repeatedly use this fact in~\cref{sec:implementation-Vq}). Finally, we highlight that
the non-normal TI case is treated in~\cref{app:non-normal}.

\subsection{All-at-once preparation of TI MPS}\label{subsec:exact-protocol}
The following protocol, sketched in~\cref{fig:exact-prep}~(c), prepares a TI MPS exactly. It is named \emph{all-at-once} (and labelled as \aao) as all the corrections are applied in parallel.

\begin{protocol}[Exact RG-based TI MPS preparation (all-at-once)]
  \label{protocol:exact-prep}
\leavevmode
\begin{enumerate}
  \item \textbf{Choose a block size.}
        Pick \(q\) as per~\cref{eq:q-choice-q-2}, and define \(L' = L/q\).
        
  \item \textbf{Classical preprocessing.}
        Perform the polar decomposition \(A^{(q)} = V^{(q)}P^{(q)}\)~\eqref{eq:polar}, and obtain the RG fixed-point operator $P^{(\infty)}$ together with the correction map \(C^{(q)} := P^{(q)}(P^{(\infty)})^{-1}\)~\eqref{eq:def-Cq}.

  \item \textbf{Prepare the RG fixed-point state.}
        Prepare \(\ket{P^{(\infty)}(L')}\) via~\cref{eq:p_infinity} (normal case). 

  \item \textbf{Probabilistic correction.}
        Apply~\cref{procedure:prob-mapping} to map \(\ket{P^{(\infty)}(L')}\) onto \(\ket{P^{(q)}(L')}\).
        If the post-selection fails, restart from Step~3.

  \item \textbf{Implement \(V^{(q)}\).}
        Apply \(V^{(q)}\) in parallel to all \(L'\) renormalised blocks.
\end{enumerate}
\end{protocol}

\paragraph*{Circuit depth}
The depth of~\cref{protocol:rg-approx-prep} provided in \cref{eq:depth-costs-bg-overall} depends on an error~$\varepsilon$.
In contrast,~\cref{protocol:exact-prep} is exact ($\varepsilon=0$) but instead may fail (Step 4).
For comparison, we consider
its \emph{expected} depth: 
\begin{equation}
  \begin{split}
  D_{\aao}(L, q)
  =
   \frac{D_{\infty} + D_{C}}{p_{\mathrm{succ}}(L, q)} + D_V(q),
  \end{split}
  \label{eq:Daao}
\end{equation}
with the first term accounting for the necessity of preparing the RG fixed point state and applying the correction maps upon each failure.

According to~\cref{thm:constant-success}, a choice of $q  =  O(\xi\log(\frac{\Gamma L}{\delta}))$ yields
$p_{\mathrm{succ}} \ge 1 - \delta - O(\delta^2)$, such that
\begin{equation}
  \label{eq:cost_aao}
   D_{\aao}(L, \delta) = \frac{D_{\infty} + D_{C}}{1 - \delta - O(\delta^2)} + 
   D_V \left(
    O\bigl(\xi\log(\tfrac{\Gamma L}{\delta})\bigr)
   \right).
\end{equation}
Recall that both $D_{\infty}$~\eqref{eq:D_inf} and $D_C$~\eqref{eq:DC} are independent of the system size $L$.
Thus, for small $\delta \ll 1$, the scaling of~\cref{eq:cost_aao} in $L$ remains dominated by $D_V(O(\xi \log(L)))$ due to the implementation of the isometries.
That is, we can achieve exact MPS preparation with similar scaling as for the approximate preparation of~\cref{sec:bg-rg}.
In addition, one can quantify fluctuations in the depth of the circuits through their variance:
\begin{remark}
The variance of the depth of~\cref{protocol:exact-prep} is
\begin{equation}
  \Var\!\left[D_{\aao}\right] = \frac{1-p_{\mathrm{succ}}}{p_{\mathrm{succ}}^{2}}\,(D_{\infty}+D_{C})^{2}.
\end{equation}
As $D_V$ dominates the expected depth~\eqref{eq:cost_aao}, 
the realised circuit depth concentrates %
around its mean.
\end{remark}

Although the success probability of~\cref{protocol:exact-prep} is already constant in $L$, in practice it can be further improved through refinements developed in the following subsection.

\subsection{Improving the success probability}\label{subsec:improve_success}

With the basis of the exact MPS preparation established, we now detail several ways to increase the success probability of~\cref{protocol:exact-prep}.
While these improvements do not change the asymptotic scalings, in practice they can result in substantial circuit depth reductions. 
In~\cref{subsec:sandc}, we present a variant of the parallel implementation of the correction maps that avoids having to re-attempt the correction of all blocks upon failure. 
In~\cref{subsec:det_corr}, we detail how \emph{deterministic} preparation is achieved through amplitude amplification.
Although not the most competitive, it shows that exact and deterministic TI MPS preparation can be achieved with the same scaling as~\cref{eq:cost_aao}.
Finally, in~\cref{subsec:optimize-parameters}, we highlight scalable numerical optimisations that can readily decrease circuit depths in concrete problems.

\subsubsection{Split-and-concatenate}\label{subsec:sandc}

\begin{figure}[t]
  \centering
  \includegraphics[width=0.5\textwidth]{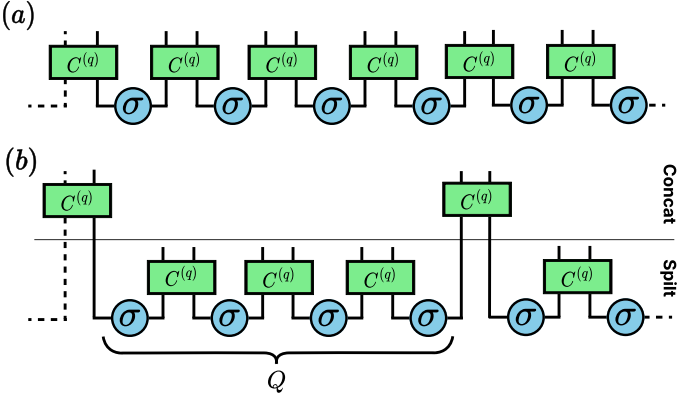}
  \caption{Preparation of the renormalised state $\ket{P^{(q)}(L')}$. a) All-at-once: all the correction maps $C^{(q)}$ are applied in parallel.  b) Split-and-concatenate: the corrections are first performed within splits (a subset of $Q$ blocks, depicted here for $Q=4$). When all splits have been corrected, a final round of correction maps between splits is performed to concatenate the $L'/Q$ splits.
  Both prepare the same state. However, upon failure of one correction map, the whole procedure needs to be repeated for a) while only the preparation of the corresponding split is required for b).
  }
  \label{fig:sac}
\end{figure}

\cref{protocol:exact-prep} can be improved:
all the corrections are applied all-at-once and upon the failure of any single correction map, the whole circuit preparation has to be repeated from scratch (Step~3) with the same probability of failing.
Instead, it is beneficial to perform rounds of correction maps on smaller and independent blocks that can later be concatenated. Such an approach, dubbed the split-and-concatenate strategy (and labelled as \snc), is analysed in this section and compared to the all-at-once approach of~\cref{subsec:exact-protocol}.

As illustrated in~\cref{fig:sac}(b), the chain of length $L'$ is split into
$L_{\rm sp}=L' / Q$ disjoint segments (called a split), each consisting of $Q$ renormalised blocks.
First, the correction maps are applied only to the blocks belonging to the same split ($Q-1$ correction maps are applied per split).
Since the splits are unentangled, a failure propagates locally and only the affected split needs to be attempted again.
Once all the splits have been successfully corrected, neighbouring splits are merged by
applying correction maps at their edges (the concatenation step), yielding the state
$\ket{{P^{(q)}}(L')}$.

\paragraph*{Circuit depth}
Before performing the concatenation step, \emph{all} of the $L_{\rm sp}$ splits must have been successfully corrected. 
Given a block size $q$ and a split size $Q$, 
the success probability of a split being corrected in a single round is  
$p_{\rm sp} = (1 / \lVert C^{(q)} \rVert_{\infty}^2)^{Q-1}$.
Let $G_{i}$ be the number of correction rounds needed for the $i$-th split to succeed for the first time. 
Then $G_{1},\dots,G_{L_{\rm sp}}$ are IID random variables following a geometric distribution with parameter $p_{\rm sp}$.
Since we can attempt to correct all splits in parallel at each round, the total number of rounds needed before concatenation is 
$
  G^{\text{max}}_{L_{\rm sp}} = \max\{G_{1},\dots,G_{L_{\rm sp}}\}.
$
The final concatenation succeeds with probability
\begin{equation}
\label{eq:p-success-cat}
    p_{\rm cat}=(1/\lVert C^{(q)}\rVert_\infty^2)^{L_{\rm sp}}.
\end{equation}
Thus, the expected depth is given by
\begin{equation}
  \label{eq:Dsc}
  \begin{split}
  D_{\snc}(q,Q)
    =
    \frac{(D_{\infty} + D_{C})\,\mathbb{E}[G^{\text{max}}_{L_{\rm sp}}] + D_{C}}{p_{\rm cat}}
    + D_V(q).
  \end{split}
\end{equation}
It accounts for both the probability of failures during the split part 
(through the expectation value $\mathbb{E}[G^{\text{max}}_{L_{\rm sp}}]$) and during the concatenation part
(through the inverse success probability $1/p_{\rm cat}$). See~\cref{app:sac} for further details.

\subsubsection{Deterministic preparation}\label{subsec:det_corr}
So far, we limited ourselves to the implementation of the correction map through
post-selection. Alternatively, one could coherently amplify the probability of preparing
the desired state by means of amplitude amplification~\cite{grover1996, brassard2002}.

Consider a unitary $U$, acting on a system $S$ and an ancillary register $A$
with $n_A$ qubits, that satisfies
\begin{equation}
  \label{eq:prepare-correct}
  \begin{split}
    U\,\ket{0}_A\ket{0}_{S}
    =\sin\theta
    \ket{0}_A\ket{\Psi}_{S}
    +
    \cos\theta
    \ket{\Psi_{\perp}},
  \end{split}
\end{equation}
with $\bra{\Psi_{\perp}} (\ket{0}_A\ket{\cdot}_S) = 0$.
Exact amplitude amplification~\cite{hoyer2000} synthesises a unitary
$\tilde{U}$ such that
$\tilde{U}\,\ket{0}_A\ket{0}_{S} = \ket{0}_A\ket{\Psi}_S$. This is achieved using $O(1/|\sin(\theta)|)$ \emph{generalised reflections} of the form $e^{-i\alpha \ket{r}\bra{r}}$, for an appropriately chosen angle $\alpha$,
around $\ket{r}=\ket{\Psi}_S$ and $\ket{r} = \ket{0}_A$.
In particular, for $\sin^2{\theta} > 1/2$, the resulting $\tilde{U}$ requires a total of three calls to $U$ or $U^\dag$: one $U$
and one $U^\dag$ to realise the generalised reflection around $\ket{\Psi}_S$ (through $U\,(\mathbb{I} \otimes e^{-i\alpha\ket{0}\bra{0}_S} )\,U^\dag$) and an additional call to $U$ for the initial state preparation. 
Further note that a reflection $e^{-i\alpha \ket{0}\bra{0}}$ over a register of $n$ qubits can be
implemented in depth 
$O(\log(n_A))$ provided another $O(n_A)$
ancillas and all-to-all circuit connectivity~\cite{claudon2024}.

In our setting, $U$ is the unitary preparing the fixed-point state followed by
the block-encoding of the $L'$ correction maps (thus $n_A = L'$),
$\ket{\Psi} = \ket{P^{(q)}(L')}$,
and the angle $\theta$ is related to the success probability through
$p_{\mathrm{succ}} = \sin(\theta)^2$.
It follows that
\begin{theorem}[Deterministic preparation of $\ket{{P^{(q)}}(L')}$]
\label{thm:det-aa}
A choice of block size
\begin{equation}
  \label{eq:q-choice-unified}
  q = O\!\bigl(\xi \log(\Gamma L)\bigr),
\end{equation}
guarantees $p_{\mathrm{succ}} > 1/2$ in~\cref{eq:p-success}.
In turn, exact amplitude amplification~\cite{hoyer2000} requires circuit depth
\begin{equation}
  3(D_{\infty} + D_{C}) +
  D_V\left(
    O\bigl(\xi\log(\Gamma L)\bigr)
  \right)
  + O(\log L),
\end{equation}
to prepare exactly $\ket{P^{(q)}(L')}$.
\end{theorem}
\begin{proof}
This follows from applying exact amplitude amplification~\cite{hoyer2000} to the unitary $U$ described above, using $\delta = 1/2$ in~\cref{thm:constant-success} to ensure $p_{\mathrm{succ}} > 1/2$. The factor of $3$ in front of $(D_\infty + D_C)$ accounts for the three calls to $U$ (or $U^\dag$) required by exact amplitude amplification discussed above.
The remaining $O(\log L)$ term accounts for the generalised reflections.
\end{proof}

That is, for the sequential implementation of $V^{(q)}$~\eqref{eq:depth-costs-bg},
one can prepare exactly the TI MPS with the same scaling as in~\cref{eq:depth-costs-bg-overall}.
Similarly to the discussion on the number of ancillary qubits in
\cref{procedure:prob-mapping}, we emphasise that the deterministic preparation of
\cref{thm:det-aa} does \emph{not} increase the total number of qubits required, as the additional ancillas are reused later on.

\subsubsection{Further numerical optimisations}\label{subsec:optimize-parameters}

The previous scaling only provides asymptotic guidance for the choice of the block size. 
In practice, for a given target MPS, we instead estimate the circuit depth for each protocol and numerically optimise over the free blocking parameters.
For the all-at-once protocol, this amounts to minimising the expected depth in \cref{eq:Daao} over $q$.
For the split-and-concatenate protocol one instead optimises \cref{eq:Dsc}, %
over the pair of parameters $(q, Q)$. Both optimisations will be used when reporting circuit depths in~\cref{sec:numerics} with additional details about the depth estimation provided in~\cref{app:depth-estimation}.

Beyond these optimisations, the gauge transformation $A\mapsto G^{-1}AG$ leaves the MPS~\eqref{eq:MPS-def} invariant but can affect the success probability of the correction step. This freedom has recently been used to reduce the gate count to implement $V^{(q)}$~\cite{scheer2025}. However, we do not optimise over the gauge in our numerical experiments of~\cref{sec:numerics}, since this freedom is expected to become most relevant for hardware-specific implementations where the circuit must be compiled to a specific gate set and connectivity.

In summary, the techniques presented in this subsection offer complementary
strategies for improving the practical efficiency of the exact preparation
protocol. The overall depth, however,
remains governed by the cost of implementing the isometry $V^{(q)}$, for which
we propose several novel constructions in the next section.

\section{Implementation of $V^{(q)}$}
\label{sec:implementation-Vq}
In the previous section, we presented exact preparation algorithms for the
renormalised state $\ket{{P^{(q)}}(L')}$.
To obtain the target state
$\ket{{A^{(q)}}(L')} = \ket{A(L)}$,
we need to implement the isometry $V^{(q)} = A^{(q)}(P^{(q)})^{-1}$, depicted as
\begin{equation}
  \label{eq:Vq}
  V^{(q)} = 
  \raisebox{-0.4\height}{\hspace{1em}\includegraphics[width=0.28\textwidth]{figs/Vq.pdf}\\[6pt] },
\end{equation}
to each renormalised block.
In this section, we present several 
circuit implementations of the isometry
$V^{(q)}$, for which the depths were reported in~\cref{tab:Vq-depth-summary}. 
In~\cref{subsec:vq_seq}, we review one approach introduced in \citet{malz2024} that will serve us as a starting point. 
In~\cref{subsec:SEQ_OAAI} and~\cref{subsec:ren-OAAI}, we present improvements in the scalings of the circuit depths with respect to the bond dimension $\chi$. As was the case in the previous section, these rely on block-encoding of the correction maps (and their inverses), and can be further refined through amplitude amplification. 
For that purpose, we introduce an exact \emph{oblivious amplitude amplification for isometries} (OAAI), stated in \cref{thm:OAAI}, which was also independently developed in~\cite{styliaris_2025}. 
Alternatively, in~\cref{subsec:quasi-prob}, we consider a quasi-probabilistic (incoherent) implementation of the correction maps, suitable when aiming at estimating expectation values of observables, that further reduces this $\chi$-dependence.

\subsection{Sequential approach}\label{subsec:vq_seq}
We begin by recalling one circuit implementation introduced
in \citet{malz2024}.
The key idea is to define a new local tensor \(
  A^{\prime} := A \otimes \mathbb{I}_{\chi}\)
and reinterpret $V^{(q)}$ given in~\cref{eq:Vq} as %
\begin{equation}
  \label{eq:Vq-chain-open}
  V^{(q)} = 
  \raisebox{-0.35\height}{\hspace{1em}\includegraphics[width=0.32\textwidth]{figs/Vq_chain.pdf}\\[6pt] }. 
\end{equation}
Since $A^{\prime}$ has effective bond dimension $\chi^{2}$, the resulting
representation can be viewed as an open-boundary MPS of length~$q$ with bond dimension
$\chi^{2}$.
By repeatedly applying singular value decompositions, this MPS can be
brought into \emph{left-canonical form}~\cite{schollwock2011}.
As a result, the isometry $V^{(q)}$ is now decomposed as %
\begin{equation}
  \label{eq:Vq-chain-cano}
  V^{(q)} = 
  \raisebox{-0.1\height}{\hspace{1em}\includegraphics[width=0.26\textwidth]{figs/Vq_chain_cano_double.pdf}\\[6pt] },
\end{equation}
where each $V_k$ is itself an isometry 
and thus can be embedded into a unitary operator.
Given that the largest isometries $V_k$ have dimension $(d\chi^{2}) \times \chi^{2}$, their unitary embedding
requires depth $O(\chi^{4})$.
In turn, $V^{(q)}$ can be implemented by \emph{sequentially} applying $q$ of such
unitaries resulting in an overall circuit depth
scaling as 
\begin{equation}\label{eq:dv_seq}
    D_{V}^{\rm seq}(q) = O(q \chi^{4}),
\end{equation}
that was reported earlier in~\cref{eq:depth-costs-bg}.
While straightforward, this construction is relatively costly in $q$ due to the $\chi^4$ prefactor. In the following, we present novel constructions improving upon this scaling.

\subsection{Sequential approach with oblivious amplitude amplification}
\label{subsec:SEQ_OAAI}
The reason why the sequential approach incurs a depth scaling as
$O(q\chi^{4})$ is that the inverse operator $\bigl(P^{(q)}\bigr)^{-1}$
acts on both ends of the chain of tensors $A$ in~\cref{eq:Vq},
which requires treating the chain as an MPS with an effective bond dimension~$\chi^{2}$ rather than $\chi$.
To circumvent this issue, we exploit the decomposition
\begin{equation}
  \label{eq:Pq-inv-decomp}
  \begin{split}
  (P^{(q)})^{-1}
  &=
  (P^{(\infty)})^{-1} (C^{(q)})^{-1} \\
  &= 
  \left( \sigma^{-1}\otimes \mathbb{I}_{\chi}\right) (C^{(q)})^{-1},
  \end{split}
\end{equation}
where we used $P^{(\infty)} = \sigma \otimes \mathbb{I}_{\chi}$ with $\sigma \succ 0$~\eqref{eq:pinf_product} for normal MPS (see~\cref{app:non-normal} for the non-normal case).
Using this factorisation, the isometry becomes

\begin{equation}
  \label{eq:Vq_decomp_chain}
  V^{(q)} = 
  \raisebox{-0.4\height}{\hspace{1em}\includegraphics[width=0.3\textwidth]{figs/Vq_decomp_chain.pdf}\\[6pt] }.
\end{equation}
Let us block the first $q_h$ sites on the left
and decompose the resulting tensor $A^{(q_h)}$
using~\cref{eq:polar,eq:def-Cq-imp}. Noting that the 
$\sigma$ appearing in this decomposition cancels out $\sigma^{-1}$, we obtain the isometry recast in a form amenable to implementation:
\begin{equation}
  \label{eq:Vq_decomp_chain_2}
  V^{(q)} = 
  \raisebox{-0.4\height}{\hspace{1em}\includegraphics[width=0.3\textwidth]{figs/Vq_decomp_chain2.pdf}\\[6pt] }.
\end{equation}

First, we apply the inverse correction map $(C^{(q)})^{-1}$.
Then, from its right virtual leg, we construct an MPS of length $(q-q_h)$ and bond dimension $\chi$. Finally, we implement the correction map $C^{(q_h)}$, followed by the isometry $V^{(q_h)}$.
Both the inverse correction map $(C^{(q)})^{-1}$ and the correction map $C^{(q_h)}$ are block-encoded (\cref{subsec:prob-correction}), which requires a circuit depth $O(\chi^4)$~\eqref{eq:DC}. 
Construction of the MPS is carried out in circuit depth $O((q-q_h)\chi^{2})$,
while $V^{(q_h)}$ is realised in circuit depth $O(q_h\chi^{4})$ using the sequential approach of~\cref{subsec:vq_seq}.
Hence, the circuit depth entailed by~\cref{eq:Vq_decomp_chain_2} scales as
\begin{equation}\label{eq:vq_seq_oaai_partial}
  O(q\chi^{2}) + O(q_h\chi^{4}).
\end{equation}
However, this does not account for the potential failure of the post-selection of the block-encodings $(C^{(q)})^{-1}$ and $C^{(q_h)}$.
We now show that the post-selection can be made deterministic, for each individual implementation of $V^{(q)}$, retaining the scaling of~\cref{eq:vq_seq_oaai_partial} while simultaneously taking $q_h$ to be independent of $q$.

Central to the deterministic post-selection is an exact oblivious amplitude amplification for isometries (\emph{exact OAAI}, proven in~\cref{app:proofs_OAAI}) that generalises results of Ref.~\cite{Berry2014} to isometries. A similar construction was independently introduced in~\cite{styliaris_2025} in a different context.
\begin{lemma}[Exact oblivious amplitude amplification for isometries (exact OAAI)]
\label{thm:OAAI}
Let $U$ be a unitary and $V$ be an isometry ($V^\dag V=\mathbb{I}$), both acting on a system $S$ and ancillary register $A$, the latter with $n_A$ qubits, such that 
\begin{equation}
  \label{eq:OAAI-assump}
U\bigl(\ket{0}_{A}\ket{\Psi}_{S}\bigr)
=
\sin(\theta)\ket{0}_{A'}V\ket{\Psi}_{S}
+\cos(\theta)\ket{\Psi_\perp}
\end{equation}
holds for any $\ket{\Psi}_S$.
Here, $(\bra{0}_{A'}\otimes I_{S})\ket{\Psi_\perp}=0$ and $A'$ is the output ancillary register with $n_{A'} \leq n_{A}$.%
Then, one can synthesise a unitary $\tilde{U}$ such that
\begin{equation}
\tilde{U}\bigl(\ket{0}_{A}\ket{\psi}_{S}\bigr)
=
\ket{0}_{A'}V\ket{\psi}_{S},
\end{equation}
using $r=O(1/|\sin(\theta)|)$ amplification rounds, each consisting of one call to $U^\dag$, two generalised reflections around $\ket{0}\bra{0}_A$ and $\ket{0}\bra{0}_{A'}$, and one additional call to $U$.
Including the initial application of $U$, this gives $2r+1$ calls to $U$ or $U^\dag$ and $2r$ generalised reflections.
Whenever $\sin^2{\theta}\ge 1/2$, one round is sufficient.
\end{lemma}
By employing exact OAAI, the success probability of implementing $V^{(q)}$ can be boosted to unity without any prior knowledge of (i.e. oblivious to) the input state on which $V^{(q)}$ acts. 
We stress that exact OAAI is applicable only when the target operation is an isometry:
While it cannot be applied to the individual block-encoded corrections, it can be applied to the full circuit implementing $V^{(q)}$ that encompasses the block-encoding of both $(C^{(q)})^{-1}$ and $C^{(q_h)}$. 
With this result, it remains to identify adequate choices of the blocking sizes $q$ and $q_h$.

Recall
that the success probability of a block-encoded correction map scales as the inverse of its squared spectral norm. With two corrections involved, we therefore need to bound $\|C^{(q_h)}\|_{\infty}\,\|(C^{(q)})^{-1}\|_{\infty}$. For $C^{(q_h)}$ we have
\begin{align}
\begin{split}
  \lVert C^{(q_h)} - \mathbb{I} \rVert_{\infty}
  &\leq \lVert R^{(q_h)}\rVert_{2} \lVert (P^{(\infty)})^{-1} \rVert_{\infty} \\
  &\leq \Gamma\,\lVert (P^{(\infty)})^{-1} \rVert_{\infty}\, e^{-q_h/\xi},
\end{split}
\end{align}
akin to~\cref{eq:spect_Cq}.
It follows that a choice of
\begin{equation}\label{eq:qh-choice}
q_h=\Omega(\xi\log(\Gamma/\tilde \delta)),
\end{equation}
for a small constant $\tilde \delta$, yields $1/\| C^{(q_h)}\|^2_{\infty}\geq 1-\tilde \delta - O(\tilde \delta^2)$. Given that $q \gg q_h$, we do not need to worry about the contribution of $\|(C^{(q)})^{-1}\|_{\infty}$ in our scaling analysis, such that~\cref{eq:qh-choice} is sufficient to guarantee an overall probability of post-selection greater than a half. We stress that~\cref{eq:qh-choice} does not scale with $L$, in contrast to~\cref{eq:q-choice-q-2}, as we only need to account for the post-selection of a single block-encoding.
This results in the following theorem.
\begin{theorem}[Sequential implementation of $V^{(q)}$ with exact OAAI]
\label{thm:Vq-OAAI}
A choice of $q$ as per~\cref{eq:q-choice-unified} and of $q_h$ as per~\cref{eq:qh-choice} 
guarantees the implementation of the isometry $V^{(q)}$ through~\eqref{eq:Vq_decomp_chain_2} with success probability at least $1/2$.
By applying exact OAAI (\cref{thm:OAAI}), the implementation of $V^{(q)}$ can be made fully deterministic.
As per~\cref{eq:vq_seq_oaai_partial}, the resulting circuit has depth
\begin{equation}
  \label{eq:depth-Vq-OAAI-seq}
  D_V^{\seqAA}= O(q\,\chi^{2}) + O(\chi^{4}\xi\log(\Gamma)),
\end{equation}
where the first contribution corresponds to the construction of the $(q-q_h)$-site MPS with bond dimension~$\chi$,
and the second to the sequential implementation of $V^{(q_h)}$.
\end{theorem}
Compared to~\cref{eq:dv_seq}, the $\chi$-dependence associated to $q$ has been quadratically improved from $\chi^{4}$ to $\chi^{2}$, at the price of an extra $q$-independent $\chi^4$ term.

\subsection{Renormalisation approach with oblivious amplitude amplification}~\label{subsec:ren-OAAI}
We can further improve the previous scaling~\eqref{eq:depth-Vq-OAAI-seq}, by repeating the nested steps of blocking and decomposition, in a renormalisation-based fashion.
Concretely, we introduce a block size $q_2 < q$ and block the $q$ tensors $A$, in~\cref{eq:Vq_decomp_chain}, into $q/q_2$ tensors $A^{(q_2)}$. Further decomposing these blocks,~\cref{eq:Vq_decomp_chain} takes the form
\begin{equation}
  \label{eq:Vq_rg}
  V^{(q)} = 
  \raisebox{-0.4\height}{\hspace{1em}\includegraphics[width=0.22\textwidth]{figs/V_q_RG.pdf}\\[6pt] },
\end{equation}
which can be used for implementation. 
After implementation of $(C^{(q)})^{-1}$ and preparation of the product state $\ket{\sigma}^{\otimes (q/q_2 -1)}$, the additional correction maps $C^{(q_2)}$ followed by the isometries $V^{(q_2)}$ are applied (both in parallel).
This results in a circuit depth scaling as
\begin{equation}\label{eq:vq_ren_oaai_partial}
  O(\chi^{4}) + D_V(q_2),
\end{equation}
with $O(\chi^4)$ contributions from the block-encoding implementations of the $1+q/q_2$ correction maps and the depth $D_V(q_2)$ for implementing the isometries given a block size $q_2$ that is still to be fixed.

As before, the choice of $q_2$ is taken to guarantee an 
overall probability of success 
greater than a half. 
In particular, a choice of 
\begin{equation}\label{eq:q2-choice}
    q_2 = \Omega\!\bigl(\xi \log(\Gamma q / \delta)\bigr),
\end{equation}
ensures that this probability is bounded through
\begin{equation}
\begin{split}
    p_{\mathrm{succ}} &= \|C^{(q_2)}\|_{\infty}^{-2q/q_2} 
    \|(C^{(q)})^{-1}\|_{\infty}^{-2} \\
    &\geq 1 - \delta - O(\delta^2).
\end{split}
\end{equation}
This results in the following theorem.
\begin{theorem}[Renormalised implementation of $V^{(q)}$ with exact OAAI]
\label{thm:Vq-RG-OAAI}
A choice of $q$ as per~\cref{eq:q-choice-unified} and of $q_2$ as per~\cref{eq:q2-choice} 
guarantees the implementation of the isometry $V^{(q)}$ through~\cref{eq:Vq_rg} with success probability at least $1/2$.
By applying exact OAAI (\cref{thm:OAAI}), the implementation of $V^{(q)}$ can be made fully deterministic.
As per~\cref{eq:vq_ren_oaai_partial}, and performing the isometries $V^{(q_2)}$ through~\cref{thm:Vq-OAAI}, the resulting circuit has depth
\begin{equation}
  \label{eq:depth-Vq-OAAI-ren}
  D_V^{\renAA}= O(\chi^{2} \xi \log(\Gamma q)) + O(\chi^{4} \xi\log(\Gamma)).
\end{equation}
\end{theorem}

Exact OAAI (\cref{thm:OAAI}) requires $q/q_2$ ancillary qubits. However, noting that these ancillas can be reused during the implementation of the $V^{(q_2)}$, and following the argument made in~\cref{subsec:prob-correction}, we see that this does not increase the overall qubit requirement.
Compared to the sequential depth~\cref{eq:depth-Vq-OAAI-seq}, the dependence on $q$ is reduced from linear to logarithmic while retaining the $\chi^2$ prefactor.

\subsection{Quasi-probabilistic implementation approach}
\label{subsec:quasi-prob}
In the constructions presented so far, the implementation of the unitary matrices that block-encode the correction maps (or their inverses) requires circuit depths scaling as $O(\chi^{4})$~\eqref{eq:DC}. In terms of the bond dimension, these remain the largest contributions to our circuit depths.
In this section we show that if we allow incoherent preparation, this scaling can be improved through the use of \emph{quasi-probabilistic decomposition} (QPD).

Generally speaking, QPD techniques implement a superoperator (possibly non-physical, as is the case in quantum error
mitigation~\cite{temme2017,endo2018,cai2023,dai_koczor_2026}) through a probabilistic mixture of physical channels and re-weighting, at the cost of an increased sampling overhead.
For our purposes, the correction superoperator
$\mathcal{C}(\cdot) = C(\cdot)\,C^{\dag}$ is decomposed as
\begin{equation}\label{eq:qdp}
\mathcal{C} = \sum_i \gamma_i \mathcal{B}_i
\end{equation}
with $\gamma_i \in \mathbb{R}$ and physical channels $\mathcal{B}_i$ implementable with $O(1)$ circuit depth (we take tensor products of single-qubit operations~\cite{temme2017,endo2018}).
Let us define the $\ell^1$-norm $|\gamma|_1 = \sum_i |\gamma_i|$.
When measuring expectation values, a QPD implementation of $\mathcal{C}$ consists of applying one of the channels $\mathcal{B}_i$, sampled with probability $|\gamma_i|/|\gamma|_1$, and rescaling the measurement outcome by a factor ${\rm sign}(\gamma_i)|\gamma|_1$.
The latter incurs an increase in the variance, and thus in the required number of samples, by a factor $|\gamma|_1^2$ called the \emph{sampling overhead}. When $Q$ occurrences of the corrections are implemented through QPD, this overhead becomes $|\gamma|_1^{2Q}$
(see~\cref{app:incoherent} for further details about QPD).

Given that the distance of $C^{(q)}$ to the identity operator decays exponentially with $q$,
as we see in~\cref{residual,eq:spect_Cq}, 
we expect a similar scaling for the sampling overhead.
This is captured by the following theorem (proven in~\cref{app:incoherent}).
\begin{theorem}[Quasi-probabilistic decomposition of $C^{(q)}$]
\label{thm:Cq-qpd}
Let $C^{(q)}$ be the correction map for a block size~$q$~\eqref{eq:def-Cq-imp}, and $\mathcal{C}(\cdot) = C^{(q)}(\cdot)\,{C^{(q)}}^{\dag}$ the corresponding superoperator.
The QPD of $\mathcal{C}$ in a basis of
quantum channels implementable in $O(1)$ depth~\cite{temme2017,endo2018}, has an $\ell^1$-norm bounded through
\begin{equation}\label{eq:bound_l1}
  |\gamma|_1 = 1 + O\!\bigl(\Gamma\,\mathrm{poly}(\chi)\, e^{-q/\xi}\bigr).
\end{equation}
\end{theorem}
That is, one can substitute each $O(\chi^4)$-depth circuit implementing the block-encodings of the correction maps, with randomly chosen constant-depth circuits, at the cost of an increased number of shots bounded through~\cref{eq:bound_l1}. Then, for a given number of QPDs to be employed per circuit, one can identify a block size $q$ that ensures constant sampling overhead. 

First, consider the circuit preparation of the renormalised state $\ket{P^{(q)}(L')}$ through~\cref{procedure:prob-mapping}.
Since we have a number $L'=O(L)$ of maps $C^{(q)}$ to be replaced through QPD, choosing
$q = O\!\left(\xi \log(\chi \Gamma L)\right)$
is enough to ensure $O(1)$ sampling overhead, and to decrease the circuit depth from $O(\chi^4)$ to $O(\chi^2)$, with the preparation of the fixed-point state~\eqref{eq:p_infinity} now becoming the dominant contribution.

Second, and as detailed in~\cref{app:qpd_vq}, the isometry $V^{(q)}$ in~\cref{eq:Vq_decomp_chain} admits an alternative implementation that preserves the bond dimension $\chi$ of the underlying MPS rather than the $\chi^2$ of~\cref{eq:Vq-chain-open}, via a Bell-state measurement.
While the latter introduces a post-selection probability of $1/\chi^2$, it can be made deterministic using exact OAAI (\cref{thm:OAAI}) at the cost of $O(\chi)$ calls to the corresponding unitary or its dagger. Considering the $L'=O(L)$ isometries to be implemented, this yields $O(\chi L)$ correction maps to be replaced by QPD.
Thus, a block size
\begin{equation}\label{eq:scaling-q-inc}
q = O\!\left(\xi \log(\chi \Gamma L)\right)
\end{equation}
ensures constant sampling overhead, and a circuit depth
\begin{equation}
    D^{\rm QPD}_V = O(q \chi^3).
\end{equation}

To conclude, in this section we have presented several implementations for the isometry $V^{(q)}$ 
improving on its depth scaling with respect to the bond dimension. These improvements are summarised in \cref{tab:Vq-depth-summary}.

\section{Non-translationally invariant MPS}
\label{sec:nonTI}
So far, we have limited ourselves to the preparation of TI MPS and broken down their preparations into two subroutines: an exact preparation of the renormalised state $\ket{P^{(q)}(L')}$ (\cref{sec:exact-prep}), and the implementation of the isometries $V^{(q)}$ (\cref{sec:implementation-Vq}).
In this section, we detail how both are generalised to non-TI MPS.
For the sake of exposition, we focus on the all-at-once preparation (\cref{subsec:exact-protocol}) and the sequential implementation of the isometry (\cref{subsec:vq_seq}), and leave details of the other constructions to the appendices (\cref{app:nonTI-other-Vq}).

For notational simplicity we consider the case of a non-TI MPS with open boundary conditions and uniform bond dimension $\chi$.
Such an MPS is defined as
\begin{equation}
  \ket{\tilde A(L)}
  =
  \raisebox{-0.25\height}{\hspace{1em}\includegraphics[width=0.3\textwidth]{figs/NTIMPS.pdf}\hspace{3pt},  \\[6pt] }
  \label{eq:nonTI_MPS_def}
\end{equation}
and is specified by a collection of site-dependent tensors
$\{ A^{[k]}\}_{k=1}^{L}$ with $A^{[k]}\in\mathbb{C}^{\chi\times d\times \chi}$.
Throughout, we use the same canonical convention as for the TI case: For all
$k$,
\begin{equation}
  \sum_{s=1}^{d} A^{[k],s\dag} A^{[k],s}=\mathbb{I}_{\chi}.
  \label{eq:left_canonical_nonTI}
\end{equation}

Unlike the TI case, we allow non-uniform blocking, with
$\vec q := (q_1,\ldots,q_{L'}) $
a list of $L'$ block sizes satisfying $\sum_{\ell=1}^{L'} q_{\ell} = L$.
For each block labelled by $\ell$, let $A_{\ell}^{(q_{\ell})}$ be the
corresponding blocked tensor~\eqref{eq:blocking}, with polar decomposition~\eqref{eq:polar}
\begin{equation}
    A_{\ell}^{(q_{\ell})} = V_{\ell}^{(q_{\ell})}\; P_{\ell}^{(q_{\ell})}.
  \label{eq:polar_nonTI_main}
\end{equation}
Let us further take $\sigma_{\ell} \succ 0$ to be some positive-definite matrix associated with the block $\ell$ 
and define the block-dependent \emph{fixed-point operator} to be
\begin{equation}
  P_{\ell}^{(\infty)} := \sigma_{\ell} \otimes \mathbb{I}_{\chi}.
  \label{eq:pinf_nonTI}
\end{equation}

To port results from the TI setting we wish to identify a $\sigma_\ell$, for each block, ensuring that $P_{\ell}^{(q_{\ell})}$ exponentially converges to $P_{\ell}^{(\infty)}$ as the block size $q_\ell$ increases, similar to~\cref{residual}.
With no particular structure, we do not have guarantees of their existence.
However when the tensors $A^{[k]}$ are IID we can leverage the theory of ergodic sequences of the associated channels~\cite{Movassagh}, to show (see~\cref{thm:ergodic-blockwise-fixedpoint}) that there exists a $\sigma_\ell$ such that
\begin{equation}
  \label{eq:Rq-upper-nonti_blocks_main}
  \|P^{(q_\ell)}_{\ell} - P^{(\infty)}_\ell\|_2 =: \|R^{(q_\ell)}_{\ell}\|_2
  \leq \Gamma_\ell \,e^{-q_\ell/\tilde{\xi}},
\end{equation}
for some constants (independent of $L$) $\Gamma_\ell$ and $\tilde{\xi}$, the latter also being independent of the block index $\ell$.
These results guarantee the existence of $\sigma_\ell$, $\Gamma_\ell$ and $\tilde{\xi}$ such that~\cref{eq:Rq-upper-nonti_blocks_main} holds, but in general, these quantities cannot be evaluated directly. 

With the existence of adequate fixed-point operators, we can define the non-TI RG fixed-point state as
\begin{equation}
  \ket{\tilde P^{(\infty)}(L')} := \bigotimes_{\ell=1}^{L'} \ket{\sigma_{\ell}},
  \label{eq:p_infinity_nonTI}
\end{equation}    
and the correction maps for each block through
\begin{equation}
  C_{\ell}^{(q_{\ell})} := P_{\ell}^{(q_{\ell})}(P_{\ell}^{(\infty)})^{-1}.
  \label{eq:def-Cq-nonTI}
\end{equation}
As before, each correction map is implemented via block-encoding, introducing a probability of failure.
Unlike the TI case, where a uniform block size $q = O(\xi\log L)$ suffices to achieve constant success probability~(\cref{subsec:exact-protocol}), the correction maps $ C_{\ell}^{(q_{\ell})}$, and the corresponding probabilities of success, now may vary from block to block, so one must choose block sizes that keep all local correction map norms (and thus the overall probability of failure) controlled simultaneously.
To account for this, define the block-size dependent local errors
\begin{equation}
  \delta_{\ell} := 2\bigl(\|C_{\ell}^{(q_{\ell})}\|_{\infty} - 1\bigr),
  \label{eq:local_budget_nonTI}
\end{equation}
together with the constraint that
\begin{equation}
  \sum_{\ell=1}^{L'} \delta_{\ell} \le \delta
  \label{eq:global_budget_nonTI}
\end{equation}
holds for some global error $\delta$. 
By the same argument as in \cref{subsec:exact-protocol} (and detailed in~\cref{proof_th6_app}), 
it follows from~\cref{eq:global_budget_nonTI} that the probability that all the corrections succeed simultaneously 
can be bounded by $p_{\mathrm{succ}}(L,\vec{q}) \ge 1 - \delta - O(\delta^2)$.
After successful correction, we apply the isometries $V_{\ell}^{(q_{\ell})}$
in parallel to each block $\ell$ to prepare the desired state $\ket{\tilde A(L)}$. This is summarised in the following.

\begin{protocol}[Exact RG-based non-TI MPS preparation (all-at-once)]
\label{protocol:nonTI-exact}
\leavevmode
\begin{enumerate}
\item \textbf{Choose a blocking.}
Choose $\vec q=(q_1,\ldots,q_{L'})$ so that~\cref{eq:global_budget_nonTI} is satisfied.

\item \textbf{Block and polar-decompose.}
For each block $\ell$, form the blocked tensor $ A_{\ell}^{(q_{\ell})}$, compute its
polar decomposition
$ A_{\ell}^{(q_{\ell})}=V_{\ell}^{(q_{\ell})}P_{\ell}^{(q_{\ell})}$.
Then, obtain the RG fixed-point operator $P_{\ell}^{(\infty)}$~\eqref{eq:pinf_nonTI} together with the correction map
$C_{\ell}^{(q_{\ell})} = P_{\ell}^{(q_{\ell})}(P_{\ell}^{(\infty)})^{-1}$~\eqref{eq:def-Cq-nonTI}.

\item \textbf{Prepare the fixed-point state.}
Prepare $\ket{\tilde P^{(\infty)}(L')}$~\eqref{eq:p_infinity_nonTI}, the product state formed by the block-dependent local states $\sigma_{\ell}$.

\item \textbf{Apply the correction maps.}
Apply the correction map $ C_{\ell}^{(q_{\ell})}$ on every block in parallel via~\cref{procedure:prob-mapping}.
If the correction fails, restart from Step~3.

\item \textbf{Apply the block isometries.}
Apply $ V_{\ell}^{(q_{\ell})}$ in parallel for all $\ell$.
\end{enumerate}
\end{protocol}

\paragraph*{Circuit depth}
With the sequential implementation, 
and since all isometries are applied in parallel, the isometry depth
is given by $D_V^{\mathrm{seq}}(q_{\max}) = O(\chi^4 q_{\max})$ and depends on the largest block size $q_{\max} := \max_{\ell}(q_{\ell})$.
Hence, for
the all-at-once construction, the overall depth of~\cref{eq:Daao} becomes in the non-TI setting
\begin{equation}
  \tilde{D}_{\aao}(L, \vec{q})
  =
   \frac{D_{\infty} + D_{C}}{p_{\mathrm{succ}}(L, \vec{q})} + D_V^{\rm seq}(q_{\max}).
  \label{eq:cost_nonTI_main}
\end{equation}
Thus, exact preparation of the MPS in depth scaling as $O(\log L)$ is possible provided that one can choose the block sizes
such that $q_{\max}=O(\log L)$ while keeping constant probability of success, $p_{\mathrm{succ}}(L, \vec{q})=O(1)$, through~\cref{eq:global_budget_nonTI}.

As before, we cannot show it for an arbitrary non-TI state, but provided that the random tensors $A^{[k]}$ are IID, we can prove that this is possible. This is captured by the following theorem (see~\cref{app:random-nonTI} for the proof).

\begin{theorem}[IID random non-TI MPS]
\label{thm:random-nonTI}
Assume that the random tensors $\{ A^{[k]}\}_{k=1}^{L}$ in~\cref{eq:nonTI_MPS_def} are IID and satisfy the non-negligible-decoherence condition of Ref.~\cite{Movassagh}.
Then there exists a
constant $\tilde{\xi}>0$, depending only on the underlying distribution, such that for small $\delta$ one can choose a uniform block size $q_\ell = q$ for all $\ell$
in \cref{protocol:nonTI-exact} so that
\begin{equation}
  q=O\!\left(\tilde \xi \log\!\frac{L}{\delta}\right),
  \label{eq:qmax_random_main_v2}
\end{equation}
and the success probability of the correction step obeys
\begin{equation}
  p_{\mathrm{succ}} \ge 1-\delta-O(\delta^2)
  \label{eq:psucc_random_main_v2}
\end{equation}
almost surely over the IID draw.
\end{theorem}

This shows that the exact preparation framework can be extended to IID non-TI MPS,
with the guarantee that a block size growing as
$O(\log L)$ suffices to maintain a constant overall success probability.
That is, for these non-TI MPS, similar results and scalings as for the TI case can be achieved.
In practice, at a given choice of the uniform block-size $q$, any unknown fixed-point $\sigma_\ell$ can be approximated numerically from the leading right eigenvector of the corresponding block transfer matrix (see~\cref{rem:practical-sigma} of the appendices) such that the corresponding correction maps and, ultimately, the circuit depths~\eqref{eq:cost_nonTI_main} can be evaluated. In turn, the choice of an adequate uniform block size $q$ is identified through optimisation.
In the next section, we use numerical simulations to validate that even with these approximations, and even in regimes not covered by our theory, the scalings of~\cref{thm:random-nonTI} apply. 
We further note that while focusing on the IID case, the results presented here could extend to more general ergodic sequences~\cite{Movassagh} of non-TI tensors.

As a final remark, the split-and-concatenate protocol of~\cref{subsec:sandc} can be particularly advantageous in this non-TI setting. Some local tensors may have larger residual prefactors $\Gamma_\ell$~\eqref{eq:Rq-upper-nonti_blocks_main} or larger norms of the correction maps~\eqref{eq:def-Cq-nonTI}, both contributing to larger circuit depths. By introducing splits that isolate neighbourhoods around such ill-conditioned tensors, one can prevent their adverse properties from propagating and reduce the circuit depth.

\section{Numerical results}\label{sec:numerics}
The preceding sections show that MPS preparation can be achieved exactly and with improved scalings with respect to the bond dimension, compared to existing methods.
However, these scalings do not account for the different overheads involved. They also leave open two practical questions: which construction is preferable, and how much, if any, our methods improve in a specific preparation task. Concretely, this means preparing an MPS at a given system size, bond dimension and required precision.
Furthermore, the applicability of the non-TI extension to physically relevant states beyond the regime covered by~\cref{thm:random-nonTI} remains to be assessed numerically.
The aim of this section is to address these questions through concrete numerical examples.

In~\cref{subsec:meth} we start by reviewing the methodology adopted in our studies.
In~\cref{sec:vq-selection}, we compare different implementations of the isometries, in particular
contrasting constructions with depths scaling as $O( \log (L))$ and $O(\log \log (L))$. 
This shows that the former is preferable in the settings of interest, namely bond dimension $\chi>2$ and not exceedingly large system sizes, despite its worse asymptotic scaling.
In~\cref{sec:compare-approx-rg}, we proceed by comparing different methods for the exact preparation of the renormalised state, showcasing the improvements brought by our techniques.
This confirms the benefits of the split-and-concatenate layout of~\cref{subsec:sandc}. It further shows that exact preparation is competitive with, or even more efficient than, previous approximate MPS preparations, except when the approximation is crude, with error $\varepsilon\geq 10\%$.
Finally, in~\cref{sec:nonTI_numerics} %
we confirm applicability of our constructions to non-TI MPS.
\subsection{Methodology}\label{subsec:meth}
In what follows, we compare different MPS preparations in terms of their optimised circuit depths for different ensembles of MPS (either TI or non-TI, either randomly generated or taken as ground states of physical models). Here, we detail how such comparisons are carried out.

\paragraph{Methods considered}
For the exact preparation of the renormalised state (\cref{sec:exact-prep}) we focus on the all-at-once (\aao of~\cref{subsec:exact-protocol}) and the split-and-concatenate (\snc of~\cref{subsec:sandc}) constructions.
For the implementations of the isometry $V^{(q)}$ (\cref{sec:implementation-Vq}), we focus on the sequential implementation (\emph{Seq} of~\cref{subsec:vq_seq}), the sequential implementation with exact OAAI (\seqAA of~\cref{subsec:SEQ_OAAI}), and the renormalised implementation with OAAI (\renAA of~\cref{subsec:ren-OAAI}).
We leave aside the quasi-probabilistic implementation of~\cref{subsec:quasi-prob}, as it targets incoherent estimation of expectation values rather than coherent state preparation, and the deterministic correction strategy of~\cref{subsec:det_corr}, whose additional overhead does not improve the asymptotic depth scaling.
Finally, when comparing our methods to existing competitive MPS preparations, we consider as a benchmark the RG-based approximate preparation algorithm of \citet{malz2024} detailed in~\cref{sec:bg-rg}.
We note that this RG-based approximate implementation is the one employed in the recent experimental demonstration of Ref.~\cite{scheer2025}, 
making it a natural baseline.

\paragraph{Circuit depth estimation}
All the circuit-depth scalings reported so far were based on the fact that any $M\times N$ isometry can be realised as a circuit with depth scaling as $O(M\times N)$.
Although an exact figure will depend on the details of the isometry to be implemented, and on the underlying quantum platform 
(e.g. its connectivity and primitive gate set),
we will use as a proxy the CNOT-depth lower bound of Ref.~\cite{iten2016}.
When the isometry is defined from a register of $n$ qubits ($n = \log_2(N)$) to a register of $m$ qubits ($m = \log_2(M)$), the CNOT-depth is given by:
\begin{equation}
  N_{\mathrm{iso}}(m,n)
  =
  \tfrac{1}{4}\Bigl(2^{\,n+m+1} - 2^{\,2m} - 2n - m - 1\Bigr).
  \label{eq:Niso-bound}
\end{equation}
With this, we can estimate the different circuit depths at play
(see~\cref{app:depth-estimation} for detail).

\paragraph{Circuit depth optimisation}
As previously discussed in~\cref{subsec:optimize-parameters}, some of our constructions exhibit parameters that can be optimised given a specific instance of an MPS to be prepared. 
We systematically proceed to such optimisations and consistently report \emph{optimised depths}.
In particular, the block size $q$ that appears in any of the circuit constructions discussed can be chosen to minimise their depths. For the implementation of \citet{malz2024}, and given a target error $\varepsilon$, we choose the smallest block size that achieves $\varepsilon$-error~\eqref{eq:trace-distance-error}.
For the \aao protocol, we instead optimise $q$ based on~\cref{eq:Daao}.
For the \snc protocol we optimise both $q$ and the split size $Q$ based on~\cref{eq:Dsc}.
(Representative optimisations for $q$ and $Q$ are reported in~\cref{app:optimisation-plots}.)

\subsection{Comparison of linear and logarithmic OAAI implementations}
\label{sec:vq-selection}
As a preliminary study, we compare two implementations of $V^{(q)}$
introduced in \cref{sec:implementation-Vq}, namely the \seqAA
construction of \cref{subsec:SEQ_OAAI} and the \renAA construction of
\cref{thm:Vq-RG-OAAI}. While \renAA has a better asymptotic
dependence on $L$, namely $O(\log \log L)$ scaling~\eqref{eq:depth-Vq-OAAI-ren} rather than $O(\log L)$ scaling~\eqref{eq:depth-Vq-OAAI-seq}, it also carries a larger constant overhead. This comparison aims at determining which implementation is favourable at varied system sizes ($L \in [10^2, 10^6]$) and bond dimensions ($\chi \in \{2, 4, 8\}$).

For each setting, we generate $64$ independent Haar-random TI MPS~\cite{garnerone2010typicality} with
correlation length $\xi \approx 3.3$. For each MPS instance, we estimate the circuit depths for \aao preparation of the RG fixed-point state followed by either the \seqAA~\eqref{eq:depth-Vq-OAAI-seq} or \renAA~\eqref{eq:depth-Vq-OAAI-ren} circuit constructions for the isometries $V^{(q)}$ and compute their ratios. A ratio smaller than one indicates that the former is preferable to the latter.
Averages of these ratios, over the random MPS instances, are reported in~\cref{fig:d_vs_L_oaa_linear_log} together with their standard deviations.

As can be seen, except for the smallest bond dimension $\chi=2$, 
in most of the settings considered 
the constant-factor overhead of \renAA dominates, 
making it less competitive than \seqAA.
Although \renAA is preferable asymptotically, 
the system size at which \renAA becomes favourable
roughly scales as $\Omega(\exp(\chi^2))$ and 
is far beyond the regime studied here
(e.g. for $\chi=8$ this would occur around $L \approx 10^{34}$).
Given our focus on larger bond dimensions, 
and based on the random MPS instances studied here, 
we conclude that implementations of the isometry 
that only scale in $\log L$ (rather than $\log \log L$) are preferable.
As such, in the subsequent studies, 
we will only consider such implementations.

\begin{figure}[htp]
  \centering
  \includegraphics[width=0.45\textwidth]{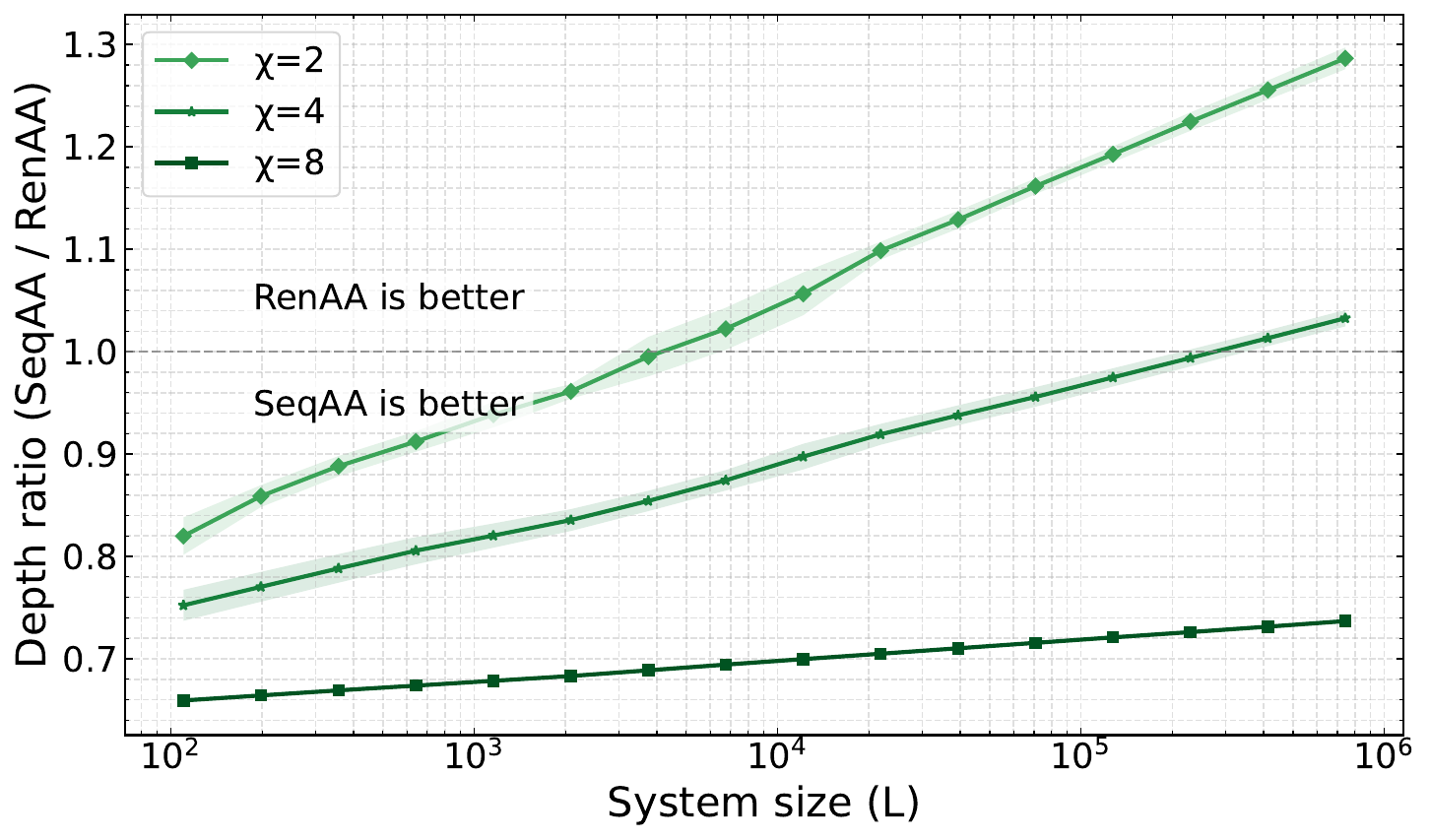}
  \caption{
  Ratio of estimated circuit depths between the \seqAA and \renAA
  implementations of $V^{(q)}$, shown as a function of
  the system size $L$ (log scale) for bond dimensions $\chi\in\{2,4,8\}$.
  Each data point shows the mean $\pm$ one standard deviation over $N_{\mathrm{MPS}}=64$
  Haar-random TI MPS instances with correlation length $\xi \approx 3.3$.
  Extrapolating from the data, the crossover at which \renAA becomes
  favourable for $\chi=8$ is estimated to occur at around $L \sim 10^{34}$.
  }
  \label{fig:d_vs_L_oaa_linear_log}
\end{figure}

\subsection{Comparison with the approximate RG method}
\label{sec:compare-approx-rg}
Next, we compare three ways of preparing the renormalised state $\ket{P^{(q)}(L')}$: the approximate RG protocol of
\citet{malz2024}, 
and our exact preparations including the \aao and \snc protocol. 
The purpose of this benchmark is twofold. First, we wish to probe the benefits of the \snc approach compared to the \aao (see~\cref{fig:sac}). 
Second, we wish to characterise the trade-off between approximate and exact preparation, and
in particular to identify the trace-distance error regimes in which our exact protocols become competitive with existing approximate ones.
As these studies aim at exploring the preparation of the renormalised state, for those we fix the implementation of the isometry to be the sequential implementation of~\cref{subsec:SEQ_OAAI}.
In addition, we also assess improvements brought by an exact preparation coupled to an improved isometry circuit implementation. All these studies are based on the same ensemble of $64$ Haar-random TI MPS instances with $\xi \approx 3.3$ that was used previously. The estimated optimised depths for the various circuit constructions are reported in~\cref{fig:7deg:a} (for $\chi =4$) and in~\cref{fig:7deg:c} (for $\chi =8$). 

\begin{figure}[htp]
\begin{subfloat}[$\chi = 4$\label{fig:7deg:a}]{%
  \includegraphics[width=0.96\columnwidth]{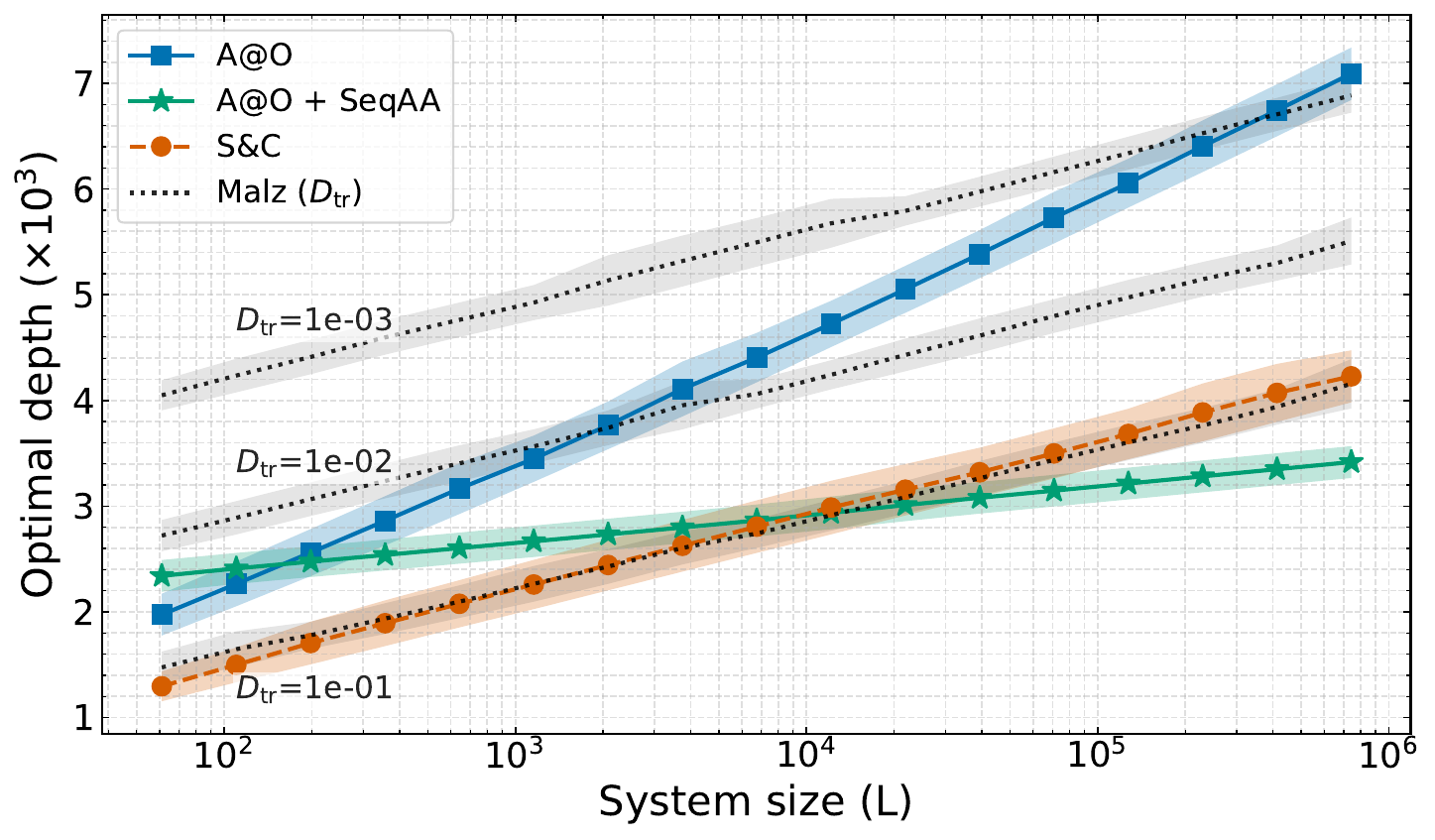}%
}\end{subfloat}

\begin{subfloat}[$\chi = 8$\label{fig:7deg:c}]{%
  \includegraphics[width=0.96\columnwidth]{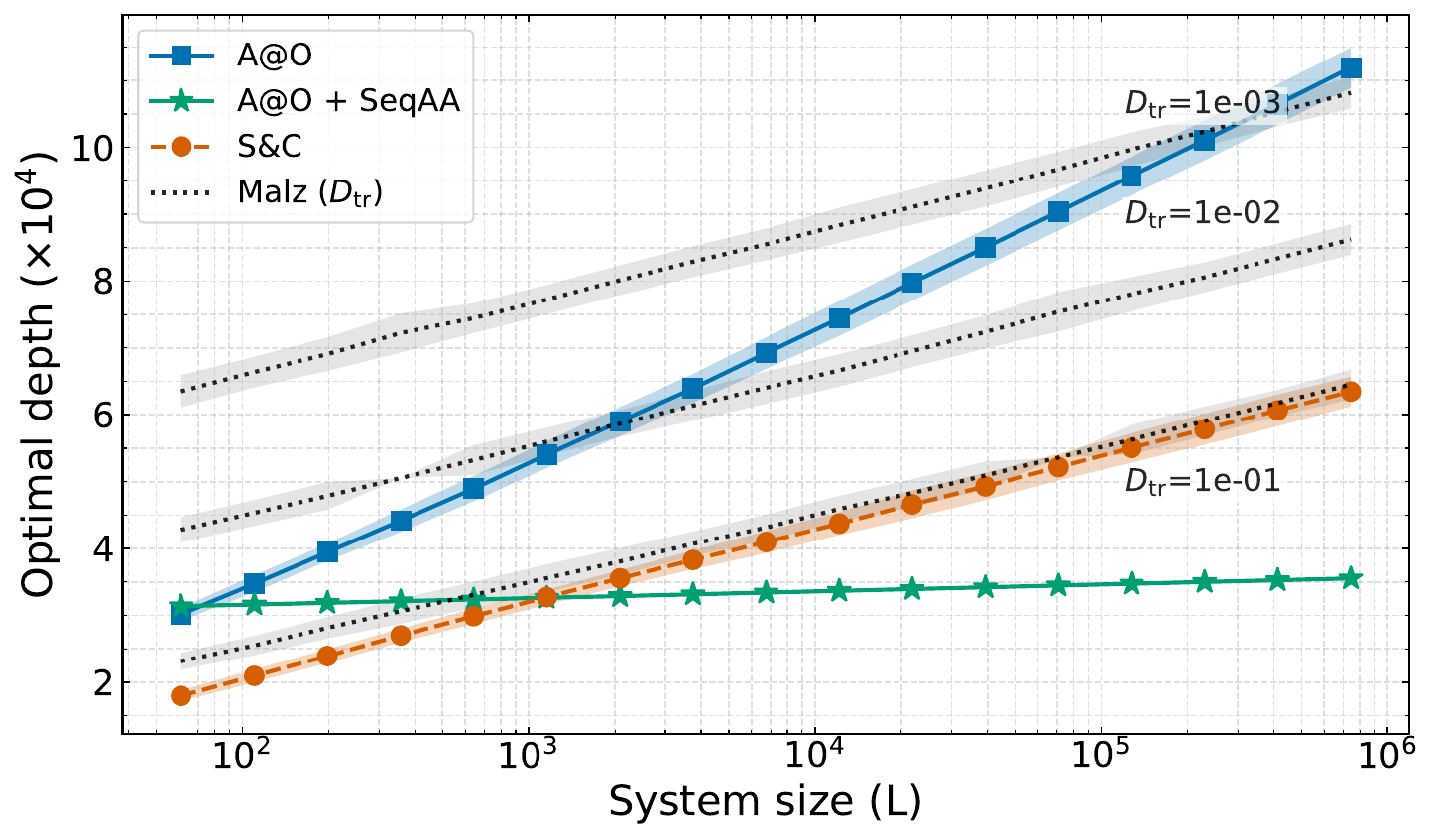}%
}\end{subfloat}
  \caption{
    Estimated circuit depth as a function of the system size $L$ (log scale) for the same ensemble of
    Haar-random TI MPS as in~\cref{fig:d_vs_L_oaa_linear_log}
    for bond dimensions $\chi\in \{4,8\}$.
    Each data point shows the mean $\pm$ one standard deviation over $N_{\mathrm{MPS}}=64$ instances.
    We compare the approximate RG-based approach of \citet{malz2024}
    for trace-distance errors $\varepsilon \in \{10^{-1}, 10^{-2}, 10^{-3}\}$ (grey)
    with our exact protocols:
    all-at-once (\aao) correction with sequential implementation of $V^{(q)}$ (blue),
    all-at-once (\aao) correction with \seqAA implementation of $V^{(q)}$ (green),
    and split-and-concatenate (\snc) correction with sequential $V^{(q)}$ (orange).  }
  \label{fig:d_vs_L_comp}
\end{figure}

Let us start by comparing the two exact protocols \aao and \snc, both with the sequential implementation of $V^{(q)}$.
While \snc is guaranteed to be at least as good as \aao, since we can always take $Q=L'$ to recover \aao from \snc, we wish to more precisely understand how much it improves on the latter. 
As can be seen, while both have similar scalings in terms of the system size (linear in $\log L$),
 the \snc curve (orange line) has essentially half the slope of the \aao curve (blue line).
This is understood as the \snc protocol halving the effective block size required to achieve a given probability of successful post-selection. It translates directly into an approximately factor-two
reduction in the circuits implementing $V^{(q)}$ (as these are linear in $q$). 
Supporting discussion of this effect is provided in~\cref{app:sac}.

Let us now compare exact and approximate preparation (dashed grey lines for different values of the trace-distance error $\varepsilon$).
Given the dependence of $\log(1/\varepsilon)$ in the circuit depth for the approximate MPS preparation~\eqref{eq:depth-costs-bg-overall}, we always expect exact preparation to be favourable for small enough error $\varepsilon$.
As seen in~\cref{fig:d_vs_L_comp}, this crossover already occurs for relatively large $\varepsilon$.
For both the bond dimensions probed, the \snc circuits become on par with the circuits for approximate preparation for values of the error $\varepsilon\approx 10\%$.
In other words, exact preparation becomes comparable to, or better than, the
approximate protocol once the desired error is less than $10\%$.

Finally, we study a combination of exact preparation together with improved implementation of the isometry.
Concretely, we report circuit depths entailed by \aao combined with \seqAA (green line).
Again, the scaling of the circuit depths obtained remains linear in $\log L$. However, we see that the slope of the combined methods is substantially decreased. 
Already, for $\chi=8$ and $L\geq10^4$ this shows that our circuit constructions would be smaller than approximate preparation even when accepting errors of $\varepsilon\approx 10\%$.
The dependence on the bond dimension can also be read off directly by comparing the \aao curves in~\cref{fig:7deg:a} and~\cref{fig:7deg:c}. At $L\approx 10^{2}$, the estimated depth grows from roughly $2200$ at $\chi=4$ to roughly $38000$ at $\chi=8$, a factor of approximately $16$. Since $(8/4)^{4}=16$, this ratio is in quantitative agreement with the $O(q\chi^{4})$ depth of the sequential implementation~\eqref{eq:dv_seq}.
While not always the most competitive over all the system sizes and bond dimensions probed, comparison of~\cref{fig:7deg:a} and~\cref{fig:7deg:c} shows that their advantage becomes more pronounced as $\chi$ increases.
For the \aao with \seqAA construction (green curves in~\cref{fig:7deg:a,fig:7deg:c}), the depth grows by roughly $10^3$ between $L=10^2$ and $L=10^6$ at $\chi=4$, and by roughly $4\times 10^3$ at $\chi=8$, a factor of approximately $4$. Since $(8/4)^{2}=4$, this $L$-growth ratio is
in agreement with the theoretical reduction of the $q$-dependent cost of $V^{(q)}$ from
$O(q\chi^{4})$ to $O(q\chi^{2})$ (see~\cref{tab:Vq-depth-summary}).
Overall, this shows that for entangled states, captured with MPS with relatively large bond dimensions, our methods become the most competitive with substantial improvements in the corresponding circuit depths especially as $\chi$ increases.

\subsection{Benchmarks for non-translationally invariant MPS}
\label{sec:nonTI_numerics}
Finally, we assess whether the extensions to the preparation of non-TI MPS, as presented in~\cref{sec:nonTI},
preserve the logarithmic depth scaling with respect to the system size.
To that end, we consider two classes of non-TI states: randomly generated non-TI MPS
with IID Haar-random tensors in~\cref{sec:random-nonTI-numerics}, and ground states of a disordered
Heisenberg chain obtained by DMRG in~\cref{sec:dmrg_impurity}.
While~\cref{thm:random-nonTI} guarantees preparability in $O(\log L)$ depth for the former, the latter falls outside the scope of the theorem as the local tensors can be correlated.
This allows us to probe the practical reach of~\cref{protocol:nonTI-exact} beyond the regime covered by the theoretical guarantees.
In both cases, we use the sequential implementation of the isometries, since it applies in parallel across non-TI instances and is agnostic to the detailed local structure of the state.
Although allowing non-uniform block sizes could provide additional room for optimisation, in the following
we fix the block sizes to be the same across all blocks as per~\cref{thm:random-nonTI}.
\subsubsection{Random non-TI MPS from IID Haar distribution}
\label{sec:random-nonTI-numerics}
\begin{figure}[!htp]
  \centering
    \includegraphics[width=0.5\textwidth]{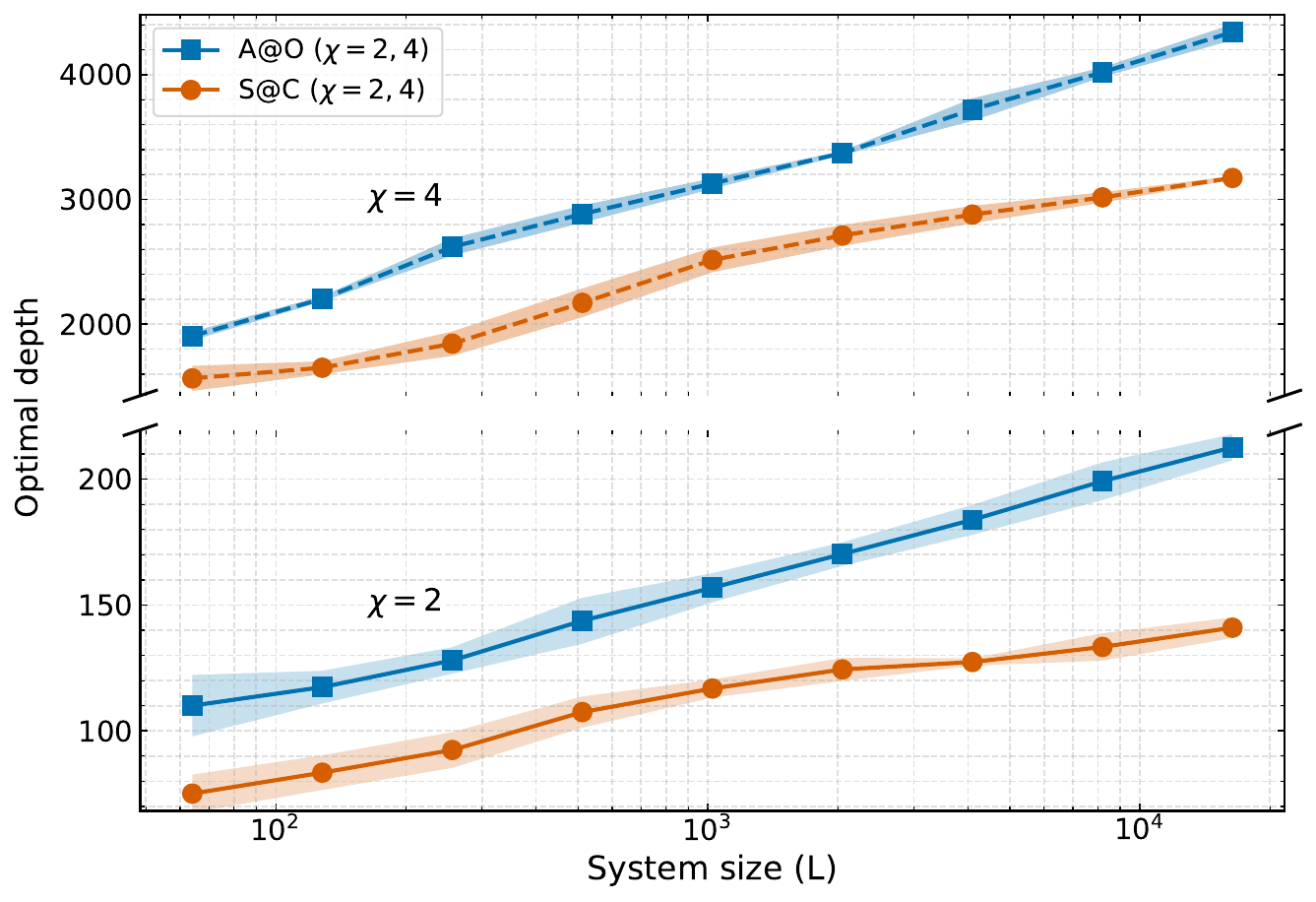}
  \caption{
    Estimated circuit depth as a function of the system size $L$ (log scale) for randomly generated
    non-TI MPS with open boundary conditions, with $\chi=2$ (solid lines) and $\chi=4$ (dashed lines).
    Each data point shows the mean $\pm$ one standard deviation over $N_{\mathrm{MPS}}=64$ instances.
    Two strategies for implementing $\ket{P^{(q)}(L')}$ are compared:
    all-at-once (\aao, blue) and split-and-concatenate (\snc, orange),
    both using the sequential implementation of $V^{(q)}$.
  }
  \label{fig:nonTI_depth_vs_L}
\end{figure}

We generate random non-TI MPS by independently sampling each left-canonical tensor $A^{[k]}$~\eqref{eq:nonTI_MPS_def} from the Haar measure.
For each system size $L$, we draw $N_{\mathrm{MPS}}=64$ independent instances and 
report the sample mean and standard deviation of the estimated depths.
As in the TI scenario, for \aao we optimise over the block size $q$ while for \snc we optimise over both $q$ and $Q$.
We consider bond dimensions $\chi \in \{2,4\}$, with results reported
in~\cref{fig:nonTI_depth_vs_L}.
As can be seen, the circuit depths scale approximately as $\log L$, with mild fluctuations, consistent with the asymptotic scaling of~\cref{thm:random-nonTI}.
Moreover, the absolute depths for $\chi=4$ are comparable to those reported for the TI case in~\cref{fig:7deg:a}, indicating that the loss of translational invariance does not introduce significant overhead.
The dependence on the bond dimension also mirrors the TI case:
at $L\approx 10^{2}$, the \aao depths for $\chi=2$ and $\chi=4$ are approximately $130$ and $2100$ respectively, yielding a ratio of about $16 = (4/2)^{4}$.
This is consistent with the dominant depth contribution $D^{\rm seq}_V(q) = O(\chi^4 q)$ coming from the sequential implementation of the isometries $V^{(q_{\ell})}$, which is the same for both TI and non-TI MPS.

\subsubsection{Depth for disordered Heisenberg model}
\label{sec:dmrg_impurity}

To assess if a scaling of the block sizes (and thus the circuit depth) in $\log L$ can also be expected in more general situations, we now consider a physically motivated family of non-TI MPS: ground states of a
disordered Heisenberg chain with random longitudinal fields~\cite{pal2010, luitz2015}. The Hamiltonian of the chain is defined as
\begin{equation}
  H = -J \sum_{i=1}^{L-1} \mathbf{S}_i \cdot \mathbf{S}_{i+1}
      - \sum_{i=1}^{L} h_i\, S_i^{z},
  \label{eq:impurity_ham}
\end{equation}
where $S_i^{\alpha} = \sigma_i^{\alpha}/2$ are spin-$\tfrac{1}{2}$ operators,
$J = 1$ sets the uniform exchange scale, and
$h_i \stackrel{\mathrm{IID}}{\sim} \mathcal{N}(0, h_{\mathrm{imp}}^{2})$
are independent Gaussian random longitudinal fields with strength
$h_{\mathrm{imp}}$.
Open boundary conditions are used throughout.
Because the random disorder breaks translational invariance, the ground state
is a non-TI MPS, and 
provides a natural physical instance for benchmarking our 
circuit constructions.
We compute these ground states via DMRG~\cite{white1992, schollwock2011} finding that the energy converges
rapidly with $\chi$.
We therefore restrict our depth analysis to the MPS obtained at $\chi = 2$ and $\chi = 4$,
which are already well-converged representatives of the ground-state manifold.

\begin{figure}[!htp]
  \centering
    \includegraphics[width=0.5\textwidth]{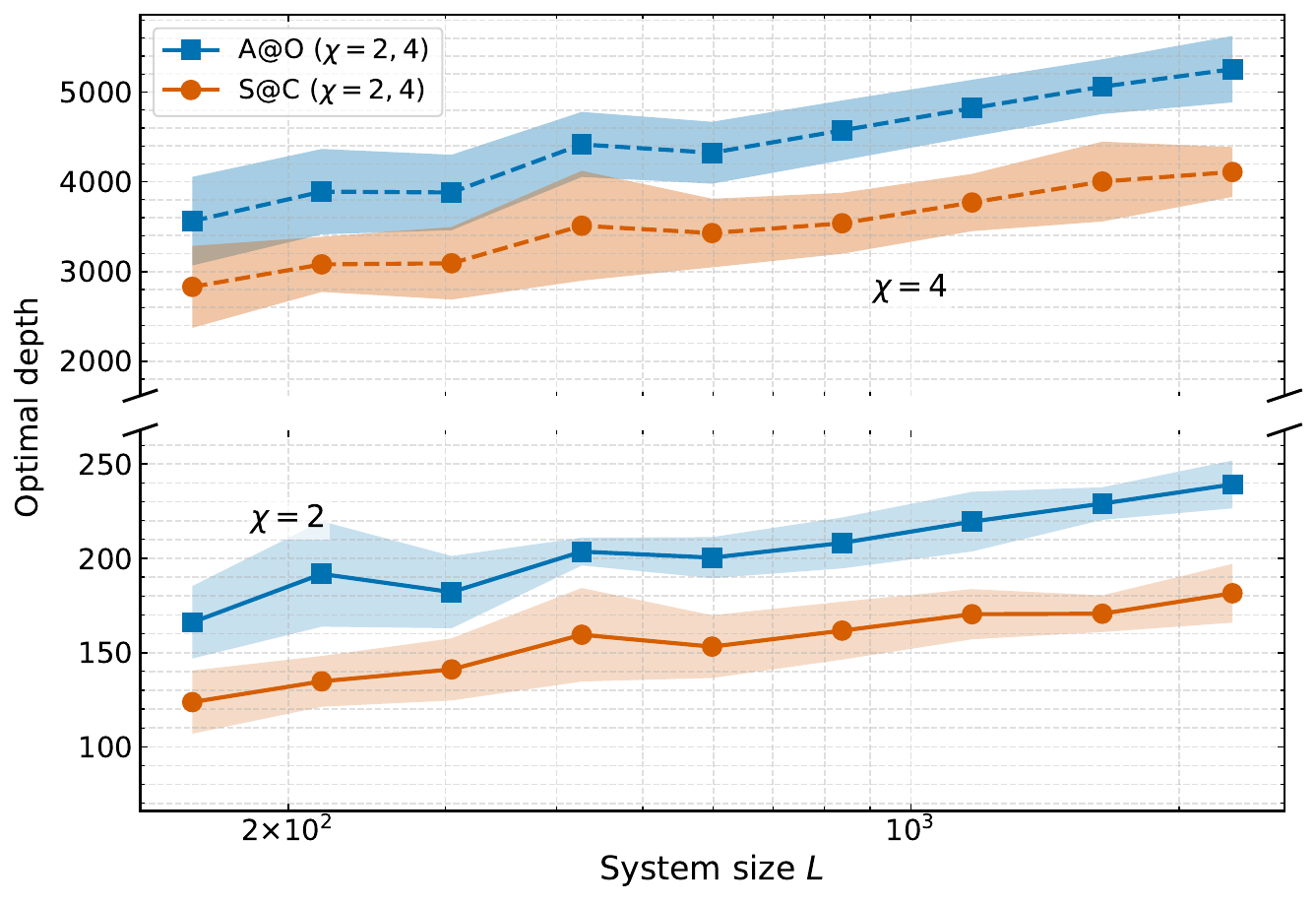}  
    \caption{
    Estimated circuit depth as a function of the system size $L$ (log scale) for ground states of the
    disordered Heisenberg chain~\cref{eq:impurity_ham}
    with $J=1$ and $h_{\mathrm{imp}}=0.5$.
    Ground states are obtained through DMRG and expressed as MPS with $\chi=2$ (solid lines) and $\chi=4$ (dashed lines).
    Each data point shows the mean $\pm$ one standard deviation over $N_{\mathrm{MPS}}=64$
    independent disorder realisations.
    Similarly to~\cref{fig:nonTI_depth_vs_L}, we compare the all-at-once (\aao) and
    split-and-concatenate (\snc) protocols, both using the sequential implementation of $V^{(q)}$.
  }
  \label{fig:dmrg_impurity_depth}
\end{figure}

As shown in~\cref{fig:dmrg_impurity_depth}, the required circuit depths
exhibit an approximately logarithmic growth with the system size $L$ for both bond dimensions, consistent with the behaviour observed for Haar-random non-TI MPS in~\cref{fig:nonTI_depth_vs_L}.
The $\chi$-dependence is also similar: at comparable system sizes, the ratio of depths between $\chi=4$ and $\chi=2$ is again consistent with an $O(\chi^4)$ scaling, as expected from the sequential implementation of the isometries.
This agreement suggests that the logarithmic scaling is a robust feature of non-TI MPS preparation that can extend beyond random ensembles to physically relevant ground states.

Overall, the numerical studies of this section support the practicability of our exact preparation protocols across the TI and non-TI settings. The studies indicate that \seqAA dominates \renAA at most relevant system sizes, and that our exact constructions match or improve on the approximate RG depths of~\citet{malz2024}. 
They also provide numerical support that the $O(\log L)$ scaling for non-TI MPS can persist beyond IID Haar-random tensors, as illustrated for ground states of a disordered Heisenberg chain.

\section{Conclusion}

In this work, we improved on RG methods~\cite{malz2024} for the preparation of MPS on quantum devices. First, by introducing correction maps implemented through block-encoding, we eliminated the approximation errors inherent to previous schemes (\cref{sec:exact-prep}). While these block-encodings introduce post-selection, we proved that its success probability can be made constant without altering the circuit-depth scalings. We further showed that this probability can be increased by careful layout of the circuits (the split-and-concatenate strategy) and that the preparation can even be made deterministic through exact amplitude amplification (\cref{subsec:improve_success}).

Second, combining the block-encoded correction maps with a generalisation of oblivious amplitude amplification to isometries (\cref{thm:OAAI}), we developed new implementations of the isometry $V^{(q)}$ (\cref{sec:implementation-Vq}), with the circuit depths of these implementations, along with end-to-end scalings summarised in~\cref{tab:Vq-depth-summary}. The resulting \seqAA circuit construction prepares a TI MPS exactly in depth $\tilde{O}(\chi^{2}\log L + \chi^{4})$, improving over $\tilde{O}(\chi^{4}\log(L/\varepsilon))$ for the sequential approximate preparation of~\citet{malz2024}. A double-logarithmic variant, \renAA, achieves $\tilde{O}(\chi^{2}\log\log L + \chi^{4})$ with all-to-all connectivity, improving over $\tilde{O}(\chi^{6}\log\log(L/\varepsilon))$ for the tree construction of~\citet{malz2024}. However, as shown in~\cref{sec:vq-selection}, while the latter is favourable asymptotically, its advantage only materialises at impractically large system sizes. We further proposed a quasi-probabilistic implementation of the circuits, suitable for the estimation of expectation values, that achieves $\tilde{O}(\chi^{3}\log L)$ circuit depth at the cost of a constant multiplicative sampling overhead.

The benefits of the different constructions were probed in numerical studies (\cref{sec:numerics}), confirming the constant-factor gains of the split-and-concatenate strategy, the improved bond-dimension dependences, and ultimately the improved depths of our circuits in practice.

While most of our results concerned TI MPS, we extended the theory to the non-TI case. In particular, leveraging the theory of ergodic quantum channels~\cite{Movassagh}, we proved logarithmic-depth exact preparation for non-TI MPS generated by IID random tensor sequences. Furthermore, our numerical results indicate that this scaling persists beyond the IID regime, as illustrated by ground-state preparations for a disordered Heisenberg model (\cref{sec:nonTI}).

Taken together, these constructions constitute, to the best of our knowledge, the most efficient exact MPS preparation protocols, and advance the preparation of highly entangled states on quantum computers. Several directions are left for future work.
First, the measure of circuit complexity adopted here was the CNOT-depth. When targeting fault-tolerant compilation, and depending on the details of the quantum error correction code employed, other measures can become more relevant. For instance, we expect the T-count of our circuits to follow the same scaling in $\chi$ and $L$ as the CNOT-depth reported here, with an additional multiplicative overhead scaling as $O(\log(1/\varepsilon_{\rm s}))$ which arises from the compilation of  rotations into Clifford+T gates~\cite{ross_selinger2014} up to a synthesis error $\varepsilon_{\rm s}$. Recent results on T-count-optimal state preparation~\cite{gosset2024} suggest that further improvements may be achievable. Still, a more precise analysis of the T-count (or of other metrics relevant to the fault-tolerant regime) would be valuable.
Second, our numerical studies probed only a subset of the proposed constructions. In particular, the quasi-probabilistic method and the split-and-concatenate variant of \seqAA, both likely to yield the shortest circuits but more demanding to benchmark, were not considered. A more systematic comparison across these combinations is a natural next step. 
Finally, several techniques can further reduce the depth of our circuits. For instance, a gauge optimisation of the MPS tensors was used in~\cite{scheer2025} to reduce the gate count when compiled to a specific gate-set. Also, while our numerical studies for non-TI MPS used a uniform block size across the chain, we expect that allowing this block size to vary across the blocks, and optimising over this choice, will have a substantial effect on the circuit depths.
All these refinements will contribute to advancing the preparation of highly entangled states of matter on quantum devices.
Experimentally, the RG scheme of Ref.~\cite{malz2024} has already been demonstrated on hardware~\cite{scheer2025} for system size of up to $L=80$ qubits but at small bond dimension $\chi=2$. The constructions presented here, despite additional but scalable classical preprocessing, will be instrumental in preparing states with higher entanglement relevant to many condensed-matter~\cite{haghshenas_2024, robertson_2025} and quantum-chemistry~\cite{White1999, stoudenmire_2012, chan_sharma_2011, baiardi_reiher_2020} simulation problems.

\section{Acknowledgements}

We thank Lewis Wright and Marcello Benedetti for carefully reading the manuscript and for their valuable feedbacks.

\bibliography{main}

\appendix

\newpage
\onecolumngrid

\setcounter{secnumdepth}{3}

\section*{Supplementary Material}
\makeatletter
\renewcommand{\bibnumfmt}[1]{[S#1]}
\renewcommand{\citenumfont}[1]{S#1}
\counterwithout{equation}{section}
\setcounter{equation}{0}

\renewcommand\theequation    {S\arabic{equation}}

\renewcommand\thefigure      {S\arabic{figure}}
\renewcommand\thetable       {S\arabic{table}}

\renewcommand\thetheorem     {S\arabic{theorem}}
\renewcommand\thelemma       {S\arabic{lemma}}
\renewcommand\thecorollary   {S\arabic{corollary}}
\renewcommand\thedefinition  {S\arabic{definition}}
\makeatother

In \cref{sec:sm-notation}, we recall the main notation used in the main text and review our tensor conventions.
\Cref{app:normal_mps} provides additional information for the treatment of normal TI MPS. \Cref{app:block-encoding} reviews basics of block-encoding. \Cref{app:qubit-count} details the use of ancillas in the different circuit constructions proposed. \Cref{app:non-normal} introduces extensions to non-normal MPS while~\cref{app:nonTI} addresses extensions to non-TI MPS. 
\Cref{app:numerics} provides additional details supporting the numerics section of the main text.
\Cref{app:proofs_OAAI} is dedicated to proving the exact oblivious amplitude amplification for isometries.
Finally,~\cref{app:sac} analyses the cost of the split-and-concatenate strategy, and~\cref{app:incoherent} develops the quasi-probabilistic implementation of the correction map.

\section{Notation and conventions}
\label{sec:sm-notation}
In this section we recall the main tensors and notations introduced in the main text. These will repeatedly be used in the following appendices. We further clarify (sometimes ad-hoc) conventions taken when dealing with tensors. 

Recall that the TI MPS~\eqref{eq:MPS-def-indices}, that we wish to prepare as a circuit, is built from a 3-leg tensor $A\in\mathbb{C}^{\chi\times d\times\chi}$ with $\chi$ the bond (or virtual) dimension and $d$ the physical dimension. Blocking of $q$ of such tensors yields the \emph{blocked tensor} $A^{(q)}\in \mathbb{C}^{\chi \times d^q \times \chi}$~\eqref{eq:blocking}.
To explicitly indicate entries of tensors, we decorate them with brackets and 
reserve Greek letters for the virtual indices ($\alpha,\beta,\gamma,\delta \in [\chi]:= \{1,\cdots ,\chi\}$) together with Latin letters for the physical indices ($i, j,k,l\in[d]$). For instance, $[A]^i_{\alpha, \beta}$ is a single entry of $A$ while $[A^{(q)}]^\mathbf{i}_{\alpha,\beta}$ is a single entry of the blocked tensor $A^{(q)}$. For the latter, we collected all the physical indices into a single multi-index $\mathbf{i}=(i_1,\dots,i_q)\in[d]^q$.
Alternatively, when fixing only a partial number of indices we refer to a slice of the original tensor. In particular, $[A]^i\in \mathbb{C}^{\chi \times \chi}$ (often simply written as $A^i$) is used to define the \emph{transfer matrix} associated to $A$ as 
\begin{equation}
    E=\sum_i A^i\otimes\bar{A}^i\in\mathbb{C}^{\chi^2\times\chi^2},
\end{equation}
which maps a pair of virtual spaces ($\mathbb{C}^{\chi^2}$) to another pair of virtual spaces ($\mathbb{C}^{\chi^2}$). To indicate its entries, we use a pair of virtual indices such that
\begin{equation}\label{eq:transfer_entries}
    [E]_{(\alpha,\gamma),(\beta,\delta)}=\sum_i [A]^i_{\alpha,\beta}\,\overline{[A]^i_{\gamma,\delta}}.
\end{equation}

At times, we overload notations and reshape tensors as matrices. 
In particular, tensors $A$ or $A^{(q)}$ are usually viewed as maps from the virtual-pair space to the physical space such that $A\in \mathbb{C}^{d\times\chi^2}$ and $A^{(q)}\in\mathbb{C}^{d^q\times\chi^2}$.
For instance, this reshaping is used when defining the polar decomposition $A^{(q)}=V^{(q)}P^{(q)}$~\eqref{eq:polar} that produces a \emph{positive factor} $P^{(q)}\in\mathbb{C}^{\chi^2\times\chi^2}$ together with an \emph{isometry} $V^{(q)}\in\mathbb{C}^{d^q\times\chi^2}$.
To distinguish it from other spaces involved, we refer to the output space of $P^{(q)}$ as the \emph{renormalised physical space} (effectively isomorphic to a pair of virtual spaces) and index it through a renormalised physical index $\mu\in [\chi^2]$.
Fixing $\mu$ then yields the matrix $[P^{(q)}]^\mu \in \mathbb{C}^{\chi \times \chi}$.

Another form of matrix manipulation, that will be used in a few places, does not change the dimension of the matrix but only the arrangement of its entries, and we talk about a change of \emph{matrix view}. Consider a tensor with four virtual indices $(\alpha, \beta, \gamma, \delta)$ reshaped as a $\chi^2\times\chi^2$ matrix. There is freedom in how the pairs of virtual indices are grouped, yielding different matrix views. For instance, for the transfer matrix~\cref{eq:transfer_entries}, we could also define $F=A^\dag A$ (the Gram matrix associated to $A\in \mathbb{C}^{d \times \chi^2}$) for which it can be verified that $[F]_{(\alpha, \gamma), (\beta, \delta)} = [E]_{(\alpha, \beta), (\gamma, \delta)}$. That is, $E$ and $F$ are different matrix views of the same tensor (we denote this equivalence as $E \equiv F$).
Given that only the position of the matrix entries changes when adopting a new view, the Frobenius norm, denoted by $\|\cdot\|_2$, is preserved. For the transfer matrix, we see that
\begin{equation}\label{eq:frob_view}
  \|E\|_2
  = \bigl\|\sum_i A^i\otimes\bar A^i\bigr\|_2
  = \|A^\dag A\|_2.
\end{equation}
By contrast, the spectral norm $\|\cdot\|_\infty$, defined as the largest singular value, can depend on the chosen view.
Note however that for any finite-rank matrix $M$, the Frobenius and spectral norms can be related through the inequalities
\begin{equation}\label{eq:ineq_schatten}
  \|M\|_\infty \le \|M\|_2 \le \sqrt{\operatorname{rank}(M)}\,\|M\|_\infty.
\end{equation}

Finally, recall that when $q\to\infty$ the positive factor $P^{(q)}$ converges to the \emph{RG fixed-point operator} $P^{(\infty)}\in \mathbb{C}^{\chi^2 \times \chi^2}$. For normal MPS, this fixed point has the form~\eqref{eq:pinf_product}
\begin{equation}\label{eq:pinf_product_app}
    P^{(\infty)}=\sigma\otimes\mathbb{I}_\chi\in\mathbb{C}^{\chi^2\times\chi^2},
\end{equation}
in terms of a $\sigma\in \mathbb{C}^{\chi \times \chi}$ which is positive-definite ($\sigma\succ0$) and satisfies $\tr[\sigma \sigma^\dag] = 1$~\cite{malz2024, piroli2021}.
We further define the \emph{residual} operator~\eqref{eq:residual} as their difference
\begin{equation}\label{def:residual_app}
    R^{(q)}:=P^{(q)}-P^{(\infty)}\in\mathbb{C}^{\chi^2\times\chi^2}
\end{equation}
and the \emph{correction} map~\eqref{eq:def-Cq} as the operation mapping one to the other, $C^{(q)}:=P^{(q)}(P^{(\infty)})^{-1}\in\mathbb{C}^{\chi^2\times\chi^2}$.

\section{Normal MPS, MPS transfer matrix, and residual bound}
\label{app:normal_mps}
In this appendix, we provide additional details concerning the treatment of normal TI MPS. This includes a proof of the asymptotic normalisation stated in~\cref{eq:norm} of the main text, a derivation of the bound on the residual norm reported in~\cref{residual} and used throughout the main text, and a proof that~\cref{eq:scaling-q} is indeed the scaling of the block size $q$ needed to achieve $\varepsilon$-error in the MPS preparation task.

\subsection{Basic properties of normal MPS}
\label{app:normal-basic}
By using the gauge freedom of a TI MPS, every MPS tensor $A$ can be brought to a block-diagonal canonical form~\cite[Section~2.3]{cirac_2017}, such that each of the $A^i\in \mathbb{C}^{\chi \times \chi}$ (for $i =1, \hdots,d$) can be block decomposed as
\begin{equation}\label{eq:canonical-form}
  A^i \;=\; \bigoplus_{b=1}^{g} \nu_b\, A^{i}_b,
\end{equation}
with $\nu_b \in \mathbb{C}$ with $|\nu_b|\leq 1$,
and where each block $A^i_b \in \mathbb{C}^{\chi_b \times \chi_b}$ (such that $\sum \chi_b = \chi$) satisfies the left-canonical condition $\sum_i A^{i\dag}_b A^i_b = \mathbb{I}_{\chi_{b}}$.
When the decomposition consists of a single block ($g=1$), we say that the MPS is \emph{normal}~\cite{perez_garcia_2007}, which is the setting considered in this section.

Since the MPS transfer matrix $E = \sum_i A^i \otimes \bar{A}^i$ also represents the channel
$X\mapsto\sum_i A^i X A^{i\dag}$,
the left-canonical condition $\sum_i A^{i\dag}A^i=\mathbb{I}_{\chi}$ implies that the identity $\mathbb{I}_\chi$ is a left eigenvector of $E$ with eigenvalue $1$, in the sense that
\begin{equation}\label{eq:transfer-left-ev}
  \dbra{\mathbb{I}_\chi}E=\dbra{\mathbb{I}_\chi}.
\end{equation}
Here $\dket{X}$ denotes the vectorisation of an operator $X\in\mathbb{C}^{\chi\times\chi}$ through the mapping $\ket{i}\bra{j}\mapsto\ket{i}\otimes\ket{j}$, and $\dbra{X}$ is the dual of $\dket{X}$ under this inner product $\dbraket{Y}{X}=\tr(Y^\dag X)$.
For a normal MPS~\cite{cirac_2017}, the eigenvalue $\lambda_1=1$ is unique and every other eigenvalue $\lambda_j$ satisfies $|\lambda_j|<1$. In addition to this unit eigenvalue, we denote the subleading one (second largest in modulus) as $\lambda_2$, and define the correlation length of $A$ as 
\begin{equation}\label{eq:correll}
  \xi := \frac{-1}{\log |\lambda_2|} > 0.
\end{equation}
It is known that the right eigenvector corresponding to $\lambda_1$ is the vectorisation of a unique positive-definite matrix $\rho\in\mathbb{C}^{\chi\times\chi}$ satisfying~\cite{cirac_2017,piroli2021}
\begin{equation}\label{eq:fixed_point_one}
  E\dket{\rho}=\dket{\rho},
\end{equation}
or equivalently $\sum_i A^i \rho A^{i\dag} = \rho$, where $\rho$ is normalised as $\dbraket{\rho}{\mathbb{I}_\chi} = \tr(\rho) = 1$.
The matrix $\sigma = \sqrt{\rho}$ then entirely specifies the RG fixed-point operator through~\cref{eq:pinf_product_app}.
\paragraph{Computing $\sigma$.}
\label{app:compute-sigma}
In practice, $\rho$ is obtained by diagonalising $E$ and reshaping the resulting right eigenvector from $\mathbb{C}^{\chi^2}$ back into a $\chi \times \chi$ matrix (i.e. the inverse of the vectorisation map).
The matrix $\sigma$ is then given by the unique positive square root $\sigma = \sqrt{\rho}$.
Since $\rho \succ 0$ for a normal MPS, $\sigma$ is also positive-definite, as stated in the main text after~\cref{eq:pinf_product}.
Straightforward diagonalisation of $E \in \mathbb{C}^{\chi^2 \times \chi^2}$ requires $O(\chi^6)$ arithmetic operations. 
Alternatively, the dominant right eigenvector can be obtained through power iteration of $E$ (equivalently by repeatedly applying the channel $X \mapsto \sum_i A^i X A^{i\dag}$), where each iteration necessitates $O(d\chi^3)$ operations for the matrix multiplications and convergence is exponential in the iteration count, at rate $1/\xi$.
Taking the square root of the $\chi \times \chi$ matrix $\rho$ requires $O(\chi^3)$ additional operations.
That is, obtaining $\sigma$ is efficient in $\chi$, and can easily be done for the bond dimensions typically encountered in practice.

\subsection{Asymptotic normalisation of the normal MPS}
\label{app:asymptotic-normalisation}
For the $L$-site periodic TI MPS~\eqref{eq:MPS-def-indices},
the squared norm can be expressed in terms of the transfer matrix, as
\begin{equation}\label{eq:norm-transfer}
  \|\ket{A(L)} \|^2 := \braket{A(L) | A(L)} = \tr\!\bigl(E^L\bigr).
\end{equation}
We now show that this norm converges to unity exponentially in the system size $L$, as per~\cref{eq:norm} of the main text. In what follows, we consider both the case when $E$ is diagonalisable and when it is not.

If $E$ is diagonalisable, it can be decomposed as
\begin{equation}
  E = S\Lambda S^{-1},
\end{equation}
for some invertible matrix $S$ and with $\Lambda = \mathrm{diag}(\lambda_1, \ldots, \lambda_{\chi^2})$ the diagonal matrix of eigenvalues. Hence,
\begin{equation}
  \tr\!\bigl(E^L\bigr) = \tr\!\bigl(\Lambda^L\bigr) = \sum_{j=1}^{\chi^2} \lambda_j^L
  = 1 + \sum_{j=2}^{\chi^2} \lambda_j^L.
\end{equation}
In turn, we can bound deviation of the norm from $1$ through
\begin{equation}\label{eq:norm-bound-diag}
  \bigl|\braket{A(L)|A(L)} - 1\bigr|
  \leq \sum_{j=2}^{\chi^2} |\lambda_j|^L
  \leq (\chi^2 - 1)\,|\lambda_2|^L,
\end{equation}
where we recall that $|\lambda_2|<1$.
Inserting the \emph{correlation length}~\eqref{eq:correll}, 
the previous bound becomes 
\begin{equation}\label{eq:norm-bound-diag-2}
  \bigl|\braket{A(L)|A(L)} - 1\bigr|
  \leq(\chi^2 - 1)\,e^{-L/\xi}.
\end{equation}

If $E$ is not diagonalisable, the bound~\eqref{eq:norm-bound-diag-2} still applies.
To see this, put $E$ into Jordan normal form $E = S\,J\,S^{-1}$ with $J = \bigoplus_\alpha J_\alpha(\lambda_\alpha)$, where each $J_\alpha(\lambda_\alpha)$, with size $m_\alpha \times m_\alpha$, is the Jordan block corresponding to the eigenvalue $\lambda_\alpha$.
Since the trace is similarity-invariant,
\begin{equation}\label{eq:non_diag_trace}
  \tr\!\bigl(E^L\bigr)
  = \tr\!\bigl(J^L\bigr)
  = \sum_\alpha \tr\!\bigl(J_\alpha(\lambda_\alpha)^L\bigr).
\end{equation}
Write each Jordan block as $J_\alpha(\lambda_\alpha) = \lambda_\alpha \, I_\alpha + N_\alpha$, where $I_\alpha$ is the identity on $\mathbb{C}^{m_{\alpha}}$ and $N_\alpha$ is the nilpotent shift matrix (satisfying $N_\alpha^{m_\alpha} = 0$ and $\|N_\alpha\|_\infty \leq 1$). Then its $L$-th power is given by
\begin{equation}\label{eq:jordan-power}
  J_\alpha(\lambda_\alpha)^L
  = \sum_{k=0}^{m_\alpha - 1} \binom{L}{k}\, \lambda_\alpha^{L-k}\, N_\alpha^k.
\end{equation}
For $k \geq 1$, any of the matrices $N_\alpha^k$ is strictly upper-triangular and therefore traceless.
Only the $k = 0$ term contributes to the sum in~\cref{eq:non_diag_trace}, giving
\begin{equation}\label{eq:tr-jordan-block}
  \tr\!\bigl(J_\alpha(\lambda_\alpha)^L\bigr) = m_\alpha\, \lambda_\alpha^L.
\end{equation}
Summing over all blocks and using that $\lambda_1 = 1$ is simple ($m_1 = 1$), one obtains
\begin{equation}
  \braket{A(L) | A(L)} = 1 + \sum_{\lambda_\alpha \neq 1} m_\alpha\, \lambda_\alpha^L.
\end{equation}
Since $|\lambda_\alpha| \leq |\lambda_2| < 1$, we get
\begin{equation}\label{eq:norm-bound}
  \bigl|\braket{A(L) | A(L)} - 1\bigr|
  \leq \sum_{\lambda_\alpha \neq 1} m_\alpha\, |\lambda_\alpha|^L
  \leq C\, |\lambda_2|^L,
\end{equation}
where $C = \sum_{\lambda_\alpha \neq 1} m_\alpha \leq \chi^2 - 1$, akin to~\cref{eq:norm-bound-diag}. Hence, in both cases we have
\begin{equation}\label{eq:A-norm}
  |\braket{A(L)|A(L)}-1| =  O\!\bigl(\chi^2\, e^{-L/\xi}\bigr),
\end{equation}
which justifies the assumption $\|\ket{A(L)}\| = 1$ for $L \gg \xi$.
For non-normal MPS ($g > 1$), the above properties hold independently within each block $A_b$ appearing in~\cref{eq:canonical-form}.
The detailed treatment is provided in~\cref{app:non-normal}.

\subsection{Derivation of the residual bound}
\label{app:residual-bound}

We now derive the bound on the norm of the residual
$R^{(q)} := P^{(q)} - P^{(\infty)}$ that was stated in~\cref{residual} of the main text, using the MPS transfer matrix $E$, its spectral properties, and the correlation length $\xi$ introduced above.

As before, when $E$ is diagonalisable we have
$E = S\Lambda S^{-1}$ with
$\Lambda = \mathrm{diag}(\lambda_1, \ldots, \lambda_{\chi^2})$,
such that $|\lambda_j| \leq e^{-1/\xi} <1$ for all $j \geq 2$.
Since $E^q = S\Lambda^q S^{-1}$, only the leading eigenvalue survives in the limit:
\begin{equation}\label{eq:Einfty-rank-one}
  E^\infty := \lim_{q\to\infty}E^q = S\,\mathrm{diag}(1,0,\ldots,0)\,S^{-1} = \dket{\rho}\dbra{\mathbb{I}_\chi} = \dket{\sigma^2}\dbra{\mathbb{I}_\chi},
\end{equation}
The rightmost equality identifies $\dket{\sigma^2}$ and $\dbra{\mathbb{I}_\chi}$ as the right and left eigenvectors of $E^\infty$ and corresponds, up to vectorisation, to $P^{(\infty)} = \sigma\otimes\mathbb{I}_\chi$ of \cref{eq:pinf_product_app}. In turn,
\begin{equation}
  E^q - E^\infty = S\,\mathrm{diag}(0,\lambda_2^q,\ldots,\lambda_{\chi^2}^q)\,S^{-1}.
\end{equation}
Note that the spectral norm of a diagonal matrix equals its largest entry in absolute value.  
Then, submultiplicativity of the spectral norm together with $|\lambda_j| \le e^{-1/\xi}$ for $j \ge 2$ from~\cref{eq:correll} gives
\begin{equation}
  \label{eq:Eq-Frob-bound}
  \|E^q - E^\infty\|_\infty
  \leq \|S\|_\infty\,\|S^{-1}\|_\infty\,\max_{j\geq 2}|\lambda_j^q|
  \leq \|S\|_\infty\,\|S^{-1}\|_\infty\,e^{-q/\xi}.
\end{equation}

Recall that $P^{(q)} = \sqrt{(A^{(q)})^\dag A^{(q)}}$ is the positive factor
in the polar decomposition of the blocked tensor $A^{(q)}$, with
\begin{equation}
  \label{eq:P2-as-Eq}
  P^{(q)} P^{(q)} = (A^{(q)})^\dag A^{(q)} = \sum_{\mathbf{i}} [A^{(q)}]^{\mathbf{i}\dag} [A^{(q)}]^{\mathbf{i}} \equiv E^q.
\end{equation}
That is, as per the discussion in~\cref{sec:sm-notation}, $(P^{(q)})^2$ (or its difference from the fixed point) is equivalent to the corresponding $E^q$ (or its difference from the fixed point), up to a change of view. In turn, their Frobenius norms are the same.
Hence,
\begin{equation}
  \label{eq:P2-bound}
  \|(P^{(q)})^2 - (P^{(\infty)})^2\|_2
  = \|E^q - E^\infty\|_2 \leq \chi \|E^q - E^\infty\|_{\infty}
  \leq \chi \|S\|_\infty\,\|S^{-1}\|_\infty\,e^{-q/\xi}.
\end{equation}
For the second inequality, we used~\cref{eq:ineq_schatten} together with the fact that the operators act on a space of dimension $\chi^2$. For the last inequality, we made use of~\cref{eq:Eq-Frob-bound}.

The identity $A^2 - B^2 = A(A - B) + (A - B)B$ applied with $A = P^{(q)}$, $B = P^{(\infty)}$ and $R^{(q)}=P^{(q)}- P^{(\infty)}$ yields:
\begin{equation}
  \label{eq:sylvester-Rq}
  P^{(q)}\,R^{(q)} + R^{(q)}\,P^{(\infty)}
  = (P^{(q)})^2 - (P^{(\infty)})^2.
\end{equation}

Let us take a Hilbert-Schmidt inner product
with the Hermitian matrix $(R^{(q)})^{\dag}$.
For the left hand side of~\cref{eq:sylvester-Rq}, we get
\begin{equation}
\begin{split}
\operatorname{tr}\!\bigl((R^{(q)})^\dag (   P^{(q)}\,R^{(q)} + R^{(q)}\,P^{(\infty)})\bigr)
  &= \operatorname{tr}\!\bigl( P^{(q)}\,R^{(q)}(R^{(q)})^\dag\bigr) + \operatorname{tr}\!\bigl(P^{(\infty)}\, (R^{(q)})^\dag R^{(q)}\bigr) \\
  &\geq  \bigl(\lambda_{\min}(P^{(q)}) + \lambda_{\min}(P^{(\infty)})\bigr)
    \|R^{(q)}\|_2^2.
\end{split}
\end{equation}
where we used that $\operatorname{tr}\!(AB)\geq \lambda_{\min}(A)\operatorname{tr}\!(B)$ as long as $B\geq0$.
For the right hand side of~\cref{eq:sylvester-Rq} and using Cauchy-Schwarz inequality, we get
\begin{equation}
\begin{split}
    \operatorname{tr}\!\bigl((R^{(q)})^\dag
      \bigl[(P^{(q)})^2 - (P^{(\infty)})^2\bigr]\bigr)
    \leq
    \|R^{(q)}\|_2\,\|(P^{(q)})^2 - (P^{(\infty)})^2\|_2.
\end{split}
\end{equation}
Taken together these yield:
\begin{equation}
\begin{split}
    \bigl(\lambda_{\min}(P^{(q)}) + \lambda_{\min}&(P^{(\infty)})\bigr)
    \|R^{(q)}\|_2^2 \leq
    \|R^{(q)}\|_2\,\|(P^{(q)})^2 - (P^{(\infty)})^2\|_2.
\end{split}
\end{equation}

Dividing by $\|R^{(q)}\|_2$ and dropping the non-negative
$\lambda_{\min}(P^{(q)})$ term, as $P^{(q)}\succ0$, we get
\begin{equation}
  \label{eq:Rq-upper}
  \|R^{(q)}\|_2
  \leq \frac{\|(P^{(q)})^2 - (P^{(\infty)})^2\|_2}
            {\lambda_{\min}(P^{(\infty)})}.
\end{equation}
Substituting~\cref{eq:P2-bound} into~\cref{eq:Rq-upper}, and 
recalling that $\|(P^{(\infty)})^{-1}\|_\infty = \lambda_{\min}(P^{(\infty)})^{-1}$, we finally obtain
\begin{equation}\label{eq:app_res_bound}
  \|R^{(q)}\|_2 \leq \Gamma\,e^{-q/\xi},
  \quad {\rm where} \quad
  \Gamma
  := \chi\,\|S\|_\infty\,\|S^{-1}\|_\infty \|(P^{(\infty)})^{-1}\|_\infty.
\end{equation}
Note that the definition of $\Gamma$ stated in~\cref{def:gamma}
of the main text uses the Frobenius norm rather than the spectral norm.
Since $\|\cdot\|_\infty \le \|\cdot\|_2$, the main text version yields looser (but still equally valid) bounds.

For completeness, we include the case of non-diagonalisable MPS transfer matrices, although this case is often not relevant in practice. In particular, none of the MPS instances generated during our numerical experiments~\cref{sec:numerics} produced a non-diagonalisable $E$.
If $E$ is not diagonalisable, we use the Jordan decomposition $E = S\,J\,S^{-1}$ and the power formula~\eqref{eq:jordan-power} introduced in~\cref{app:normal_mps}.
The limit $J^\infty := \lim_{q\to\infty} J^q$ has a single $1$ in the leading-eigenvalue position and zeros elsewhere, since $\lambda_1=1$ while all other $|\lambda_\alpha|<1$.
Since $\|N_\alpha^r\|_\infty = \|N_\alpha\|_\infty^r = 1$, taking norms of~\eqref{eq:jordan-power} and using $\binom{q}{r}\le q^r/r!$,
$|\lambda_\alpha|\le|\lambda_2|$, and $q^r\le q^{m_\alpha-1}$ for $r\le m_\alpha-1$, gives
\begin{equation}\label{eq:jordan-block-bound}
  \|J_\alpha(\lambda_\alpha)^q\|_\infty
  \le \sum_{r=0}^{m_\alpha-1} \binom{q}{r}\,|\lambda_\alpha|^{q-r}
  \le |\lambda_2|^q \sum_{r=0}^{m_\alpha-1} \frac{q^r}{r!\,|\lambda_2|^r}
  \le C\, q^{m_\alpha-1}\,|\lambda_2|^q,
\end{equation}
where $C := e^{1/|\lambda_2|}$ bounds the partial sum of its exponential series.
Let $m := \max\{m_\alpha:\ |\lambda_\alpha|<1\}$ be the largest Jordan block size among the
subdominant eigenvalues. Since $J^q - J^\infty$ is block diagonal with a vanishing leading block, its spectral norm equals the largest of the subdominant block norms, so~\eqref{eq:jordan-block-bound} gives $\|J^q - J^\infty\|_\infty \le C\,q^{m-1}\,|\lambda_2|^q$ and
\begin{equation}\label{eq:Ejordan-bound}
  \|E^q - E^\infty\|_\infty
  = \|S(J^q - J^\infty)S^{-1}\|_\infty
  \le C\, \|S\|_\infty\|S^{-1}\|_\infty\,  q^{m-1}\,|\lambda_2|^q
  = C\, \|S\|_\infty\|S^{-1}\|_\infty\,  q^{m-1}\,e^{-q/\xi},
\end{equation}
where we also used $E^\infty = S\,J^\infty\,S^{-1}$, submultiplicativity of $\|\cdot\|_\infty$, and $|\lambda_2| = e^{-1/\xi}$ from~\cref{eq:correll}.
Following the same steps as in the diagonalisable case, the bound~\eqref{eq:Ejordan-bound}
implies
\begin{equation}
  \|(P^{(q)})^2 - (P^{(\infty)})^2\|_2
  = \chi \|E^q - E^\infty\|_\infty
  \le C\, \chi \|S\|_\infty\|S^{-1}\|_\infty\, q^{m-1}\,e^{-q/\xi},
\end{equation}
Plugging this expression into the bound~\eqref{eq:Rq-upper} gives the
non-diagonalisable analogue of the residual estimate:
\begin{equation}
  \|R^{(q)}\|_2
  \le \Gamma\, q^{m-1}\,e^{-q/\xi},
  \quad {\rm where} \quad
  \Gamma := C \chi \|S\|_\infty\|S^{-1}\|_\infty\, \|(P^{(\infty)})^{-1}\|_\infty,
\end{equation}
with an additional prefactor $q^{m-1}$ compared to~\cref{eq:app_res_bound}, the absolute constant $C$ being absorbed into the redefinition of $\Gamma$.
The polynomial prefactor can be absorbed into a purely exponential bound at the cost of a slightly worse correlation length.
For any $\xi' > \xi$, define $\eta = 1/\xi - 1/\xi' > 0$.
The function $q \mapsto q^{m-1}\,e^{-\eta q}$ is maximized at $q^* = (m-1)/\eta$, giving
\begin{equation}
  \sup_{q \ge 0}\, q^{m-1}\,e^{-\eta q}
  = \left(\frac{m-1}{e\,\eta}\right)^{m-1}
  = \left(\frac{(m-1)\,\xi\,\xi'}{e\,(\xi'-\xi)}\right)^{m-1}
  =: C_m.
\end{equation}
Since $q^{m-1}\,e^{-q/\xi} = q^{m-1}\,e^{-\eta q} e^{-q/\xi'} \le C_m\,e^{-q/\xi'}$, the non-diagonalisable residual bound becomes
\begin{equation}
  \|R^{(q)}\|_2 \le \tilde{\Gamma}\,e^{-q/\xi'},
  \qquad
  \tilde{\Gamma} := \Gamma\,C_m.
\end{equation}
All results from the diagonalisable case therefore carry over with $\xi \to \xi'$ and $\Gamma \to \tilde{\Gamma}$.
In particular, the block size scaling $q = O(\xi\log(\Gamma L/\varepsilon))$ from~\cref{eq:scaling-q} becomes $q = O(\xi'\log(\tilde{\Gamma} L/\varepsilon))$, and likewise for the choice of $q_h$ in the OAAI construction.

However, the cost of this absorption can, in the worst case, be substantial.
Since $m$ can be as large as $\chi^2 - 1$ (the dimension of the subdominant eigenspace of $E$), the constant $C_m$ can be exponential in $\chi^2$:
\begin{equation}
  \log C_m = (m-1)\log\!\left(\frac{(m-1)\,\xi\,\xi'}{e\,(\xi'-\xi)}\right)
  = O\!\left(\chi^2\,\log\frac{\chi^2  \xi\,\xi'}{\xi'-\xi}\right).
\end{equation}
This contributes an $O(\chi^2\,\xi')$ term to the required block size $q$ (through $\xi'\log \tilde{\Gamma}$), which can dominate the $O(\xi\log(\Gamma L/\varepsilon))$ scaling of the diagonalisable case.
If one wishes to turn this bound into a concrete estimate, the free parameter $\xi'$ must be fixed. Choosing $\xi' = 2\xi$ sets $\xi'-\xi = \xi$ and renders $C_m$ a finite constant determined by $\xi$ and $m$.

In practice,
non-diagonalisable matrices form a measure-zero set in the space of
$\chi^2 \times \chi^2$ matrices,
and since $E$ depends smoothly on the MPS tensor $A$,
for generic $A$ the MPS transfer matrix is diagonalisable ($m = 1$) such that the tighter bound of~\cref{eq:app_res_bound} applies.
The non-diagonalisable bound above is included for mathematical completeness.

\subsection{Trace distance bound}
\label{app:trace-distance}

In this appendix, we wish to translate the previous bound~\eqref{eq:app_res_bound} onto a bound over the MPS-preparation error, and ultimately to derive the block-size prescription (\cref{eq:scaling-q} of the main text) guaranteeing $\varepsilon$ state-preparation error for~\cref{protocol:rg-approx-prep}.
Errors between states are quantified by the trace distance. In particular, as discussed in the main text, we have that
\begin{equation}\label{def:td_appendices}
D_{\rm tr}(\ket{P^{(\infty)}(L')},\ket{P^{(q)}(L')})
= \sqrt{1-|\braket{P^{(\infty)}(L')\,|\,P^{(q)}(L')}|^2}.
\end{equation}
Given that the trace distance is invariant under isometries,~\cref{def:td_appendices} is also the trace distance between the state prepared by~\cref{protocol:rg-approx-prep} and the targeted state.
In what follows our goal is thus to bound the overlap \begin{equation}
    |\braket{P^{(\infty)}(L')\,|\,P^{(q)}(L')}|.
\end{equation}
We note that a similar treatment appears in Ref.~\cite{piroli2021}. However, our definition of the residual $R^{(q)}$~\eqref{def:residual_app} differs slightly from the corresponding quantity used there. Moreover, because we care about dependences beyond the system size (in particular the prefactors depending on $\chi$ and $\Gamma$) we give a self-contained proof based on the residual bound above.
Throughout, we work in the normal case but the non-normal case can be treated similarly by applying the same steps within each block of the MPS transfer matrix, 
as we discuss in~\cref{app:non-normal}.

We start by rewriting the desired overlap as a trace. For that, define the operator
\begin{equation}\label{eq:mixed-transfer}
  E^{q,\infty} := \sum_{\mu=1}^{\chi^2} [P^{(q)}]^\mu \otimes \overline{[P^{(\infty)}]^\mu} \in \mathbb{C}^{\chi^2 \times \chi^2}.
\end{equation}
It follows that the overlap can be recast as
\begin{equation}\label{eq:overlap-trace}
  \braket{P^{(\infty)}(L')|P^{(q)}(L')} = \tr\bigl((E^{q,\infty})^{L'}\bigr).
\end{equation}
At $q = \infty$, the operator $E^{q,\infty}$ reduces to the operator $E^\infty = \dket{\sigma^2}\dbra{\mathbb{I}_\chi}$ encountered in~\cref{eq:Einfty-rank-one}. We further note that since $\dbraket{\mathbb{I}_\chi}{\sigma^2}=\tr(\sigma^2)=1$, this is a rank-one projector with a single eigenvalue $1$, and its non-zero eigenspace is spanned by the vector $\dket{\sigma^2}$, while its zero eigenspace is the kernel space $\ker(\dbra{\mathbb{I}_\chi})$. To analyse the spectrum of $E^{q,\infty}$, we consider perturbations of the spectrum $\operatorname{spec}(E^\infty) = (1, 0, \ldots, 0)$, and define such perturbation as
\begin{equation}\label{eq:Delta-decomp}
  \Delta^{q,\infty} := E^{q,\infty} - E^\infty = \sum_{\mu=1}^{\chi^2} [R^{(q)}]^\mu \otimes \overline{[P^{(\infty)}]^\mu}\in \mathbb{C}^{\chi^2 \times \chi^2}.
\end{equation}
Using the inequality between the spectral and Frobenius norms, alternative matrix views~\eqref{eq:frob_view} and submultiplicativity of the Frobenius norm, the spectral norm of the perturbation is bounded through
\begin{equation}\label{eq:Delta-q-bound}
  \|\Delta^{q,\infty}\|_\infty \le \|\Delta^{q,\infty}\|_2 =  \|R^{(q)}P^{(\infty)}\|_2\le \|R^{(q)}\|_2\,\|P^{(\infty)}\|_2 = O\!\bigl(\sqrt{\chi}\Gamma\, e^{-q/\xi}\bigr).
\end{equation}
The last step made use of the residual bound~\cref{eq:app_res_bound} previously derived together with the fact that $\|P^{(\infty)}\|_2 = \sqrt{\chi}$ following from $\tr(\sigma\sigma^\dag)= 1$.

The spectrum of $E^{q,\infty} = E^\infty + \Delta^{q,\infty}$ is then controlled by the Bauer-Fike theorem.
\begin{lemma}[Bauer-Fike~\cite{bauer1960}]\label{lem:bauer-fike}
Let $P \in \mathbb{C}^{n\times n}$ be diagonalisable such that $P = V \Lambda V^{-1}$ with $\Lambda = \mathrm{diag}(\Lambda_1, \ldots, \Lambda_n)$ and $V$ the matrix of eigenvectors. Let $A \in \mathbb{C}^{n\times n}$ be a perturbation.
Then every eigenvalue $\lambda$ of $P + A$ satisfies
\begin{equation}\label{eq:bauer-fike}
  \min_{j} |\lambda - \Lambda_j| \le \kappa(V)\,\|A\|_\infty,
\end{equation}
where $\kappa(V) := \|V\|_\infty\,\|V^{-1}\|_\infty$ is the condition number of the eigenvector matrix.
\end{lemma}
Applying this lemma with $P = E^\infty$ and $A = \Delta^{q,\infty}$, and denoting as $V$ a matrix of eigenvectors of $E^\infty$, confines each eigenvalue of $E^{q,\infty}$ to within a distance $\kappa(V)\,\|\Delta^{q,\infty}\|_\infty$ from $\textrm{spec}(E^{\infty})=\{0, 1\}$. To progress further, we need to bound $\kappa(V)$. For an orthogonal projector this condition number would be one; however, a little more care is required as $E^{\infty}$ is not orthogonal.

Let us organise the matrix of eigenvectors as $V = (u, W)$ with $\dket{u} := \dket{\sigma^2}/\|\dket{\sigma^2}\|_2$ the normalised eigenvector corresponding to the eigenvalue $1$ and $W$ an orthonormal basis of eigenvectors (arranged column-wise) corresponding to the $0$ eigenvalue.
There is freedom in the choice of $W$, and we take the first vector (denoted $\dket{w_1}$) to be the orthogonal projection of $\dket{u}$ onto the zero eigensubspace (taking the convention that the overlap $\alpha=\dbraket{w_1}{u}>0$), such that $\dket{u}$ is orthogonal to the remaining vectors forming the basis of $W$. In turn, the Gram matrix $V^\dag V$ has a block-diagonal structure with a $2 \times 2$ upper-left matrix 
\begin{equation*}
  \begin{pmatrix}
    1 & \alpha \\
    \alpha & 1
  \end{pmatrix},
\end{equation*}
and the bottom-right part the identity $\mathbb{I}_{\chi^2-2}$. 
As discussed earlier, the zero eigensubspace is $\ker(\dbra{\mathbb{I}_\chi})$ and we can evaluate the overlap as
\begin{equation*}
  \alpha = \sqrt{1-m^2},
  \qquad
  m := \frac{\dbraket{\mathbb{I}_\chi}{\sigma^2}}{\|\dket{\sigma^2}\|\,\|\dket{\mathbb{I}_\chi}\|},
\end{equation*}
where $m$ is the normalised overlap between $\dket{\sigma^2}$ and $\dket{\mathbb{I}_\chi}$. Since $\dbraket{\mathbb{I}_\chi}{\sigma^2} = 1$, $\|\dket{\mathbb{I}_\chi}\| = \sqrt{\chi}$ and $\|\dket{\sigma^2}\|\leq 1$ we have 
\begin{equation}
    \frac{1}{m} = O(\sqrt{\chi}).
\end{equation}
Given the structure of $V^\dag V$, its largest eigenvalue is $1+\alpha$ such that $\|V\|_{\infty} = \sqrt{1+\alpha}$, and similarly $\|V^{-1}\|_{\infty} = 1/\sqrt{1-\alpha}$.
We can thus bound the condition number through
\begin{equation*}
  \kappa(V)
  = \sqrt{\frac{1+\alpha}{1-\alpha}}
  = \frac{1+\sqrt{1-m^2}}{m}
  = O(1/m) = O(\sqrt{\chi}).
\end{equation*}
    
Defining $\Gamma^{\prime\prime}:= \sqrt{\chi} \Gamma$, we obtain from~\cref{eq:Delta-q-bound}
\begin{equation}\label{eq:bauer-fike-radius}
  \delta := \kappa(V)\,\|\Delta^{q,\infty}\|_\infty = O\!\bigl(\Gamma^{\prime\prime}\,e^{-q/\xi}\bigr).
\end{equation}
For $q$ large enough such that $\delta \ll 1$, 
and by the continuity of eigenvalues of $E^\infty + t\Delta^{q,\infty}$ for $t$ varied from $0$ to $1$ (from $E^\infty$ to $E^{q,\infty}$), we see that the leading eigenvalue has to stay within $O(\delta)$ of $1$ while the remaining $\chi^2-1$ eigenvalues stay within $O(\delta)$ of $0$.
In other words, we can bound each eigenvalue through
\begin{equation}\label{eq:lambda1-perturb}
  \begin{cases}
    |\lambda_1(E^{q,\infty}) - 1| &\le \delta, \\
    |\lambda_j(E^{q,\infty})| &\le \delta \quad (j \ge 2),
  \end{cases}
\end{equation}
both of order $O(\Gamma^{\prime\prime}\,e^{-q/\xi})$ as per~\cref{eq:bauer-fike-radius}.
Expanding the trace by eigenvalues, we get
\begin{equation}\label{eq:trace-expansion}
  \tr\bigl((E^{q,\infty})^{L'}\bigr) = \lambda_1(E^{q,\infty})^{L'} + \sum_{j=2}^{\chi^2} \lambda_j(E^{q,\infty})^{L'}.
\end{equation}
In the regime where 
$L'\,\Gamma^{\prime\prime}\,e^{-q/\xi} \ll 1$
we see that \cref{eq:lambda1-perturb} gives
\begin{equation}\label{eq:top-eigval-expansion}
  \begin{cases}
    |\lambda_1(E^{q,\infty})^{L'} - 1| = O(L'\,\Gamma^{\prime\prime}\,e^{-q/\xi}), \\
    \Bigl|\sum_{j=2}^{\chi^2}\lambda_j(E^{q,\infty})^{L'}\Bigr| \le (\chi^2 - 1)\,\delta^{L'}.
  \end{cases}
\end{equation}
Given that $\delta=O(\Gamma^{\prime\prime}e^{-q/\xi})\ll1$, the second term is no larger than the leading $O(L'\delta)$ correction.
Therefore
\begin{equation}\label{eq:Tr-Eq-asymp}
  |\tr\bigl((E^{q,\infty})^{L'}\bigr) -1| = O\!\bigl(L'\,\Gamma^{\prime\prime}\,e^{-q/\xi}\bigr).
\end{equation}
Then, combining \cref{def:td_appendices,eq:overlap-trace,eq:Tr-Eq-asymp}, we obtain
\begin{equation}\label{eq:Dtr-from-overlap}
  D_{\rm tr}\bigl(\ket{P^{(q)}(L')},\,\ket{P^{(\infty)}(L')}\bigr) = O\!\left(\sqrt{L'\,\Gamma^{\prime\prime}\,e^{-q/\xi}}\right),
\end{equation}
Hence, a sufficient condition to ensure that $D_{\rm tr} \le \varepsilon$ is
\begin{equation}\label{eq:scaling-q-app}
q=O\left(\xi\log\left(\frac{\Gamma^{\prime\prime} L'}{\varepsilon^2}\right)\right) = O\left(\xi\log\left(\frac{\Gamma L}{\varepsilon}\right)\right)
\end{equation}
which is equivalent, up to constants inside the logarithm, to~\cref{eq:scaling-q} in the main text.
Note that to get the RHS in \cref{eq:scaling-q-app}, we replaced (i) $L'$ by $L$ as $L'=O(L)$, (ii) $\Gamma^{\prime\prime}$ by $\Gamma$ since they only differ by a polynomial in $\chi$ and $\Gamma = O(\chi)$~\eqref{eq:app_res_bound}, and (iii) $\varepsilon^2$ by $\varepsilon$ inside the big-$O$ logarithm.

\section{Block-encoding of the correction map}
\label{app:block-encoding}

In this appendix, we recall how a general linear map can be implemented on a quantum computer by embedding it into a larger unitary and post-selecting on an ancillary register~\cite{nibbi2024}.
We then apply this construction to the correction map $C^{(q)}$ used in~\cref{procedure:prob-mapping}.
Here $C^{(q)}$ denotes the linear operator on the virtual-pair space that maps the fixed-point operator $P^{(\infty)}$ to $P^{(q)}$, while $U_{C^{(q)}}$ denotes a unitary block-encoding of a rescaled version of that operator.

Consider an operator $B \in \mathbb{C}^{n \times n}$ with spectral norm $\|B\|_\infty \leq 1$, that is, $B$ is a contraction.
Assume that we have a unitary $U_B \in \mathbb{C}^{2n \times 2n}$ acting on the system and a single ancilla qubit, such that
\begin{equation}\label{eq:block-encoding-def}
  (\bra{0}_{\mathrm{anc}} \otimes \mathbb{I}_n)\, U_B\, (\ket{0}_{\mathrm{anc}} \otimes \mathbb{I}_n) = B.
\end{equation}

In matrix form, $B$ appears as the top-left $n \times n$ block of the $2n \times 2n$ unitary $U_B$.
For any input state $\ket{\psi_{\mathrm{in}}} \in \mathbb{C}^{n}$, applying $U_B$ to $\ket{0}_{\mathrm{anc}} \otimes \ket{\psi_{\mathrm{in}}}$ and post-selecting the ancilla onto $\ket{0}_{\mathrm{anc}}$ prepares the state
\begin{equation}\label{eq:block-encoding-state}
  \frac{B \ket{\psi_{\mathrm{in}}}}{\sqrt{p_{\mathrm{succ}}}},
\end{equation}
with success probability
\begin{equation}\label{eq:block-encoding-prob}
  p_{\mathrm{succ}} = \| B\ket{\psi_{\mathrm{in}}} \|^2.
\end{equation}
Since $B$ is a contraction, we see that $p_{\mathrm{succ}} \in [0,1]$.

It remains to show that $U_B$ exists and to construct it explicitly.
Write $U_B$ in block form:
\begin{equation}\label{eq:block-encoding-UA}
  U_B =
  \begin{pmatrix}
    B & M_1 \\
    M_2 & M_3
  \end{pmatrix},
\end{equation}
where $M_1, M_2 \in \mathbb{C}^{n \times n}$ and $M_3 \in \mathbb{C}^{n \times n}$.
The unitarity condition $U_B^\dag U_B = U_B U_B^\dag = \mathbb{I}_{2n}$ imposes:
\begin{align}
  B^\dag B + M_2^\dag M_2 &= \mathbb{I}_n, \label{eq:uu-col1}\\
  B B^\dag + M_1 M_1^\dag &= \mathbb{I}_n, \label{eq:uu-row1}\\
  B^\dag M_1 + M_2^\dag M_3 &= 0. \label{eq:uu-cross}
\end{align}
From~\eqref{eq:uu-col1}, $M_2^\dag M_2 = \mathbb{I}_n - B^\dag B \succeq 0$, so we can take
\begin{equation}\label{eq:M2-def}
  M_2 = \sqrt{\mathbb{I}_n - B^\dag B}.
\end{equation}
Similarly, from~\eqref{eq:uu-row1},
\begin{equation}\label{eq:M1-def}
  M_1 = \sqrt{\mathbb{I}_n - B B^\dag}.
\end{equation}
Note that both square roots are well-defined since $B$ is a contraction.
Finally, from~\eqref{eq:uu-cross},
\begin{equation}\label{eq:M3-def}
  M_3 = -M_2^{+}\, B^\dag\, M_1,
\end{equation}
where $M_2^{+}$ denotes the Moore-Penrose pseudo-inverse of $M_2$.
One can verify that this choice satisfies all unitarity conditions.
To apply this construction to the correction map, the case $\|B\|_\infty > 1$ must also be handled.
We rescale to $B/\alpha$ with $\alpha \geq \|B\|_\infty$ before embedding in a unitary $U_{B/\alpha}$. The state obtained after post-selection is the same as before~\eqref{eq:block-encoding-prob}
but with success probability
\begin{equation}\label{eq:block-encoding-prob-rescaled}
  p_{\mathrm{succ}} = \frac{1}{\alpha^2}\| B\ket{\psi_{\mathrm{in}}} \|^2,
\end{equation}
that depends both on the rescaling factor $\alpha$ and the norm $\|B \ket{\psi_{\rm in}}\|$.
In~\cref{procedure:prob-mapping}, we apply this construction to $B = C^{(q)}$ and $\alpha =\|C^{(q)}\|_\infty$.
The resulting unitary $U_{C^{(q)}/\alpha}$ is the block-encoding used in~\cref{eq:block-encoding-Cq}, and is reused for every block-encoded correction throughout the protocols.

\section{Qubit count and ancilla usage}
\label{app:qubit-count}

In this appendix, we show that (asymptotically) all the ancilla qubits required by our circuits, in particular in the context of \cref{protocol:exact-prep}, can be reused. Concretely, we show that the total number of qubits, needed at any point when constructing the circuits preparing the desired MPS, never exceeds the number of qubits $N_{\mathrm{phys}} = L\lceil\log_2 d\rceil$ that is required to store the final state $\ket{A(L)} \in \mathbb{C}^{d^L}$.
As before, $L' = L/q$ is the number of renormalised blocks such that for each renormalised block we shall not use more than the \emph{per-block} qubit budget $q\lceil\log_2 d\rceil$ qubits. 

The MPS preparation protocols start from the RG fixed-point state encoded in $N_{\mathrm{virt}} = L'\lceil\log_2(\chi^2)\rceil$ qubits (or $\lceil \log_2(\chi^2) \rceil$ qubits per renormalised block) and can be followed by:
\begin{itemize}
  \item \textit{Block-encoding of the correction map (\cref{procedure:prob-mapping}).}
    The unitary block-encoding of $C^{(q)}/\alpha$ requires one ancilla qubit per block, for a total of $L'$ ancillas.
    The per-block qubit count after this step is $\lceil\log_2(\chi^2)\rceil + 1$.

  \item \textit{Deterministic preparation via amplitude amplification (\cref{thm:det-aa}).}
    Making the correction step fully deterministic requires generalised reflections around the $\ket{0}^{\otimes L'}$ ancilla subspace.
    Implementing these reflections in depth $O(\log L')$ requires one additional qubit per block. The per-block qubit count after this step is $\lceil\log_2(\chi^2)\rceil + 2$.

  \item \textit{Sequential implementation with OAAI (\cref{thm:Vq-OAAI}).}
    The decomposition of~\cref{eq:Vq_decomp_chain_2} applies $(C^{(q)})^{-1}$ and $C^{(q_h)}$ in sequence, each contributing to one block-encoding ancilla per block before $V^{(q_h)}$ acts.
    The per-block qubit count for this step is thus $\lceil\log_2(\chi^2)\rceil + 2$.

  \item \textit{Renormalised implementation with OAAI (\cref{thm:Vq-RG-OAAI}).}
    The decomposition of~\cref{eq:Vq_rg} first applies $(C^{(q)})^{-1}$ and then splits each renormalised $q$-site block into $q/q_2$ smaller $q_2$-site blocks, on which the correction maps $C^{(q_2)}$ and the corresponding isometries $V^{(q_2)}$ act in parallel.
    Each smaller block is implemented using the sequential OAAI construction of~\cref{thm:Vq-OAAI}, and therefore requires $\lceil\log_2(\chi^2)\rceil + 2$ qubits: the $\chi^2$-dimensional virtual register and two ancillary qubits for the block-encoded corrections.
    Equivalently, this step uses $(q/q_2)(\lceil\log_2(\chi^2)\rceil + 2)$ qubits across the $q/q_2$ small blocks inside a renormalised $q$-site block.
\end{itemize}

Note that the first two items occur during the preparation of the renormalised state $\ket{P^{(q)}(L')}$, after which the ancillas are reset
and become available for the implementation of the isometries ($V^{(q)}$, $V^{(q_h)}$, or $V^{(q_2)}$) discussed in the two last items.
For the sequential constructions, the condition for the reuse of ancilla qubits is $q\lceil\log_2 d\rceil \ge \lceil\log_2(\chi^2)\rceil + 2$. For the renormalised ones this condition is rather $q_2\lceil\log_2 d\rceil \ge \lceil\log_2(\chi^2)\rceil + 2$.
Hence, at fixed $d$ and $\chi$ these conditions are always satisfied for large enough $q$ and $q_2$ (or equivalently, large enough $L$),
and no additional qubits beyond $N_{\mathrm{phys}}$ are required at any stage of the corresponding protocols.

\section{Non-normal MPS}
\label{app:non-normal}

This appendix details how the normal-MPS circuit constructions can be extended to non-normal MPS~\cite{cirac_2017,cirac2021mpspeps}.
Recall the block-diagonal canonical form~\eqref{eq:canonical-form} of an MPS from~\cref{app:normal_mps},
\begin{equation*}
  A^i = \bigoplus_{b=1}^{g} \nu_b\, A^{i}_b,
\end{equation*}
with blocks indexed by $b=1,\dots,g$ ($g>1$ in the non-normal case), and where the representatives $A_b$ are chosen to be inequivalent in the sense of Ref.~\cite{cirac_2017}.
Each $A_b^i \in \mathbb{C}^{\chi_b\times \chi_b}$ is a normal tensor with $\chi = \sum_{b=1}^{g} \chi_b$, and $\nu_b\in\mathbb{C}$.
In what follows, every block-indexed quantity (such as $A_b^{(q)}$, $P_b^{(q)}$, $P_b^{(\infty)}$, $\rho_b$, $\sigma_b$, $\Gamma_b$, $\xi_b$) denotes the corresponding object used in~\cref{app:normal_mps} restricted to the $b$-th normal block.
Furthermore, to keep track of the various dimensions involved, we decompose the full virtual-pair space as
\begin{equation}
  \mathcal H_{\rm full}:=
  \mathbb{C}^{\chi}\otimes\mathbb{C}^{\chi}
  =
  \mathcal H_{\rm diag}\oplus\mathcal H_{\rm off},
\end{equation}
where
\begin{equation}
  \mathcal H_{\rm diag}:=\bigoplus_{b=1}^{g}\mathcal H_b,
  \quad
  \mathcal H_b:=\mathbb{C}^{\chi_b}\otimes\mathbb{C}^{\chi_b},
  \quad
  \mathcal H_{\rm off}:=
  \bigoplus_{b\neq b'}
  \left(\mathbb{C}^{\chi_b}\otimes\mathbb{C}^{\chi_{b'}}\right).
\end{equation}
We denote by $\Pi_{\rm off}$ the projector onto $\mathcal H_{\rm off}$.

Since the shared block structure of the $A^i$ is preserved under products, the $L$-site MPS decomposes as
\begin{equation}
  \ket{A(L)} = \sum_{b=1}^{g} \nu_b^{L}\, \ket{A_b(L)},
  \label{eq:non-normal-state-b}
\end{equation}
up to normalisation of the state.
For each representative block $b$, the tensor blocked  over $q$ sites, $A_b^{(q)} \in \mathbb{C}^{d^q \times \chi^2}$, is defined through its entries as
\begin{equation}
  \bigl[A_b^{(q)}\bigr]^{\mathbf{i}}_{\alpha,\beta}
  := \bigl[A_b^{i_1}A_b^{i_2}\cdots A_b^{i_q}\bigr]_{\alpha,\beta}.
\end{equation}
Blocking preserves the block-diagonal structure of $A^i$ as
\begin{equation}
  [A^{(q)}]^{\mathbf{i}} = \bigoplus_{b=1}^{g} \nu_b^q\, [A_b^{(q)}]^{\mathbf{i}}.
  \label{eq:Aq-nonnormal-b}
\end{equation}
We can then apply the polar decomposition to the matrix $A^{(q)}$ to obtain
\begin{equation}
  A^{(q)} = V^{(q)} P^{(q)},
\end{equation}
with $P^{(q)} = \sqrt{(A^{(q)})^\dag A^{(q)}} \in \mathbb{C}^{\chi^2 \times \chi^2}$ 
a positive semidefinite operator on $\mathcal H_{\rm full}$,
and $V^{(q)} := A^{(q)}(P^{(q)})^+ \in \mathbb{C}^{d^q \times \chi^2}$ the corresponding isometry, where $(\cdot)^+$ denotes the pseudo-inverse.
Since $A^{(q)}$ is block diagonal on the virtual-pair space, the same is true for each renormalised physical-index slice of $P^{(q)}$: 
For each renormalised physical index $\mu=1,\dots,\chi^2$,
the slice $[P^{(q)}]^\mu \in \mathbb{C}^{\chi\times\chi}$ decomposes as
\begin{equation}
  [P^{(q)}]^{\mu} = \bigoplus_{b=1}^{g} |\nu_b|^q \, [P_b^{(q)}]^{\mu},
  \label{eq:non-normal-polar-b}
\end{equation}
where $[P_b^{(q)}]^\mu \in \mathbb{C}^{\chi_b\times\chi_b}$. 
The corresponding operator $P_b^{(q)} \in \mathbb{C}^{\chi^2\times\chi_b^2}$ is the $b$-th block column of $P^{(q)}$ after factoring out $|\nu_b|^q$, mapping $\mathcal H_b$ to $\mathcal H_{\rm full}$ and non-square whenever $g>1$.
Since distinct blocks never mix under multiplication of $A^i$, the operator $P^{(q)}$ vanishes on $\mathcal H_{\rm off}$ and is supported on $\mathcal H_{\rm diag}$, which has dimension $\sum_b\chi_b^2$.
In particular, $P^{(q)}$ has at least $\chi^2-\sum_b\chi_b^2$ zero singular values.
Within each block, the normal-case analysis~\eqref{eq:fixed_point_one} yields a density matrix $\rho_b \in \mathbb{C}^{\chi_b\times \chi_b}$ satisfying
\begin{equation}
  \sum_i A_b^i \rho_b A_b^{i\dag} = \rho_b,
\end{equation}
together with $\tr(\rho_b)=1$ and $\rho_b \succ 0$.
As before, we define $\sigma_b = \sqrt{\rho_b} \in \mathbb{C}^{\chi_b\times \chi_b}$.
As in the normal case, after rewriting the $b$-th block column with explicit virtual indices, after defining
\begin{align}
  P_b^{(\infty)} := \sigma_b \otimes \mathbb{I}_{\chi_b} \quad \textrm{and} \quad
  R_b^{(q)}  := P_b^{(q)} - P_b^{(\infty)},
  \label{eq:Pinf-entrywise-b}
\end{align}
one has
\begin{align}
  \|R_b^{(q)}\|_2 &\le \Gamma_b\, e^{-q/\xi_b},
\end{align}
where $\Gamma_b$ is a constant, analogous to $\Gamma$ in~\cref{eq:app_res_bound} for the $b$-th block,
and $\xi_b$ is the correlation length, analogous to $\xi$ in~\cref{eq:correll} for the $b$-th block.
To match the shape of $P_b^{(q)} \in \mathbb{C}^{\chi^2 \times \chi_b^2}$ in~\cref{eq:Pinf-entrywise-b}, $P_b^{(\infty)} \in \mathbb{C}^{\chi_b^2 \times \chi_b^2}$ is embedded as a map $\mathcal H_b\to\mathcal H_{\rm full}$ by padding the row index with zeros.
We now define the reference-point operator $P_{\mathrm{ref}}^{(q)}$ on $\mathcal H_{\rm full}$, the non-normal generalisation of the fixed-point operator $P^{(\infty)}$, by assembling the fixed points of the different blocks: for each renormalised physical index $\mu$,
\begin{equation}
  [P_{\mathrm{ref}}^{(q)}]^{\mu}
  = \bigoplus_{b=1}^{g} |\nu_b|^q\, [P_b^{(\infty)}]^{\mu}.
  \label{eq:Pref-nonnormal-b}
\end{equation}
The corresponding reference-point state on $L'$ renormalised sites decomposes as
\begin{equation}
  \ket{P_{\mathrm{ref}}^{(q)}(L')}
  = \sum_{b=1}^{g} c_b(L)\, \ket{P_b^{(\infty)}(L')},
\end{equation}
where $L=qL'$ and $c_b(L):=|\nu_b|^L/\sqrt{\sum_k|\nu_k|^{2L}}$, together with
\begin{equation}
  \ket{P_b^{(\infty)}(L')} = \ket{\sigma_b}^{\otimes L'}.
\end{equation}
Different representative blocks occupy orthogonal subspaces $\mathcal H_b$ of the renormalised physical (virtual-pair) space, such that $\braket{\sigma_b | \sigma_k} = \tr(\sigma_b^\dag \sigma_k) = \delta_{bk}$.
With this orthogonality in hand, the reference-point state can be prepared in constant depth using MCM-FF as discussed in~\cite{piroli2021,briegel2001persistent}.

As for the normal case, the correction map is defined as the map from the reference-point state $\ket{P^{(\infty)}(L')}$ to $\ket{P^{(q)}(L')}$.
Let $\iota_b:\mathcal H_b\to\mathcal H_{\rm full}$ be an inclusion map (a map that embeds elements of $\mathcal H_b$ into the corresponding block of $\mathcal H_{\rm full}$).
Then, its adjoint $\iota_b^\dag:\mathcal H_{\rm full}\to\mathcal H_b$ is the projector onto $\mathcal H_b$.
Each $\sigma_b$ is strictly positive, so $P_b^{(\infty)}$ is invertible on $\mathcal H_b$, and we define the block-wise correction map as
\begin{equation}
  C_b^{(q)} := P_b^{(q)}\bigl(P_b^{(\infty)}\bigr)^{-1}
  = \iota_b + R_b^{(q)}\bigl(P_b^{(\infty)}\bigr)^{-1} \in \mathbb{C}^{\chi^2 \times \chi_b^2},
  \label{eq:Cq-block-b}
\end{equation}
where $(P_b^{(\infty)})^{-1}$ denotes the inverse of $P_b^{(\infty)}$ restricted to $\mathcal H_b$.
Thus $C_b^{(q)}:\mathcal H_b\to\mathcal H_{\rm full}$.
The full correction map is then the operator on $\mathcal H_{\rm full}$
\begin{equation}
  C^{(q)} :=
  \sum_{b=1}^{g} C_b^{(q)} \iota_b^\dag
  \in \mathbb{C}^{\chi^2 \times \chi^2}.
  \label{eq:Cq-block-diag-b}
\end{equation}
Equivalently,
\begin{equation}
  C^{(q)}
  =
  \mathbb{I}_{\chi^2}
  +
  \sum_{b=1}^{g}
  R_b^{(q)}\bigl(P_b^{(\infty)}\bigr)^{-1}\iota_b^\dag .
\end{equation}
Using the triangle inequality, we obtain
\begin{equation}
  \|C^{(q)}\|_\infty
  \leq 1+\sum_{b=1}^g
  \Gamma_b\,\|(P_b^{(\infty)})^{-1}\|_\infty\,e^{-q/\xi_b}
  \leq 1+\Gamma e^{-q/\xi},
\end{equation}
where we have defined $\xi:=\max_b\xi_b$ and $\Gamma:=\sum_b\Gamma_b\,\|(P_b^{(\infty)})^{-1}\|_\infty$.
With these, the same block-encoding and post-selection constructions as in the normal case apply, with success probability controlled by $\|C^{(q)}\|_\infty$, which approaches one exponentially in $q$. That is, the results for the preparation of $\ket{P^{(q)}(L')}$ (\cref{sec:exact-prep}) can be ported to the non-normal TI-MPS case.

Now we turn our focus to the implementation of the isometries $V^{(q)}$.
For the sequential approach of \cref{subsec:vq_seq},
we can simply replicate the normal-case construction, without considering the block structure.
The depth scaling of the circuits is therefore the same:
\begin{equation}
  D_V^{\rm seq, NN}(q) = O(q\chi^4).
\end{equation}

For the \seqAA approach of \cref{subsec:SEQ_OAAI}, the 
factorization of \eqref{eq:Pq-inv-decomp} extends per block. Since $P^{(q)} = C^{(q)} P_{\mathrm{ref}}^{(q)}$ with both factors invertible on $\mathcal H_{\mathrm{diag}}$, the pseudo-inverse of the product reverses their order~\cite{greville1966}, giving
\begin{equation}
  P^{(q)+} = (P_{\mathrm{ref}}^{(q)})^+\,(C^{(q)})^+,
  \label{eq:Pq-inv-decomp-b}
\end{equation}
where $(P_{\mathrm{ref}}^{(q)})^+$ is block-diagonal:
\begin{equation}
  [(P_{\mathrm{ref}}^{(q)})^+]^i
  = \bigoplus_{b=1}^{g} |\nu_b|^{-q}\, [\sigma_b^{-1}\otimes \mathbb{I}_{\chi_b}]^{i}.
  \label{eq:Pref-inv-b}
\end{equation}
Blocking the first $q_h$ sites on the left and applying the same polar
decomposition as in the normal case, as seen in \cref{eq:Vq_decomp_chain_2}, we get
\begin{equation}
  V^{(q)}
  = V^{(q_h)}\cdot C^{(q_h)}\cdot A_{\mathrm{phase}}^{(q-q_h)}\cdot \bigl(C^{(q)}\bigr)^{-1},
  \label{eq:Vq-cancelled-NN}
\end{equation}
Here $(C^{(q)})^{-1}$ is the inverse of the 
correction map defined in \cref{eq:Cq-block-diag-b}, while  the chain of tensors $A_{\mathrm{phase}}^{(q-q_h)}$
is defined as
\begin{equation}
  A_{\mathrm{phase}}^{(q-q_h)}
  := \bigoplus_{b=1}^{g} e^{i(q-q_h)\theta_b}\, A_b^{(q-q_h)},
  \qquad e^{i\theta_b} := \nu_b/|\nu_b|.
  \label{eq:A-phase-NN}
\end{equation}
Since $|e^{i\theta_b}|=1$ and each $A_b$ is right-canonical, each of the $(q-q_h)$ tensors of
$A_{\mathrm{phase}}^{(q-q_h)}$ is again a right-canonical isometry, so the chain admits a sequential 
implementation using circuit depth $O((d\chi)^2)$ per tensor, similar to the
normal case.
As for the normal case, applying the exact OAAI then deterministically implements
$V^{(q)}$ and requires $O(1)$ repetitions. The resulting depth is thus
\begin{equation}
  D_V^{\seqAA,\mathrm{NN}}(q,q_h) = O(q\chi^2) + O(q_h\chi^4),
  \label{eq:depth-Vq-seqAA-NN}
\end{equation}
structurally identical to the normal-case formula provided in \eqref{eq:vq_seq_oaai_partial}:
the $O(q\chi^2)$ contribution is the depth of the sequential implementation of 
$A_{\mathrm{phase}}^{(q-q_h)}$ 
while the $O(q_h\chi^4)$ corresponds to the sequential implementation of $V^{(q_h)}$, together with 
 the two correction maps implemented by block-encodings. 

Finally, for the \renAA approach of \cref{subsec:ren-OAAI}, 
the product state $\ket{\sigma}^{\otimes(q/q_2 - 1)}$ of the normal case is replaced by the coherent superposition $\sum_b c_b(L)\, \ket{\sigma_b}^{\otimes(q/q_2 - 1)}$ that can be implemented in depth $O(\chi^2)$ using MCM-FF, as mentioned earlier in this appendix.
The depth scaling of the circuit is thus the same as \eqref{eq:depth-Vq-OAAI-ren}: 
\begin{equation}
  D_V^{\renAA, NN} = O(\chi^2\, \xi \log(\Gamma q)) + O(\chi^4\, \xi \log(\Gamma)).
\end{equation}
Overall, this shows that, albeit more demanding to deal with, the circuit depth of the circuits that exactly prepare the TI-MPS remains similar to the normal case.

\section{Preparability of non-TI MPS}
\label{app:nonTI}

In this section, we prove the logarithmic block-size scaling stated in~\cref{thm:random-nonTI} for IID non-TI MPS.

\subsection{Preliminary}
As in the main text, we consider a non-TI MPS with open boundary conditions and uniform bond dimension $\chi$:
\begin{equation}
  \ket{\tilde A(L)}
  =
  \raisebox{-0.25\height}{\hspace{1em}\includegraphics[width=0.3\textwidth]{figs/NTIMPS.pdf}\hspace{3pt},  \\[6pt] }
  \label{eq:nonTI_MPS_def_app}
\end{equation}
specified by site-dependent tensors $A^{[k]}\in\mathbb{C}^{\chi\times d\times \chi}$ (padding the first and last tensor).
These are assumed to be in left-canonical form, so for all $k$,
\begin{equation}\label{eq:nonti_leftcan}
  \sum_{s=1}^d [A^{[k]}]^{s\dag} [A^{[k]}]^s=\mathbb{I}_{\chi}.
\end{equation}
For an interval $B=[a,b]\subseteq\{1,\ldots,L\}$, define its length as
$q_B=b-a+1$ and let $A^{[a,b]}$ denote the blocking of all sites in $B$.
Using the same notations, its polar decomposition is given by
\begin{equation}\label{eq:polar_block}
  A^{[a,b]} = V^{[a,b]} P^{[a,b]},
\end{equation}
with $P^{[a,b]}$ positive semidefinite and $V^{[a,b]}$ a partial isometry.

For each site $k$, define the map
\begin{equation}
  \phi_k(X)
  := \sum_{s=1}^{d} A^{[k],s}\,X\,A^{[k],s\dag},
  \qquad X\in\mathbb{C}^{\chi\times\chi}.
  \label{eq:phi-k-random-nonTI}
\end{equation}
Given~\eqref{eq:nonti_leftcan}, we see that each $\phi_k$ is a quantum channel.
For $a\le b$, define the composition of maps
\begin{equation}
  \Phi_{a,b}
  := \phi_{a}\circ\phi_{a+1}\circ\cdots\circ\phi_{b}.
  \label{eq:block-channel-random-nonTI}
\end{equation}
Under vectorisation, operators are treated as vectors and channels are represented by matrices. Given a channel $\Lambda$ we denote its representative matrix as $\hat{\Lambda}$.
In particular, taking the convention that $\ket{i}\bra{j} \mapsto \ket{i}\otimes \ket{j}$, we have that the channel $\phi_k$~\eqref{eq:phi-k-random-nonTI} is represented by the MPS transfer matrix of $A^{[k]}$, namely
\begin{equation}
  E_k := \sum_{s=1}^{d} A^{[k],s}\otimes \overline{A^{[k],s}} \in \mathbb{C}^{\chi^2 \times \chi^2}.
\end{equation}
That is, $\hat{\phi}_k = E_k$ and $\hat{\Phi}_{a,b} = E_a E_{a+1}\cdots E_b=:E_{a,b}$.
This identifies the channel composition used in Ref.~\cite{Movassagh} with the MPS transfer matrix of a block of tensors.
When a physical index of the blocked positive factor is fixed, the same block MPS transfer matrix can also be expressed as
\begin{equation}\label{eq:polar-block-transfer}
    E_{a,b} = \sum_{\mu=1}^{\chi^2} [P^{[a,b]}]^\mu\otimes \overline{[P^{[a,b]}]^\mu} \in \mathbb{C}^{\chi^2 \times \chi^2}.
\end{equation}

We use two related norm conventions. For an operator $X$, $\|X\|_{p}$ denotes the Schatten $p$-norm, namely the $\ell_p$ norm of its singular values. In addition to the Frobenius norm $\|X\|_2$ and the spectral norm $\|X\|_\infty$ that we already encountered, we have the trace norm $\|X\|_1$. For $1\leq p \leq q < \infty$,
\begin{equation}
    \| X \|_q \leq \| X \|_p \leq \operatorname{rank}(X)^{1/p - 1/q} \| X \|_q.
\end{equation}
One specialization of this inequality was already encountered in~\cref{eq:ineq_schatten}. For $X \in \mathbb{C}^{D \times D}$, we also have
\begin{equation}\label{ineq_schatten_spec}
    \| X \|_2 \leq \| X \|_1 \leq \sqrt{D} \| X \|_2.
\end{equation}
For a linear map $\Lambda$ acting on operators, the induced Schatten $p$-norm is defined as
\begin{equation}
    \| \Lambda \|_{p \rightarrow p} := \max_{X: \|X\|_p=1} \|\Lambda(X)\|_{p},
\end{equation}
which is distinct from the Schatten norm $\|\hat{\Lambda}\|_p$ of the channel represented as a matrix. 
Still, they can sometimes be related. In particular, under our vectorisation convention we have that
\begin{equation}
    \|\Lambda\|_{2 \rightarrow 2} = \|\hat{\Lambda}\|_{\infty}.
\end{equation}
For a channel $\Lambda: \mathbb{C}^{D \times D} \rightarrow \mathbb{C}^{D \times D}$, making repeated use of \cref{ineq_schatten_spec} yields the bound
\begin{equation}\label{eq:app_bound_11}
\begin{aligned}
  \| \hat{\Lambda}\|_{\infty}
  &= \max_{X: \|X\|_2=1} \|\Lambda(X)\|_{2} \\
  &\leq \max_{X: \|X\|_1=\sqrt{D}} \|\Lambda(X)\|_{2} \\
  &\leq \max_{X: \|X\|_1=\sqrt{D}} \|\Lambda(X)\|_{1} \\
  &\leq \sqrt{D} \|\Lambda\|_{1 \rightarrow 1}.
\end{aligned}
\end{equation}

Finally, we highlight that, in what follows, we regard the tensors $A^{[k]}$, and thus the maps $\phi_k$~\eqref{eq:phi-k-random-nonTI}, as random variables drawn from a fixed distribution. When emphasising this randomness, we write the corresponding quantities as functions of a random realisation $\omega$.

\subsection{Preparation of random MPS}
\label{app:random-nonTI}

Theorem~1 of Ref.~\cite{Movassagh} holds for a sequence of \emph{ergodic} quantum channels $\phi_0,\phi_1,\ldots$ satisfying two additional technical assumptions (\emph{Assumptions A.1 and A.2} of Ref.~\cite{Movassagh}).
The first one requires that some finite product of channels is strictly positive with nonzero probability.
The second one excludes nonzero positive operators in the kernel of the dual channel.
They are expected to hold whenever the dynamics has non-negligible decoherence.
For example, in the IID Haar-channel model considered in Appendix~B of Ref.~\cite{Movassagh}, each channel is strictly positive almost surely, so both assumptions are satisfied almost surely.
Adapting Theorem~1 of Ref.~\cite{Movassagh} to our notation gives the following.

\begin{theorem}[Movassagh-Schenker, adapted notation]
\label{thm:movassagh-adapted}
Assume that the sequence $\{\phi_k\}$ is ergodic and satisfies the two technical assumptions stated above with respect to the ordering used in~\cref{eq:block-channel-random-nonTI}.
Then there exists a constant $\tilde{\xi}>0$ (akin to a correlation length), random prefactors $\kappa_a(\omega)<\infty$, and an ergodic sequence of positive-definite density matrices
$\rho_a(\omega)\succ0$ (almost surely) such that, for every density matrix $X$ and every interval $[a,b]$,
\begin{equation}
  \bigl\|\phi_a\circ\phi_{a+1}\circ\cdots\circ\phi_b(X)-\rho_a(\omega)\bigr\|_1
  \le \kappa_a(\omega)\,e^{-(b-a+1)/\tilde{\xi}}.
  \label{eq:movassagh-adapted-bound}
\end{equation}
\end{theorem}

\begin{proof}
This is Theorem~1 of Ref.~\cite{Movassagh} applied in the ordering of~\cref{eq:block-channel-random-nonTI}, after renaming $Z_a$ to $\rho_a$ and $C_a$ to $\kappa_a$, the latter to avoid overloading the symbol $C$ already used for the correction map,
and recasting $\mu \in (0,1)$ as $e^{-1/\tilde{\xi}}$ with $\tilde{\xi}>0$.
\end{proof}

Let us comment on the previous. 
First,~\cref{eq:movassagh-adapted-bound} holds for arbitrary density matrices $X$, and can be rewritten as a bound over distances between channels:
\begin{equation}
  \bigl\|\Phi_{a,b}(\cdot)-\rho_a(\omega)\operatorname{tr}(\cdot)\bigr\|_{1 \rightarrow 1}
  \le \kappa_a(\omega)\,e^{-(b-a+1)/\tilde{\xi}},
  \label{eq:movassagh-adapted-bound-channel}
\end{equation}
where the channel $X\mapsto \rho_a(\omega)\operatorname{tr}(X)$  is the replacement channel.
Second, for each random realisation of the channels, a different value of $\kappa_a(\omega)$ applies.
This value can be treated as fixed for that realisation, although in general it cannot be evaluated analytically.
Similarly, the distribution of $\rho_a(\omega)$ is generally unknown.
Thus, except in special ensembles, \cref{thm:movassagh-adapted} is existential.
For concrete computations one can identify a proxy for $\rho_a(\omega)$ as detailed in the following remark.

\begin{remark}[Practical computation of $\sigma_n$]
\label{rem:practical-sigma}
The fixed-point density matrix $\rho_a(\omega)$ is the limiting state associated with the left endpoint $a$ of the channel product in~\cref{eq:block-channel-random-nonTI}.
In practice, for a finite block $B=[a,b]$, one computes $\Phi_{a,b} = \phi_a \circ \cdots \circ \phi_b$~\eqref{eq:block-channel-random-nonTI} and finds the density matrix $\rho_a$ that minimises
\begin{equation}
  \bigl\|\Phi_{a,b}(\cdot) - \rho_a \operatorname{tr}(\cdot)\bigr\|_{1\rightarrow 1},
\end{equation}
i.e., the best rank-one approximation to $\Phi_{a,b}$.
For a sufficiently large block size $q_B$, this is well approximated by the leading right eigenvector of the block MPS transfer matrix $E_{a,b}$.
The fixed-point state is then obtained as $\sigma_a = \sqrt{\rho_a}$.
\end{remark}
 
Finally, the ergodicity requirement includes IID channels and the translationally invariant case as special cases (more details regarding ergodic channels can be found in Ref.~\cite{Movassagh}).
For stationary ensembles, the marginal distributions of $\rho_a(\omega)$ and $\kappa_a(\omega)$ do not depend on $a$.

\subsubsection{From ergodic convergence to blockwise fixed points}

We now connect \cref{thm:movassagh-adapted} to our preparation protocol with the following theorem.

\begin{theorem}[Blockwise fixed-point convergence for ergodic MPS]
\label{thm:ergodic-blockwise-fixedpoint}
Consider a non-TI tensor~\eqref{eq:nonTI_MPS_def_app} generated by a sequence of tensors $\{A^{[k]}\}$ whose corresponding channels~\eqref{eq:phi-k-random-nonTI} are ergodic and satisfy the assumptions of \cref{thm:movassagh-adapted}.
Then there exists a constant $\tilde{\xi}>0$ (akin to a correlation length), random prefactors $\kappa_a(\omega)<\infty$, and an ergodic sequence of positive-definite density matrices
$\rho_a(\omega)\succ0$ (almost surely) such that with
\begin{equation}
 P^{(\infty)}_a:=\sigma_a(\omega)\otimes\mathbb{I}_{\chi} \quad {\rm where} \quad \sigma_a(\omega) := \sqrt{\rho_a(\omega)},
\end{equation}
for every block $B=[a,b]$ of length $q_B=b-a+1$ the following bound holds:
\begin{equation}
   \|P^{[a,b]} - P^{(\infty)}_a\|_2 =: \|R^{[a,b]}\|_2
  \le
  \Gamma_a(\omega)\,e^{-q_B/\tilde{\xi}},
  \label{eq:polar-bound-random-nonTI}
\end{equation}
where $\Gamma_a(\omega) := \chi^{3/2}\, \kappa_a(\omega)/\lambda_{\min}(\sigma_a(\omega))$ is a random prefactor.
\end{theorem}

\begin{proof}
By using~\cref{thm:movassagh-adapted}, vectorisation of the channels, and 
~\cref{eq:app_bound_11} (for $D=\chi$), we get
\begin{equation}
  \bigl\|\, E_{a,b} - \dket{\rho_a(\omega)}\dbra{\mathbb{I}_{\chi}}\,\bigr\|_{\infty}
  \le \sqrt{\chi}\, \kappa_a(\omega)\,e^{-(b-a+1)/\tilde{\xi}}.
  \label{eq:squared-bound-random-nonTI}
\end{equation}
Using~\eqref{eq:polar-block-transfer}, and changing the matrix view, as done in~\cref{eq:P2-bound}, we get
\begin{equation}
  \label{eq:P2-bound-random-nonTI}
  \|(P^{[a,b]})^2 - \rho_a(\omega)\otimes\mathbb{I}_{\chi}\|_2
  \leq \chi^{3/2}\, \kappa_a(\omega)\,e^{-(b-a+1)/\tilde{\xi}}.
\end{equation}
The fixed-point operator $P^{(\infty)}_a$ is invertible on the virtual-pair space since $\rho_a(\omega)\succ0$ (almost surely) from \cref{thm:movassagh-adapted}.
Through the same derivation as~\cref{eq:sylvester-Rq,eq:Rq-upper},
\begin{equation}
  \label{eq:Rq-upper-nonti}
  \|P^{[a,b]} - P^{(\infty)}_a\|_2 =: \|R^{[a,b]}\|_2
  \leq \frac{\|(P^{[a,b]})^2 - (P^{(\infty)}_a)^2\|_2}
            {\lambda_{\min}(P^{(\infty)}_a)}
  \leq \frac{\chi^{3/2}\, \kappa_a(\omega)}{\lambda_{\min}\!\bigl(\sigma_a(\omega)\bigr)}\,e^{-(b-a+1)/\tilde{\xi}},
\end{equation}
which is~\eqref{eq:polar-bound-random-nonTI} after identifying $q_B=b-a+1$.
\end{proof}

For each interval $B=[a,b]$, define the correction map
\begin{align}\label{eq:def-Cq-nonTI-app}
  P^{[a,b]} = C^{[a,b]} P_a^{(\infty)}.
\end{align}
In turn, following~\cref{eq:spect_Cq} and using~\cref{eq:polar-bound-random-nonTI}, its distance to the identity can be bounded through
\begin{equation}
  \|C^{[a,b]}-\mathbb{I}_{\chi^2}\|_\infty
  \le
  \|R^{[a,b]}\|_2\,\|(P^{(\infty)}_a)^{-1}\|_\infty
  \le
  \Gamma'_a(\omega)\,e^{-q_B/\tilde{\xi}},
  \label{eq:C-bound-random-nonTI}
\end{equation}
where $\Gamma'_a(\omega) := \frac{\Gamma_a(\omega)}{\lambda_{\min}(\sigma_a(\omega))}$.
With this bound, we can relate the success probability of a block-encoded correction map to the interval length $q_B$.

\subsubsection{
Proof of~\cref{thm:random-nonTI}
}\label{proof_th6_app}
For convenience, we restate the \cref{thm:random-nonTI} of the main text that is proven below.
\begin{theorem}
Assume that the random tensors $\{ A^{[k]}\}_{k=1}^{L}$ in~\eqref{eq:nonTI_MPS_def} are IID and satisfy the non-negligible-decoherence condition of Ref.~\cite{Movassagh}.
Then there exists a
constant $\tilde{\xi}>0$, depending only on the underlying distribution, such that for small $\delta$ one can choose a uniform block size $q_\ell = q$ for all $\ell$
in \cref{protocol:nonTI-exact} such that
\begin{equation*}
  q=O\!\left(\tilde \xi \log\!\frac{L}{\delta}\right),
\end{equation*}
and the success probability of the correction step obeys
\begin{equation*}
  p_{\mathrm{succ}} \ge 1-\delta-O(\delta^2)
\end{equation*}
almost surely over the IID draw.
\end{theorem}

\begin{proof}
Recall from~\cref{sec:nonTI} that the chain of $L$ sites is partitioned into $L'$ blocks which,
unlike the TI case, can be of different sizes. The $\ell$-th block occupies sites $a_\ell,\dots,b_\ell$, with
\begin{equation*}
  q_\ell=b_\ell-a_\ell+1,
  \qquad
  \sum_{\ell=1}^{L'}q_\ell=L.
\end{equation*}
We use the left endpoint $a_\ell$ as the representative site of the block, and write $\sigma_\ell:=\sigma_{a_\ell}$, $P_\ell^{(\infty)}:=P_{a_\ell}^{(\infty)}$, $P_\ell^{(q_\ell)}:=P^{[a_\ell,b_\ell]}$, and $C_\ell^{(q_\ell)}:=C^{[a_\ell,b_\ell]}$.
Note that~\cref{thm:ergodic-blockwise-fixedpoint} and~\cref{eq:C-bound-random-nonTI} apply to any choice of pair of indices and thus to any of the blocks.
The bound of the norm of the residual, for each block, is given by
\begin{equation}
  \label{eq:Rq-upper-nonti_blocks}
  \|P_\ell^{(q_\ell)} - P^{(\infty)}_\ell\|_2 =: \|R_\ell^{(q_\ell)}\|_2
  \leq \frac{\chi^{3/2}\, \kappa_{a_\ell}(\omega)}{\lambda_{\min}\!\bigl(\sigma_{\ell}(\omega)\bigr)}\,e^{-q_\ell/\tilde{\xi}}.
\end{equation}
Let us define
\begin{equation}
  \label{eq:gamma-prime-ell-def}
  \Gamma'_\ell(\omega) := \frac{\chi^{3/2}\, \kappa_{a_\ell}(\omega)}{\lambda_{\min}\!\bigl(\sigma_{\ell}(\omega)\bigr)^2},
\end{equation}
which carries an extra factor of $\lambda_{\min}(\sigma_\ell(\omega))^{-1}$ relative to the residual prefactor of \cref{eq:Rq-upper-nonti_blocks}, coming from $\|(P^{(\infty)}_\ell)^{-1}\|_\infty$.
Then \cref{eq:C-bound-random-nonTI} restricted to the block $\ell$ becomes
\begin{equation}
  \|C_\ell^{(q_\ell)}-\mathbb{I}_{\chi^2}\|_\infty
  \le
  \Gamma'_\ell(\omega)\,e^{-q_\ell/\tilde{\xi}}.
  \label{eq:C-bound-random-nonTI_blocks}
\end{equation}

In what follows, we choose block sizes such that they satisfy
\begin{equation}\label{app:bound_sum_cl}
  \sum_\ell \|C_\ell^{(q_\ell)} - \mathbb{I}_{\chi^2}\|_\infty \le \delta /2
\end{equation}
for a small constant $\delta$.
Let us focus on the sequential implementation of $V^{(q)}$ as per~\cref{subsec:vq_seq}.
In that case, since all the isometries are implemented in parallel, the depth of the corresponding circuit depends on
$q_{\max}=\max_\ell q_\ell$.
Thus to guarantee \cref{app:bound_sum_cl}, it is enough to take a uniform block size $q_\ell=q_{\max}$ satisfying
\begin{equation}
  \sum_\ell
  \Gamma'_\ell(\omega)\,e^{-q_{\max}/\tilde{\xi}} \le \delta /2.
\end{equation}
Equivalently, one can choose
\begin{equation}
  q_\ell=q_{\max}\quad\text{for all }\ell,
  \label{eq:q-choice-proportional}
\end{equation}
with
\begin{equation}
  q_{\max} =
  \left\lceil
  \tilde{\xi}
  \log\!\left( 2
    \frac{\sum_{l=1}^{L'}\Gamma'_\ell(\omega)}{\delta}
  \right)
  \right\rceil .
  \label{eq:qmax-proportional}
\end{equation}
This is the non-TI analogue of the TI block-size choice, with realisation-dependent sum $\sum_l\Gamma'_\ell(\omega)$
replacing $L\Gamma$ in \cref{eq:scaling-q}.

Finally, to see that~\cref{app:bound_sum_cl} ensures a constant success probability when implementing the correction maps in Step~4 of~\cref{protocol:nonTI-exact}, repeat the argument of~\cref{thm:constant-success}.
Define $e_\ell:=\Gamma'_\ell(\omega)e^{-q_{\max}/\tilde{\xi}}$, such that $\sum_\ell e_\ell\le\delta/2$.
Then
\begin{equation}
    \|C_\ell^{(q_\ell)}\|_\infty
  \le
  1 + e_\ell.
\end{equation}
For small $\delta \ll 1$, and up to exponentially small normalisation corrections,
\begin{equation}
    p_{\mathrm{succ}}
    = \prod_{\ell=1}^{L'} \|C_\ell^{(q_\ell)}\|_\infty^{-2}
    \geq \prod_{\ell=1}^{L'}(1+e_\ell)^{-2}
    = 1 - 2\sum_{\ell=1}^{L'}e_\ell + O\!\left(\left(\sum_{\ell=1}^{L'}e_\ell\right)^2\right)
    \geq 1-\delta-O(\delta^2).
    \label{eq:psucc-random-nonTI-app}
\end{equation}
This proves the constant-success part of~\cref{thm:random-nonTI}.

It remains to justify the logarithmic scaling of~\cref{eq:qmax-proportional}.
The ergodic theorem above gives a realisation-dependent prescription, but by itself it does not prevent rare large values for $\Gamma'_\ell(\omega)$ from making $\sum_\ell\Gamma'_\ell(\omega)$ grow faster than linearly in the number of blocks.
This is where the IID assumption in~\cref{thm:random-nonTI} is used.
For IID site tensors and disjoint blocks, the $\{\Gamma'_\ell(\omega)\}$ are identically distributed. Furthermore, we assume that
\begin{equation}
  \mathbb{E}[\Gamma'_0] < \infty,
\end{equation}
which is needed to apply the law-of-large-numbers\footnote{\makebox[0pt][l]{\parbox[t]{0.92\textwidth}{This condition is expected for the regular random ensembles considered here. We state it explicitly because the law of large numbers used below requires integrability of $\Gamma'_0$. Pathological heavy-tailed distributions could violate this assumption and require larger realisation-dependent block sizes.}}}:
\begin{equation}
  \label{eq:gamma-prime-lln}
  \frac{1}{L'}\sum_{\ell=1}^{L'}\Gamma'_\ell(\omega)
  \xrightarrow{\;a.s.\;}
  \mathbb{E}[\Gamma'_0].
\end{equation}
Therefore, $\sum_{\ell=1}^{L'}\Gamma'_\ell(\omega)=O(L')=O(L)$ almost surely, and~\cref{eq:qmax-proportional} yields
\begin{equation}
  q_{\max}
  = O\!\left(\tilde{\xi}\log\frac{L}{\delta}\right).
  \label{eq:qmax-from-sum}
\end{equation}
This is the (almost surely) logarithmic block-size scaling stated in~\cref{thm:random-nonTI}.
\end{proof}

\subsection{Extension to other implementations of the block isometries}
\label{app:nonTI-other-Vq}

The analysis above focused on the sequential implementation of the isometries
$V_\ell^{(q_\ell)}$ and resulted in a depth
\begin{equation}
  D_V^{\mathrm{seq}}(\vec q)=O(\chi^4 q_{\max}),
  \qquad
  q_{\max}:=\max_\ell q_\ell ,
\end{equation}
for the implementation used for the non-TI statement in the main text.
The quasi-probabilistic implementation of~\cref{subsec:quasi-prob} can be adapted in the same direct way: the sampling overhead is governed by the chosen block sizes, and by $q_{\max}$, where in the case of IID random tensors, the block sizes are uniform and 
$q_{\max}=O\!\left(\tilde{\xi}\log\frac{L}{\delta}\right)$.

Extending the OAAI-based constructions of~\cref{subsec:SEQ_OAAI,subsec:ren-OAAI}
to the non-TI setting is more delicate. In the TI case, the auxiliary
block size $q_h$ that appears in the OAAI implementation is controlled by a single transfer matrix and a single fixed point. In the non-TI case, the corresponding local
fixed points and prefactors depend on the block position. A direct extension would therefore require choosing position-dependent auxiliary blocks and controlling the associated success probabilities uniformly over the chain.
We therefore leave a full non-TI analysis of the OAAI-based implementations, including the renormalised construction of~\cref{subsec:ren-OAAI}, as future work. The results stated in~\cref{sec:nonTI} rely only on the sequential implementation above.

\section{Depth-estimation and optimisation of free parameters}\label{app:numerics}

\subsection{Depth-estimation details}
\label{app:depth-estimation}
All the depths reported in the numerical studies of \cref{sec:numerics} are estimated rather than obtained from compilation of the circuits into some given primitive gate set. Throughout, we use a CNOT-depth lower bound as the basis of these estimates.
For a generic isometry from an input space of $m$ qubits to an output space of $n$ qubits, this lower bound is given by
\begin{equation}
  N_{\mathrm{iso}}(m,n)
  =
  \tfrac{1}{4}\Bigl(2^{\,n+m+1} - 2^{\,2m} - 2n - m - 1\Bigr),
\end{equation}
following \citet{iten2016}, who provide explicit circuit decompositions achieving this bound up to a multiplicative constant factor close to two.
Based on this lower bound, we define elementary depth estimates as
\begin{equation}
  D_\infty = N_{\mathrm{iso}}(1,\chi^2) = O(\chi^2)
\end{equation}
for the preparation of the RG fixed-point state $\ket{{P^{(\infty)}}(L')}$~\eqref{eq:p_infinity_line},
\begin{equation}
  D_C = N_{\mathrm{iso}}(2\chi^2,2\chi^2) = O(\chi^4)
\end{equation}
for the implementation of the block-encoding of the correction map $C^{(q)}$~\eqref{eq:def-Cq},
\begin{equation}
  D_A = N_{\mathrm{iso}}(\chi,d\chi) = O(d \chi^2 )
\end{equation}
for the implementation of the isometry corresponding to one tensor $A\in \mathbb{C}^{\chi \times d \times \chi}$ appearing in \cref{eq:MPS-def}, and
\begin{equation}
  D_{A'} = N_{\mathrm{iso}}(\chi^2,d\chi^2) = O(d \chi^4)
\end{equation}
for the implementation of the isometry corresponding to one tensor $A' := A \otimes \mathbb{I}_\chi \in \mathbb{C}^{\chi^2 \times d \times \chi^2}$ appearing in \cref{eq:Vq-chain-open}. %

Building on the previous, we can further estimate the depth of the circuits implementing the isometry $V^{(q)}$. For the three implementations used in the numerics, we estimate their depths as follows. For the sequential implementation (\cref{subsec:vq_seq}), we have
\begin{equation}
  D_V^{\rm Seq}(q) = q\,D_{A'}.
\end{equation}
For the amplitude-amplified sequential implementation (\cref{subsec:SEQ_OAAI}), the depth depends on both $q$ and the auxiliary block size $q_h$ and is given by
\begin{equation}
  D_V^{\seqAA}(q,q_h)
  = D_V^{\rm Seq}(q_h) + 3 \bigl((q-q_h)\,D_A + D_\infty + 2D_C\bigr).
\end{equation}
As per~\cref{thm:OAAI}, the factor of $3$ accounts for the $2r+1=3$ calls to $U$ or $U^\dag$ in a one-round ($r=1$) implementation of exact-OAAI,
and the $2$ accounts for the two correction maps appearing in~\cref{eq:Vq_decomp_chain_2}.
To ensure that we only need one round of exact-OAAI, for each fixed $q$, we choose the smallest $q_h$ such that the estimated success probability of the \seqAA construction is larger than $1/2$. Finally, for the renormalised implementation in \cref{eq:Vq_rg}, the depth depends on $q$, $q_2$, and $q_h$:
\begin{equation}
  D_V^{\renAA}(q,q_2,q_h)
  =
  3\bigl(D_\infty + 2D_C\bigr)
  +
  D_V^{\seqAA}(q_2,q_h).
\end{equation}
In practice, for each fixed $q$, we choose the smallest $q_2$ such that the estimated success probability is at least $1/2$, and then evaluate the nested \seqAA contribution with the $q_h$ fixed as per the previous discussion.

Finally, these estimates for the circuit depths of $V^{(q)}$ are  inserted into \cref{eq:Daao,eq:Dsc} to get end-to-end circuit depths. In addition, these depths require computing the probabilities of success given in \cref{eq:p-success,eq:p-success-cat} that are readily obtained based on the spectral norms of the correction maps~\eqref{eq:def-Cq}.
As discussed in the main text, upon evaluation of these circuit depths, we systematically minimised them over $q$ for \cref{eq:Daao} or over $(q,Q)$ for~\cref{eq:Dsc}. These are the optimised depths reported all over~\cref{sec:numerics}.

\subsection{Optimisation plots}
\label{app:optimisation-plots}
In this subsection, we report optimisation landscapes over the parameters $q$ and $Q$.
These are based on a single tensor $A$, defining TI MPS~\eqref{eq:MPS-def} for different $L$, with $A$ drawn from a Haar distribution and with $\chi=4$, $d=2$, and a correlation decay rate $\xi \approx 3.3$.
The qualitative behaviour shown for this single instance is representative of typical TI MPS with the same properties.

\subsubsection{All-at-once strategy}
\cref{fig:d_vs_q_aao} reports the circuit depths for the all-at-once preparation (\cref{subsec:exact-protocol}), obtained from~\cref{eq:Daao}, as a function of the block size $q$, for three representative system sizes $L$. Consider first the linear implementation (left panel) of $V^{(q)}$, for which $D_V(q)\propto q$. At small $q$, the correction success probability $p_{\rm succ}$ vanishes and the depth diverges. Past the optimum, the depth is dominated by $D_V(q)$ and grows linearly in $q$. The numerically selected optima $q^{\star}$ (star) are $18/23/25$ for $L=10^{3}/5\cdot 10^{3}/10^{4}$, consistent with the predicted $q^{\star}=O(\log L)$ scaling. For the logarithmic implementation (right panel), where $D_V(q)\propto  \log q$, we find a characteristic step-like landscape. 
All three sizes $L$ assessed share the same optimum $q^{\star}=32$. In all cases, the optimisation over $q$ is straightforward.

\begin{figure*}[!tp]
  \centering
  \begin{minipage}[t]{0.49\textwidth}
    \centering
    \includegraphics[width=\linewidth]{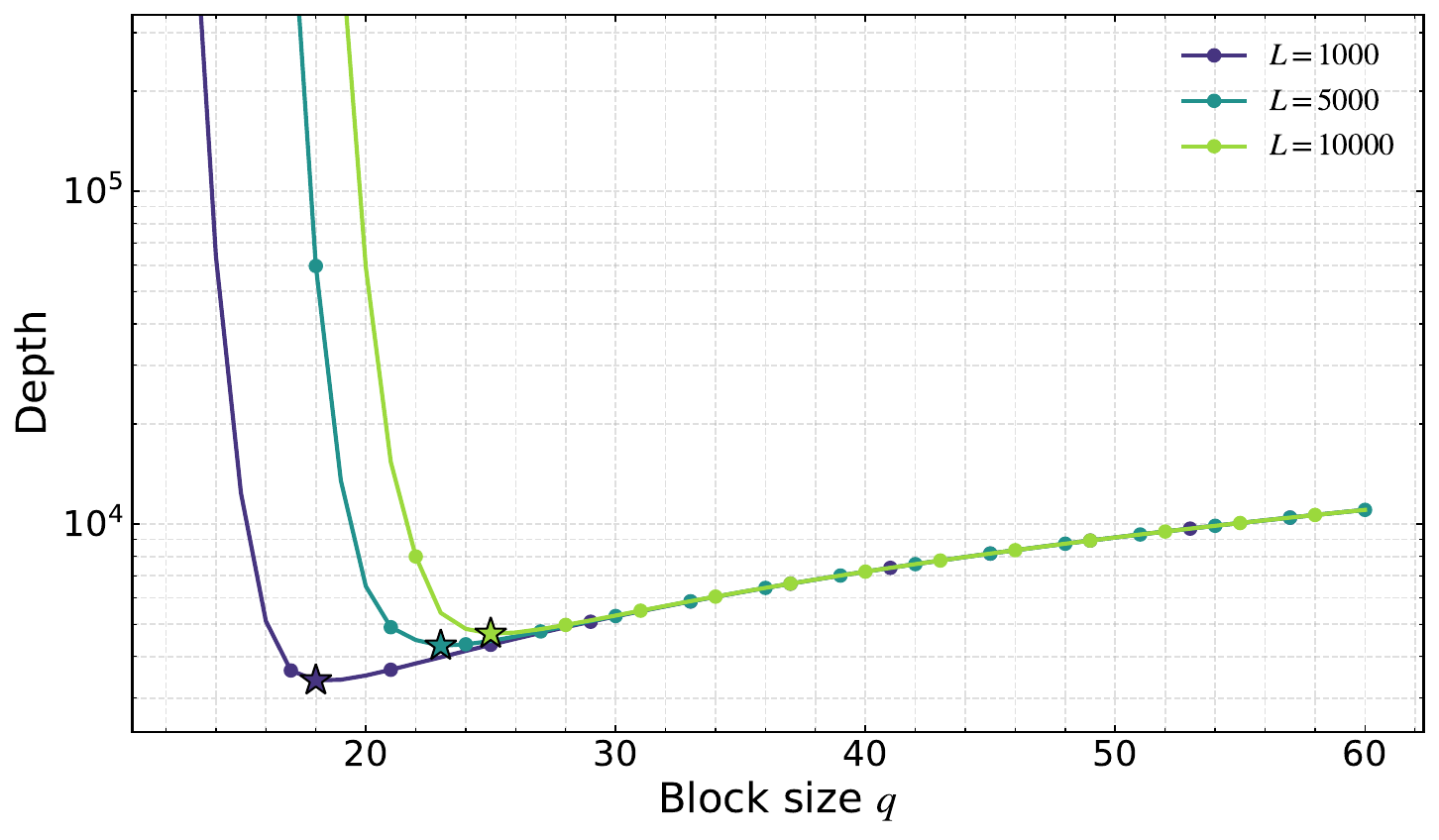}
  \end{minipage}\hfill
  \begin{minipage}[t]{0.49\textwidth}
    \centering
    \includegraphics[width=\linewidth]{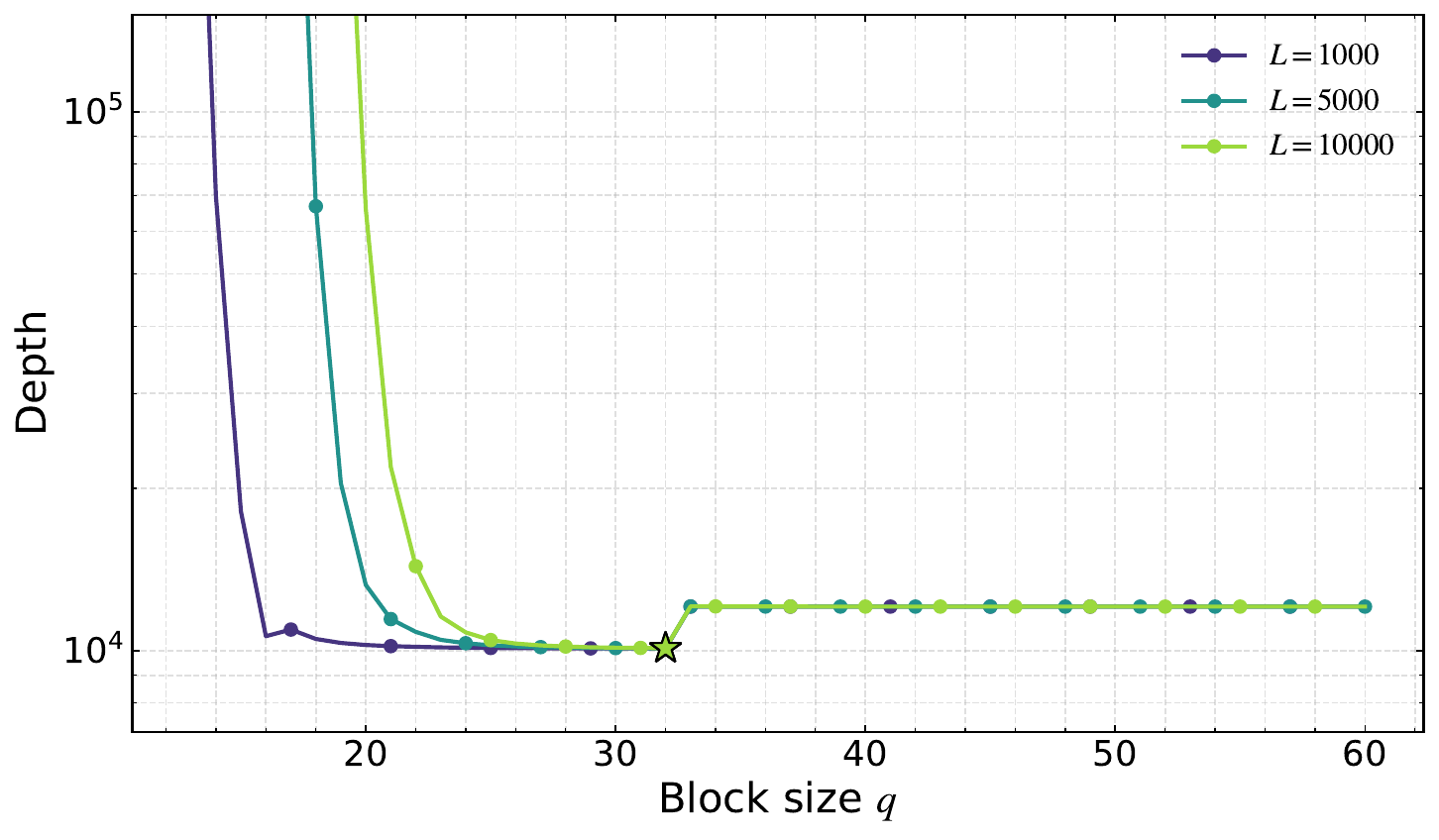}
  \end{minipage}
  \caption{
    Optimisation of the expected depth~\eqref{eq:Daao} of the all-at-once circuit preparations with respect to the block size $q$. The results are based on a single tensor $A$ randomly generated with $\chi=4$, $d=2$ and $\xi \approx 3.3$.
    We consider a \emph{linear} (left panel) and a logarithmic (right panel) implementation of $V^{(q)}$, and three system sizes $L$ (colors in legend). Stars mark the numerically selected optima.
  }
  \label{fig:d_vs_q_aao}
\end{figure*}

\begin{figure*}[!tp]
  \centering
  \begin{minipage}[t]{0.49\textwidth}
    \centering
    \includegraphics[width=\linewidth]{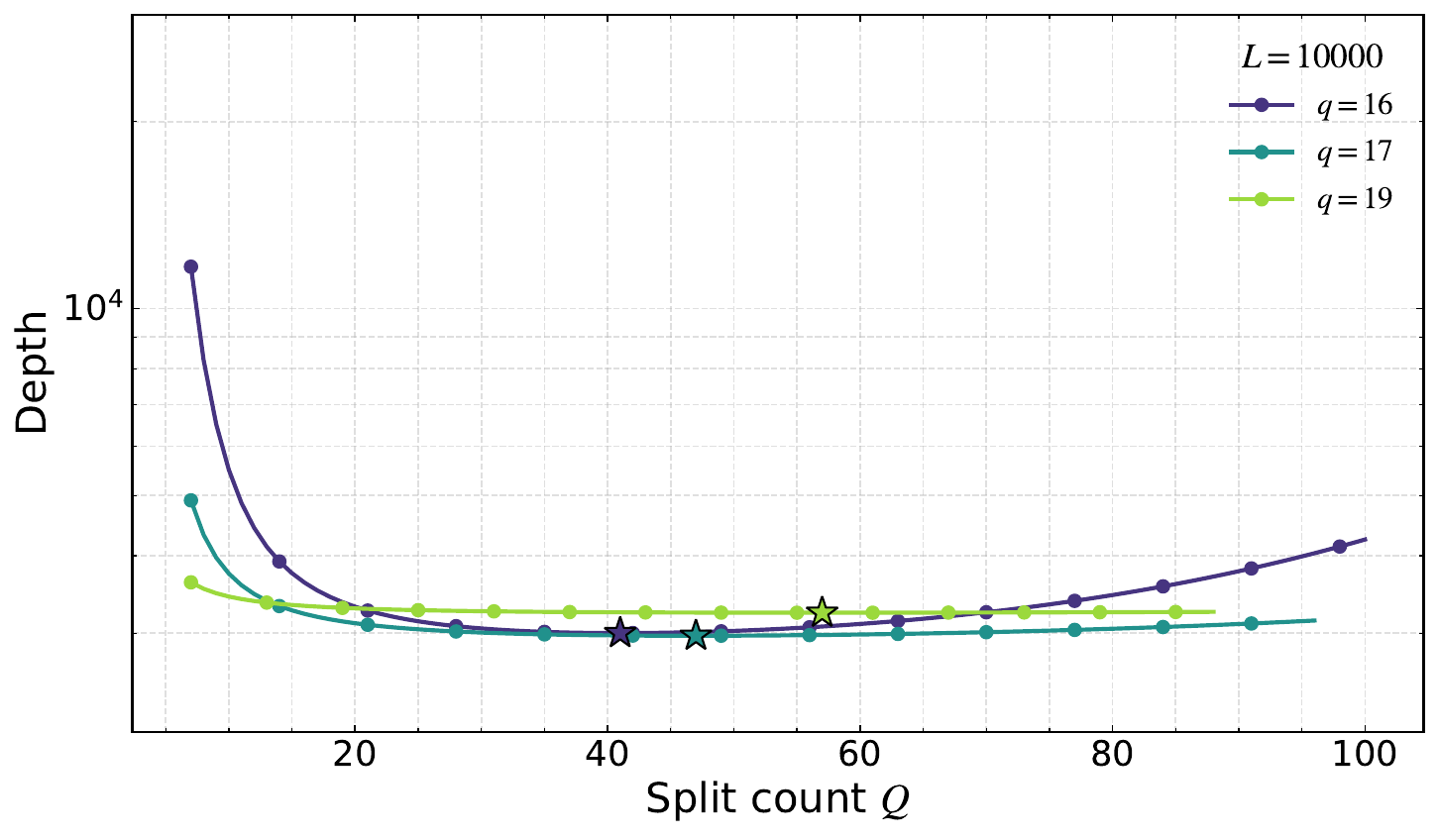}
  \end{minipage}\hfill
  \begin{minipage}[t]{0.49\textwidth}
    \centering
    \includegraphics[width=\linewidth]{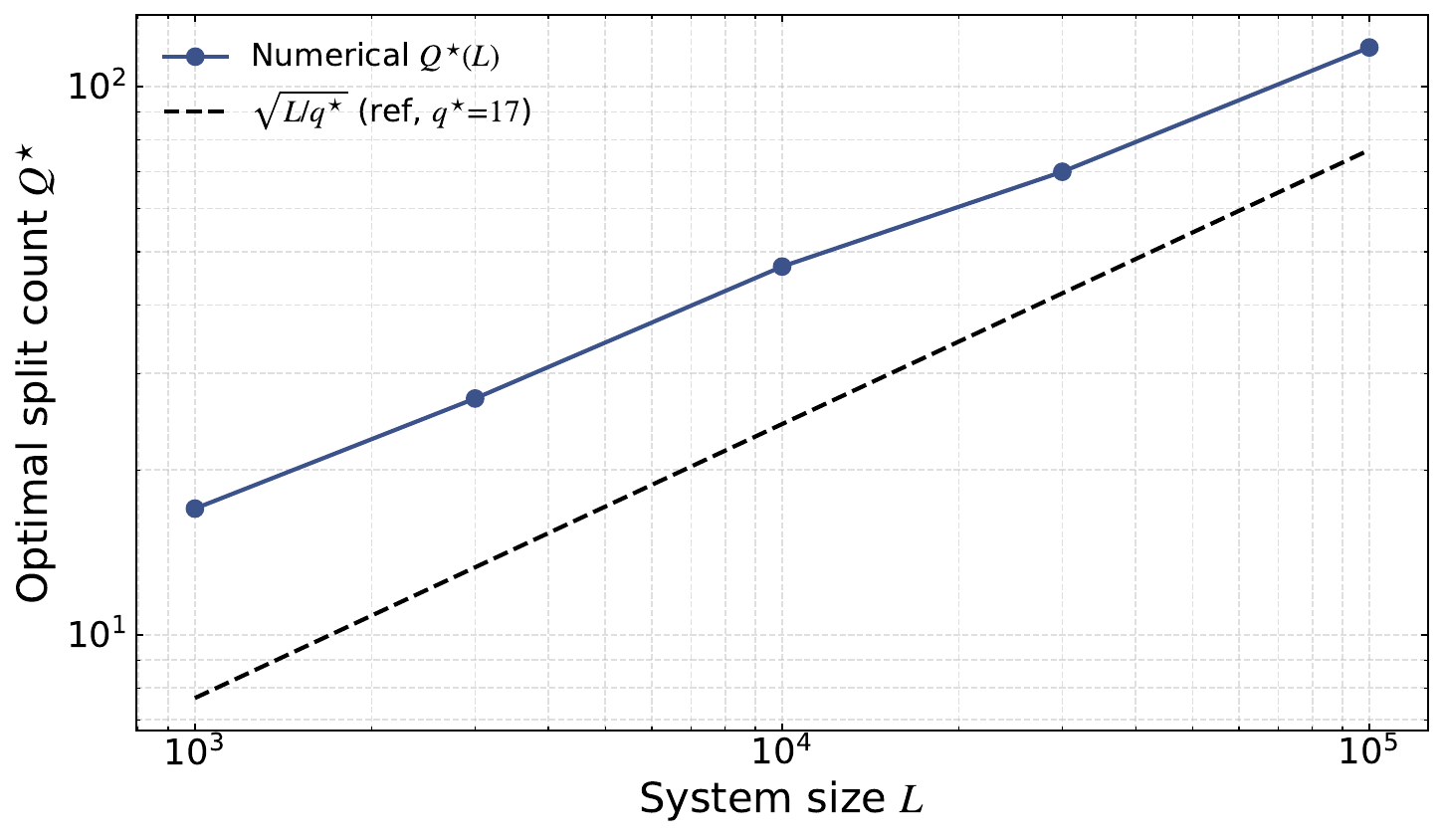}
  \end{minipage}
  \caption{
    Optimisation of the expected depth~\eqref{eq:Dsc} of the split-and-concatenate circuit preparations with respect to the split size $Q$. The results are based on a single tensor $A$ randomly generated with $\chi=4$, $d=2$ and $\xi \approx 3.3$. (Left) Average depth as a function of the split size $Q$ at $L=10^{4}$ for three block sizes $q$. Stars mark the optimum $Q^{\star}(q)$. 
    (Right) Numerically obtained optima $Q^{\star}(L)$ for the split size as a function of the system size $L$ (solid line) together with the analytic $\sqrt{L/q^{\star}}$ reference (dashed line) using $q^{\star}=17$. 
    This confirms the predicted $Q^{\star}\sim\sqrt{L/q^{\star}}$ scaling.
  }
  \label{fig:d_vs_qQ_sac}
\end{figure*}

\subsubsection{Split-and-concatenate strategy}
For the split-and-concatenate protocol (\cref{subsec:sandc}), the average depth given in~\cref{eq:Dsc} depends on both the block size $q$ and the number of splits $Q$. For our analyses, we separate the effect of these two parameters. Fixing $q$, the map $Q\mapsto D_{\snc}(q,Q)$ is convex, with the optimum $Q^{\star}(q)$ close to $\sqrt{L/q}$ as predicted by the cost analysis of~\cref{app:sac}. At fixed $Q$, the map $q\mapsto D_{\snc}(q,Q)$ is again convex in $q$. The joint optimum $(q^{\star},Q^{\star})$ is therefore easily located by a two-stage search.
\cref{fig:d_vs_qQ_sac} (left panel) reports the optimisation landscape of $Q$ at $L=10^{4}$ and for three block sizes around the optimum $q^{\star}=17$. Each curve exhibits a clear convex minimum in $Q$, with the per-$q$ optimum $Q^{\star}(q)$ (depicted as a star) close to $\sqrt{L/q}$. \cref{fig:d_vs_qQ_sac} (right panel) reports the numerically selected $Q^{\star}(L)$ as the system size is varied from $L=10^3$ to $10^5$. The analytic $\sqrt{L/q^{\star}}$ reference, with $q^{\star}=17$, is depicted as a dashed black line 
confirming the predicted $Q^{\star}\sim\sqrt{L/q^{\star}}$ scaling.

\section{Proof of exact oblivious amplitude amplification}
\label{app:proofs_OAAI}
In this section, we prove~\cref{thm:OAAI} that is restated below for convenience.
The proof closely follows the proof of lemma~3.6 in~\cite{Berry2014},
modifying it to extend it to isometries (rather than unitaries) and to make it exact~\cite{hoyer2000}.

\begin{lemma}[exact OAAI]
Let $V:\mathcal H_{S_i}\to\mathcal H_{S_o}$ be an isometry ($V^\dag V=\mathbb{I}_{S_i}$), and let
$U:\mathcal H_{A_i}\otimes\mathcal H_{S_i}\to\mathcal H_{A_o}\otimes\mathcal H_{S_o}$ be a unitary, where the input and output ancillary registers have $n_{A_i}$ and $n_{A_o}$ qubits with $n_{A_o}\le n_{A_i}$.
Assume that
\begin{equation}\label{eq:decomp_0}
U\bigl(\ket{0}_{A_i}\ket{\Psi}_{S_i}\bigr)
=
\sin(\theta)\ket{0}_{A_o}V\ket{\Psi}_{S_i}
+\cos(\theta)\ket{\Phi^\perp}
\end{equation}
holds for any $\ket{\Psi}_{S_i}$.
Here, $(\bra{0}_{A_o}\otimes \mathbb{I}_{S_o})\ket{\Phi^\perp}=0$.
Then, one can synthesize a unitary $\tilde{U}$ such that
\begin{equation}\label{eq:decomp}
\tilde{U}\bigl(\ket{0}_{A_i}\ket{\psi}_{S_i}\bigr)
=
\ket{0}_{A_o}V\ket{\psi}_{S_i},
\end{equation}
using $r=O(1/|\sin(\theta)|)$ amplification rounds, each consisting of one call to $U^\dag$, two generalised reflections around $\ket{0}\bra{0}_{A_i}$ and $\ket{0}\bra{0}_{A_o}$, and one additional call to $U$.
Including the initial application of $U$, this gives $2r+1$ calls to $U$ or $U^\dag$ and $2r$ generalised reflections.
Whenever $\sin^2{\theta}\ge 1/2$, one round is sufficient.
\end{lemma}

\begin{proof}

Let us define
$\ket{\Psi}:=\ket{0}_{A_i}\ket{\psi}_{S_i}$
 and $\ket{\Phi}:=\ket{0}_{A_o}V\ket{\psi}_{S_i}$,
such that~\cref{eq:decomp_0} can be rewritten as
\begin{equation}
U\ket{\Psi}
= \sin\theta\ket{\Phi}+\cos\theta\ket{\Phi^\perp},
\end{equation}
Denote as $\Pi_i=\ket{0}\!\bra{0}_{A_i}\otimes \mathbb{I}_{S_i}$ and as
$\Pi_o=\ket{0}\!\bra{0}_{A_o}\otimes \mathbb{I}_{S_o}$ the projectors onto the ancilla subspaces. We have that $\Pi_o\ket{\Phi^\perp}=0$.
Proceed by defining $\ket{\Psi^\perp}$ implicitly through
\begin{equation}
U\ket{\Psi^\perp}
=
\cos\theta\,\ket{\Phi}-\sin\theta\,\ket{\Phi^\perp},
\end{equation}
and note that it follows from this definition that $\braket{\Psi|\Psi^\perp} = \braket{\Psi|U^\dag U|\Psi^\perp}=0$.
We further claim that
\begin{equation}
  \label{eq:orthogonality_oaai}
  \Pi_i\ket{\Psi^\perp}=0.
\end{equation}
The argument extends the technique of Lemma~3.7 of~\cite{Berry2014} from unitary to isometric $V$.
Define the operator
\begin{equation}
  \mathcal{Q} := (\bra{0}_{A_i}\otimes \mathbb{I}_{S_i})\,U^\dag\,\Pi_o\,U\,(\ket{0}_{A_i}\otimes I_{S_i})
\end{equation}
acting on $\mathcal{H}_{S_i}$.
For every $\ket{\psi}\in\mathcal{H}_{S_i}$,
\begin{equation}
  \bra{\psi}\mathcal{Q}\ket{\psi}
  = \|\Pi_o U\ket{0}_{A_i}\ket{\psi}\|^2
  = \|\sin\theta\,\ket{0}_{A_o}V\ket{\psi}\|^2
  = \sin^2\theta,
\end{equation}
where the last equality uses the isometric property $V^\dag V = \mathbb{I}_{S_i}$.
Since $\mathcal{Q}$ is self-adjoint and its quadratic form is constant, $\mathcal{Q} = \sin^2\theta\cdot \mathbb{I}_{S_i}$.
Inverting the defining equations of $\ket{\Psi}$ and $\ket{\Psi^\perp}$ gives $U^\dag\ket{\Phi} = \sin\theta\,\ket{\Psi}+\cos\theta\,\ket{\Psi^\perp}$, hence
\begin{equation}
\begin{aligned}
  \sin^2\theta\,\ket{\psi}
  &= \mathcal{Q}\ket{\psi}
   = (\bra{0}_{A_i}\otimes \mathbb{I}_{S_i})\,U^\dag\,\Pi_o\,U\ket{\Psi}\\
  &= \sin\theta\,(\bra{0}_{A_i}\otimes \mathbb{I}_{S_i})\,U^\dag\ket{\Phi}\\
  &= \sin^2\theta\,\ket{\psi}
   + \sin\theta\cos\theta\,(\bra{0}_{A_i}\otimes \mathbb{I}_{S_i})\ket{\Psi^\perp},
\end{aligned}
\end{equation}
so $\sin\theta\cos\theta\,(\bra{0}_{A_i}\otimes \mathbb{I}_{S_i})\ket{\Psi^\perp}=0$.
For $\theta\in(0,\pi/2)$ this forces $\Pi_i\ket{\Psi^\perp}=0$, as claimed.
In what follows we engineer the amplitude amplification in the two two-dimensional subspaces 
\begin{equation}\label{def:2d_sub}
\begin{split}
\mathcal K_\psi
&:=\mathrm{span}\{\ket{\Psi},\ket{\Psi^\perp}\},\\
U\mathcal K_\psi
&=\mathrm{span}\{\ket{\Phi},\ket{\Phi^\perp}\}, 
\end{split}
\end{equation}
This is achieved through the generalised reflections defined as
\begin{equation}
    \begin{split}
S_i(\phi)
&=e^{i\phi \Pi_i} = \mathbb{I}-(1-e^{i\phi})\Pi_i, \\
S_o(\varphi)
&=e^{i\varphi \Pi_o} =\mathbb{I}-(1-e^{i\varphi})\Pi_o.
\end{split}
\end{equation}
It can be verified that these leave the subspaces invariant as we have
\begin{equation}
\begin{split}
S_o(\varphi)\ket{\Phi}
&=e^{i\varphi}\ket{\Phi}\; \textrm{and} \; S_o(\varphi)\ket{\Phi^\perp}
=\ket{\Phi^\perp},\\
S_i(\phi)\ket{\Psi}
&=e^{i\phi}\ket{\Psi}\; \textrm{and} \;
S_i(\phi)\ket{\Psi^\perp}
=\ket{\Psi^\perp}.
\end{split}
\label{eq:SoDiag}
\end{equation}
That is, we can represent them in the corresponding subspaces~\eqref{def:2d_sub} as the $2 \times 2$ matrices
\begin{equation}
S_o(\varphi) = \begin{pmatrix}
e^{i\varphi}&0
\\[6pt]
0&1
\end{pmatrix}\; \textrm{and} \; S_i(\phi) = \begin{pmatrix}
e^{i\phi}&0
\\[6pt]
0&1
\end{pmatrix}
\end{equation}
together with $U$ represented as a matrix, from $\mathcal{K}_\psi$ to $U\mathcal{K}_\psi$, as
\begin{equation}
U = \begin{pmatrix}
\sin\theta& \cos\theta
\\[6pt]
\cos\theta & - \sin\theta
\end{pmatrix}.
\end{equation}

This implies that we can use an analysis analogous to the standard exact amplitude amplification~\cite{hoyer2000}, expressed in terms of the success probability $p=\sin^2\theta$.
For arbitrary phases $\phi$ and $\varphi$, define one round of amplification as the application of the operator $U S_i(\phi)U^\dag S_o(\varphi)$ which is represented as
\begin{equation}
\begin{pmatrix}
e^{i\varphi}\!\left(\cos^2\theta + \sin^2\theta\,e^{i\phi}\right)
&
-\left(1 - e^{i\phi}\right)\sin\theta\cos\theta
\\[6pt]
- e^{i\varphi}\!\left(1 - e^{i\phi}\right)\sin\theta\cos\theta
&
\sin^2\theta + \cos^2\theta\,e^{i\phi}.
\end{pmatrix}
\end{equation}
We apply this operator to $U\ket{\Psi}$, which can be expressed in the $U\mathcal{K}_{\psi}$ subspace as
\begin{equation}
  \begin{pmatrix}
    \sin\theta \\
    \cos\theta
  \end{pmatrix},
\end{equation}
and the coefficient corresponding to the orthogonal component $\ket{\Phi^\perp}$ becomes
\begin{equation}
- e^{i\varphi}\!\left(1 - e^{i\phi}\right)
\sin\theta \cos\theta \,\sin\theta
+
\left(\sin^{2}\theta + \cos^{2}\theta\, e^{i\phi}\right)\cos\theta.
\end{equation}
Factoring out $\cos\theta$, this coefficient vanishes precisely when $e^{i\varphi}=\bigl[p+(1-p)e^{i\phi}\bigr]/\bigl[\,p\,(1-e^{i\phi})\,\bigr]$. Since $e^{i\varphi}$ has unit modulus, this fixes $\phi$ and $\varphi$ such that
\begin{equation}\label{eq:angles_final_round}
  \begin{split}
  &\cos\phi
=
1
-
\frac{1}{2p},\\
&\varphi
=
\frac{\pi}{2}
-
\frac{\phi}{2}
+
\arctan\!\left(
\frac{(1-p)\sin\phi}
{p + (1-p)\cos\phi}
\right).
  \end{split}
\end{equation}
At this solution one has $\varphi=\phi$, and the principal branch of $\arctan$ is the correct one because the denominator $p+(1-p)\cos\phi$ remains positive for $p\ge 1/2$. For $1/2\le p<1$, this gives a  choice of $\phi$ and $\varphi$ and produces the desired state with certainty.
The case $p=1$ is already exact and needs no amplification.

For a general initial angle $\theta$, first apply ordinary oblivious amplitude amplification rounds, corresponding to the special choice $\phi=\varphi=\pi$.
The same two-dimensional representation shows that after $k$ such rounds the amplitude onto $\ket{\Phi}$ has been boosted to $\sin((2k+1)\theta)$.
Choosing $k=O(1/|\sin\theta|)$ makes the amplified success probability $\sin^2((2k+1)\theta)\geq1/2$. 
Applying the final arbitrary-phase round with angles given by~\cref{eq:angles_final_round} then removes any $\ket{\Phi^\perp}$ contribution, and thus achieves~\cref{eq:decomp} with the stated $r=O(1/|\sin\theta|)$ amplification-round count.
\end{proof}

\section{Cost analysis of the split-and-concatenate}
\label{app:sac}

In this appendix, we provide additional details to the analysis of the expected circuit depth of the split-and-concatenate (\snc) strategy introduced in~\cref{subsec:improve_success}. We further explain why, when combined with the sequential implementation of the isometries, it effectively halves the required block size $q$ compared to the all-at-once (\aao) approach, as is observed in~\cref{sec:compare-approx-rg}.
For convenience, in what follows we define 
\begin{equation}\label{eq:gamma_prime}
    \Gamma' := \Gamma \lVert (P^{(\infty)})^{-1}\|_\infty.
\end{equation}

Let us first recall results of the all-at-once preparation of the fixed point state from \cref{subsec:exact-protocol}.
From \cref{eq:p-success}, and omitting the norm of the state since it is exponentially close to one, we know that the all-at-once protocol succeeds with probability \begin{equation}
    p_{\mathrm{succ}} = \|C^{(q)}\|_\infty^{-2L'}
\end{equation}
due to the $L'=L/q$ correction maps implemented through block-encoding. 
Using $\|C^{(q)}\|_\infty \le 1 + \Gamma'\,e^{-q/\xi}$, which follows from the residual bound~\eqref{eq:spect_Cq} and the definition of $\Gamma'$~\eqref{eq:gamma_prime}, the requirement that $p_{\rm succ} \geq 1-\delta$ for small $\delta$ becomes $L'\,\Gamma'\,e^{-q/\xi} \le \delta/2$, up to $O(\delta^2)$ corrections. This gives a minimum block size
\begin{equation}
  \label{eq:q-star-aao}
  q^{\star}_{\aao} \;=\; \xi\,\ln\!\frac{2\,L'\,\Gamma'}{\delta}.
\end{equation}
In what follows, we will show that the corresponding block size for the split-and-concatenate approach instead yields $q^{\star}_{\snc} = \tfrac{1}{2}\,\xi\,\ln L'$, half of $q^{\star}_{\aao}$ to leading order.

Let us now analyse the split-and-concatenate approach.
As depicted in~\cref{fig:sac}~(b), the chain consisting of $L' = L/q$ renormalised blocks is further partitioned into $L_{\mathrm{sp}} = L'/Q$ disjoint splits of size $Q$.
Each split is corrected independently, involving $Q-1$ correction maps per split. Once all splits have been successfully corrected, they are merged by applying $L_{\mathrm{sp}}$ additional correction maps at their boundaries (the concatenation step).
The success probability of correcting a single split is
\begin{equation}
  \label{eq:p-sp}
  p_{\mathrm{sp}} = \bigl(\lVert C^{(q)}\rVert_\infty^{-2}\bigr)^{Q-1},
\end{equation}
from which we can derive the expected number of rounds before concatenation.
Let $G_i$ be the number of correction rounds needed for the $i$-th split to succeed for the first time.
Corrections on different splits are independent of each other, and each $G_i$ follows a geometric distribution with parameter $p_{\mathrm{sp}}$ such that
\begin{equation}
    G_i \stackrel{\mathrm{IID}}{\sim} \mathrm{Geom}(p_{\mathrm{sp}}).
\end{equation}
Because all splits can be attempted in parallel at each round, the total number of rounds before concatenation is
\begin{equation}
  G^{\max}_{L_{\mathrm{sp}}} = \max\{G_1,\dots,G_{L_{\mathrm{sp}}}\}.
\end{equation}
The cumulative distribution function of this random variable is given by
\begin{equation}
  \Pr\!\bigl(G^{\max}_{L_{\mathrm{sp}}} \le t\bigr)
  = \bigl(1-(1-p_{\mathrm{sp}})^t\bigr)^{L_{\mathrm{sp}}},
\end{equation}
and its expectation value is obtained through the tail-sum formula:
\begin{equation}
  E\!\bigl[G^{\max}_{L_{\mathrm{sp}}}\bigr]
  = \sum_{t=0}^{\infty}\Bigl[1-\bigl(1-(1-p_{\mathrm{sp}})^t\bigr)^{L_{\mathrm{sp}}}\Bigr].
  \label{eq:E-Gmax-app}
\end{equation}
This formula is used to numerically evaluate the depth of~\cref{eq:Dsc} reported in the main text.
To make the expectation amenable to further analysis and to derive the halving of the effective block size, set $s := -\ln(1-p_{\mathrm{sp}}) > 0$ and $f(t) := 1-(1-e^{-st})^{L_{\mathrm{sp}}}$.
In turn, the Euler-Maclaurin formula gives
\begin{equation}
  \sum_{t=0}^{\infty} f(t)
  \approx \int_0^{\infty} f(u)\,du + \tfrac{1}{2}f(0) + o(1)
  = \frac{H_{L_{\mathrm{sp}}}}{s} + \frac{1}{2} + o(1),
\end{equation}
where $H_n = \sum_{k=1}^{n} 1/k$ is the $n$-th harmonic number.
Using $H_n = \ln n + \gamma_E + o(1)$ with $\gamma_E$ the Euler-Mascheroni constant, we obtain
\begin{equation}
  E\!\bigl[G^{\max}_{L_{\mathrm{sp}}}\bigr]
  \approx \frac{\ln L_{\mathrm{sp}} + \gamma_E}{-\ln(1-p_{\mathrm{sp}})} + \frac{1}{2}
  =: J.
  \label{eq:J-approx}
\end{equation}
For convenience, we recall the expected depth~\eqref{eq:Dsc} of the \snc protocol,
\begin{equation}
  \label{eq:Dsc-recall}
  D_{\snc}(q,Q) = \frac{(D_{\infty}+D_C)\,\mathbb{E}[G^{\max}_{L_{\rm sp}}] + D_C}{p_{\rm cat}} + D_V(q),
\end{equation}
with $\mathbb{E}[G^{\max}_{L_{\rm sp}}]$ given by~\eqref{eq:J-approx}. This is an explicit function of the block size $q$ and the split size $Q$, and the optimal gate count is, in principle, obtained by minimising it jointly over $q$ and $Q$. As detailed further in \cref{app:optimisation-plots}, we perform this minimisation numerically, and report the resulting optimisation landscape over $Q$ in~\cref{fig:d_vs_qQ_sac}, where the selected optimum is seen to lie close to $Q=\sqrt{L'}$. A closed-form joint minimisation over $q$ and $Q$ is not straightforward, but the optimum admits a simple approximate characterisation, which we now derive and which explains the halving of the block size relative to the \aao baseline of~\eqref{eq:q-star-aao}.

The expected depth~\eqref{eq:Dsc-recall} is dominated by the isometry layer, which for the sequential implementation scales as $D_V^{\mathrm{lin}}(q)\sim\chi^4 q$ and grows linearly with the block size $q$. Minimising the depth therefore amounts to choosing the smallest $q$ for which both the splitting and the concatenation succeed with bounded overhead.

The splitting requires the $Q$ maps of a single split to all succeed in the same round, which they do with probability $p_{\rm sp}$~\eqref{eq:p-sp}. For the number of rounds $\mathbb{E}[G^{\max}_{L_{\rm sp}}]$ in~\cref{eq:J-approx} not to blow up, its denominator $-\ln(1-p_{\rm sp})$ must be kept at a value independent of $L$. We therefore choose $Q$ and $q$ so that $p_{\rm sp}=1-\epsilon_{\rm sp}$ for a fixed constant $\epsilon_{\rm sp}$. From~\eqref{eq:p-sp} and $\lVert C^{(q)}\rVert_\infty\le 1+\Gamma'\,e^{-q/\xi}$, we have $1-p_{\rm sp}\approx 2Q\,\Gamma'\,e^{-q/\xi}$ to leading order, so $p_{\rm sp}=1-\epsilon_{\rm sp}$ requires
\begin{equation}
  \label{eq:sac-q-split}
  q\;\ge\;\xi\,\ln\frac{2Q\,\Gamma'}{\epsilon_{\rm sp}}.
\end{equation}
On the other hand, the concatenation merges the $L_{\rm sp}=L'/Q$ splits and succeeds with probability~\eqref{eq:p-success-cat}: 
\begin{equation}
    p_{\rm cat}=\lVert C^{(q)}\rVert_\infty^{-2L_{\rm sp}}.
\end{equation}
Requiring $p_{\rm cat}=1-\epsilon_{\rm cat}$ for a small constant $\epsilon_{\rm cat}$ gives $\epsilon_{\rm cat}\approx 2(L'/Q)\,\Gamma'\,e^{-q/\xi}$ to leading order, which requires
\begin{equation}
  \label{eq:sac-q-cat}
  q\;\ge\;\xi\,\ln\frac{2(L'/Q)\,\Gamma'}{\epsilon_{\rm cat}}.
\end{equation}
The block size must satisfy both~\eqref{eq:sac-q-split} and~\eqref{eq:sac-q-cat}, so it is set by the largest bound given by
\begin{equation}
  \label{eq:sac-q-max}
  q^{\star}(Q)\;\approx\;\xi\,\ln\max\{\,Q,\;L'/Q\,\}+\mathrm{const}.
\end{equation}
The term with the maximum is minimized when its two arguments coincide,
\begin{equation}
  \label{eq:sac-Q-balance}
  Q\;=\;L'/Q.
\end{equation}
That is, it is minimized for $Q=\sqrt{L'}$ that entails $L_{\rm sp}=L'/Q=\sqrt{L'}$ as well. At this choice both bounds yield
\begin{equation}
  \label{eq:q-star-sac}
  q^{\star}_{\snc}\;= \ln\frac{2\sqrt{L'}\,\Gamma'}{\epsilon_{\rm cat}}
  \approx
  \;\tfrac{1}{2}\,q^{\star}_{\aao}.
\end{equation}
This choice halves the required block size relative to the \aao protocol of~\eqref{eq:q-star-aao}, which is the origin of the empirical behaviour observed in~\cref{sec:compare-approx-rg}.
Strictly, $Q=\sqrt{L'}$ is the optimum of this approximate analysis rather than the exact minimiser of $D_{\snc}(q,Q)$. Numerically, however, the optimal split size indeed follows $Q\propto\sqrt{L'}$ closely, as shown in~\cref{fig:d_vs_qQ_sac}~(right panel), so the approximate analysis already captures the observed behaviour.

\paragraph{Achieving the halving exactly with amplitude amplification.}
The factor $\mathbb{E}[G^{\max}_{L_{\rm sp}}]$ appearing in~\cref{eq:J-approx}, due to the necessity of probabilistically reattempting the correction of the splits, can be eliminated by applying exact amplitude amplification in the same spirit as~\cref{thm:det-aa}.
Using $1-p_{\rm sp}\approx 2Q\,\Gamma'\,e^{-q/\xi}$ to leading order, and requiring $p_{\rm sp}>1/2$ (to ensure amplitude amplification in one round), results in the condition 
$Q\,\Gamma'\,e^{-q/\xi}\le \tfrac{1}{4}$.
The optimisation of the depth at $Q^{\star}=\sqrt{L'}$ proceeds as before 
and yields the same leading-order block size $q^{\star}_{\snc} \approx \tfrac{1}{2}\,\xi\,\ln L'$ as in~\eqref{eq:q-star-sac}.
Then, we can simply substitute $\mathbb{E}[G^{\max}_{L_{\rm sp}}]=1$ in \cref{eq:Dsc-recall}, to recover the halving with the amplitude amplification supported variant of the \snc protocol.

\section{Quasi-probabilistic decomposition}
\label{app:incoherent}

In this appendix, we provide a detailed account of the incoherent
implementation of the correction maps $C^{(q)}$ that is used in~\cref{subsec:quasi-prob}. 
As is discussed in~\cref{app:qpd_bg}, where we provide an overview of quasi-probabilistic decompositions (QPDs)~\cite{temme2017,endo2018,cai2023,dai_koczor_2026}, such an approach is limited to tasks of estimation of expectation values and incurs sampling overhead. In~\cref{app:th5}, we prove~\cref{thm:Cq-qpd} of the main text, showing upper bounds for the sampling overhead when implementing a single correction map through QPD. This, in turn, is used to identify scalings of the block size ensuring a constant overall sampling overhead for the preparation of the renormalised state (\cref{app:qpd_aao}) and for the implementation of the isometries (\cref{app:qpd_vq}).

\subsection{Quasi-probabilistic decomposition}\label{app:qpd_bg}

As discussed in the main text, the key idea is to replace the action of a correction map $C^{(q)}\in \mathbb{C}^{\chi^2 \times \chi^2}$ through QPD. This is done by first expressing the correction superoperator (a non-physical map)
$\mathcal{C}(\cdot) := C^{(q)}(\cdot)\,(C^{(q)})^{\dag}$ as a linear combination of
physical quantum channels $\{\mathcal{B}_i\}$,
\begin{equation}\label{eq:qpd_app}
  \mathcal{C} = \sum_i \gamma_i \mathcal{B}_i,
\end{equation}
where the coefficients $\gamma_i \in \mathbb{R}$ may take negative values. We define the $\ell^1$-norm $|\gamma|_1 := \sum_i |\gamma_i|$, such that $|\gamma|_1\geq 1$ with equality only when $\mathcal{C}$ is a quantum channel.

In the QPD implementation, an occurrence of $\mathcal{C}$ is replaced by a channel $\mathcal{B}_{i}$ sampled according to the probability distribution $\{|\gamma_i|/|\gamma|_1\}$.
In particular, consider the task of estimating the expectation value 
\begin{equation}\label{app:expval}
\operatorname{tr}(O \mathcal{C}(\rho))
\end{equation}
for an arbitrary state $\rho$.
Denote as $N_{\mathrm{shots}}$ the total number of measurements (in the eigenbasis of $O$) performed and as $m$ the index of each measurement. Let $\mu_m$ denote the measurement outcomes and $i_m$ the channel indices sampled at each run $m=1,\hdots,N_{\mathrm{shots}}$.
It can be verified from linearity of the trace and~\cref{eq:qpd_app} that
\begin{equation}
  \label{eq:qpd-estimator}
  \hat{\mu}
  :=
  \frac{|\gamma|_1}{N_{\mathrm{shots}}}
  \sum_{m=1}^{N_{\mathrm{shots}}} \mathrm{sign}(\gamma_{i_m})\,\mu_m
\end{equation}
provides an \emph{unbiased} estimator of the desired expectation value.
That is, by rescaling each measurement outcome by a factor $|\gamma|_1 \mathrm{sign}(\gamma_{i_m})$, one recovers the right expectation value on average. Note however that
while this estimator is unbiased, its variance is increased by a factor
$|\gamma|_1^{2}$ (called the \emph{sampling overhead}) compared to a coherent implementation of $\mathcal{C}$ (whenever such implementation exists). That is, to guarantee the same variance one would need to use $|\gamma|_1^2$ times more shots when using QPD.
When $Q$ applications of the correction maps are replaced by QPD, this sampling overhead becomes $|\gamma|_1^{2Q}$.

\begin{table}[t]
\centering
\begin{tabular}{c l l}
\hline\hline
Index & Basis operation & Implementation \\ \hline
1  & $[\mathbb{I}]$
   & Identity channel (no operation). \\[0.3em]
2  & $[X]$
   & Apply a Pauli-$X$ gate. \\[0.3em]
3  & $[Y]$
   & Apply a Pauli-$Y$ gate. \\[0.3em]
4  & $[Z]$
   & Apply a Pauli-$Z$ gate. \\[0.3em]
5  & $[R_x]$
   & $\pi/2$ rotation about the $x$-axis. \\[0.3em]
6  & $[R_y]$
   & $\pi/2$ rotation about the $y$-axis. \\[0.3em]
7  & $[R_z]$
   & $\pi/2$ rotation about the $z$-axis. \\[0.3em]
8  & $[R_{yz}]$
   & Apply $[R_x][R_z]^2$. \\[0.3em]
9  & $[R_{zx}]$
   & Apply $[R_z][R_x][R_z]$. \\[0.3em]
10 & $[R_{xy}]$
   & Apply $[R_x]^2 [R_z]$. \\[0.3em]
11 & $[\pi_x]$
   & Project in the $X$ basis. \\[0.3em]
12 & $[\pi_y]$
   & Project in the $Y$ basis. \\[0.3em]
13 & $[\pi_z]$
   & Project in the $Z$ basis. \\[0.3em]
14 & $[\pi_{yz}]$
   & Apply $[R_z]^3 [R_x]^3 [\pi_z][R_x]^3 [R_z]$. \\[0.3em]
15 & $[\pi_{zx}]$
   & Apply $[R_x][\pi_z][R_x]^3 [R_z]^2$. \\[0.3em]
16 & $[\pi_{xy}]$
   & Apply $[\pi_z][R_x]^2$. \\
\hline\hline
\end{tabular}
\caption{Sixteen quantum channels that form a basis for all single-qubit channels, used for quasi-probabilistic
decompositions.
Each channel admits an $O(1)$-depth implementation.}
\label{tab:basis-channels}
\end{table}

\subsection{Proof of~\texorpdfstring{\cref{thm:Cq-qpd}}{Theorem~5}}\label{app:th5}
For convenience, we restate the \cref{thm:Cq-qpd} of the main text, that is proven below.
\begin{theorem}[Quasi-probabilistic decomposition of $C^{(q)}$]
Let $C^{(q)}$ be the correction map for a block size~$q$~\eqref{eq:def-Cq-imp}, and $\mathcal{C}(\cdot) = C^{(q)}(\cdot)\,{C^{(q)}}^{\dag}$ the corresponding superoperator.
The QPD of $\mathcal{C}$ in a basis of
quantum channels implementable in $O(1)$ 
depth~\cite{temme2017,endo2018}, has an $\ell^1$-norm bounded through
\begin{equation*}
  |\gamma|_1 = 1 + O\!\bigl(\Gamma\,\mathrm{poly}(\chi)\, e^{-q/\xi}\bigr).
\end{equation*}
\end{theorem}

\begin{proof}
A convenient choice of basis for a channel acting on one qubit is given by the sixteen single-qubit operations listed in~\cref{tab:basis-channels}.
This basis, widely used in probabilistic error
cancellation~\cite{temme2017,endo2018}, is generated by Clifford operations
and projective maps. 
This is readily extended to a basis of channels acting on $m$ qubits by taking tensor products of the single-qubit channels (for a total of $16^m$ channels). This ensures that we can always obtain a decomposition of the form~\cref{eq:qpd_app} and that each basis channel can be implemented in constant depth.
We now prove~\cref{thm:Cq-qpd}, that bounds the sampling overhead of the correction superoperator in this basis.

  We encode the $\chi^2$-dimensional virtual-pair space in
  $n_v = 2\lceil \log_2\chi\rceil$ qubits, with Hilbert-space dimension
  $d_v = 2^{n_v} \geq \chi^2$.
  The $16^{n_v}$ tensor-product operations formed from the single-qubit
  channels in~\cref{tab:basis-channels} define the basis used for the
  QPD~\eqref{eq:qpd_app}~\cite{temme2017,endo2018}.
  From~\cref{eq:def-Cq}, we have that the correction map to be implemented is expressed as $C^{(q)} = \mathbb{I}_{\chi^2} + \Delta_C$ with
  $\Delta_C := R^{(q)} (P^{(\infty)})^{-1}$.
  By the residual bound~\eqref{residual} and submultiplicativity, we further get that
  \begin{equation}
    \label{eq:Delta-op-bound}
    \|\Delta_C\|_\infty
    \leq \|R^{(q)}\|_2\,\|(P^{(\infty)})^{-1}\|_\infty
    \leq \Gamma\,\|(P^{(\infty)})^{-1}\|_\infty\,e^{-q/\xi}
    = \Gamma'\,e^{-q/\xi}
    =: \eta.
  \end{equation}
  We recall from \cref{eq:gamma_prime} that $\Gamma':= \Gamma \|(P^{(\infty)})^{-1}\|_\infty$.
  The correction superoperator decomposes as
  $\mathcal{C} = \mathcal{I} + \mathcal{D}$, with $\mathcal{I}$ the identity channel and
  \begin{equation}\label{eq:dec_D}
    \mathcal{D}(\rho)
    := \Delta_C\rho + \rho\,\Delta_C^\dag + \Delta_C\rho\,\Delta_C^\dag.
  \end{equation}
  
  Let $\{P_\alpha\}_{\alpha=0}^{4^{n_v}-1}$ be the $n_v$-qubit Pauli operators,
  such that $\tr[P_\alpha P_\beta] = d_v\,\delta_{\alpha\beta}$.
  The Pauli transfer matrix $T$ of $\mathcal{C}$ has entries
  $[T]_{\alpha,\beta} = d_v^{-1}\tr[P_\alpha\, \mathcal{C}(P_\beta)]$,
  and decomposes as $T = \mathcal{I} + \mathcal{E}$ in terms of the transfer matrix of identity part, with entries $[\mathcal{I}]_{\alpha,\beta} = \delta_{\alpha\beta}$,
  and the transfer matrix of $\mathcal{D}$, with entries $[\mathcal{E}]_{\alpha,\beta} = d_v^{-1}\tr[P_\alpha\,\mathcal{D}(P_\beta)]$. To bound the magnitude of the entries of $\mathcal{E}$, first note that
  \begin{equation}
      \tr[P_\alpha\Delta_C\,P_\beta] \leq \|P_\alpha\|_2 \|\Delta_C\,P_\beta \|_2 \leq d_v \|\Delta_C\|_\infty \leq d_v \eta.
  \end{equation}
  We used Hölder's inequality to obtain the first inequality. For the second one, we used that $\|\Delta_C P_\beta\|_2 \le \|\Delta_C\|_\infty\|P_\beta\|_2$ and
  $\|P_\alpha\|_2 = \sqrt{d_v}$. The last one follows from \cref{eq:Delta-op-bound}. A similar bound applies to $\tr[P_\alpha\,P_\beta\, \Delta^\dag_C]$. In addition, we have
\begin{equation}
      \tr[P_\alpha\Delta_C\,P_\beta\Delta^\dag_C] \leq \|P_\alpha\Delta_C\|_2 \|\,P_\beta \Delta^\dag_C\|_2 \leq d_v \|\Delta_C\|^2_\infty \leq d_v \eta^2.
  \end{equation}
  Hence, from the decomposition in \cref{eq:dec_D}, we get
  \begin{equation}
    \label{eq:E-entry-bound}
    |[\mathcal{E}]_{\alpha,\beta}|
    \leq 2\eta + \eta^2.
  \end{equation}
  It remains to relate this upper bound, on elements of the transfer matrix of $\mathcal{E}=T - \mathcal{I}$, to a bound on the $\ell^1$-norm $|\gamma|_1$. This is readily achieved by noting that the set of coefficients ${\gamma_i}$ and the entries of $T$ are decompositions of $\mathcal{C}$ in different bases.
  
  Let $\{\mathcal{B}_i\}_{i=1}^{16^{n_v}}$ be the tensor-product basis channels obtained from~\cref{tab:basis-channels}, and denote by $\gamma\in\mathbb{R}^{16^{n_v}}$ the vector of weights $\gamma_i$ from~\cref{eq:qpd_app}.
  Collect, as the columns of a synthesis matrix $B\in\mathbb{R}^{16^{n_v}\times 16^{n_v}}$, the vectorised Pauli transfer matrices of the basis channels, using a vectorisation $\mathrm{vec}(\cdot)$ that respects the qubit tensor factorization.
  The Pauli transfer matrix $T$ of $\mathcal{C}$ then satisfies $\mathrm{vec}(T) = B\gamma$.
  The $16^{n_v}$ basis channels are linearly independent and span the space of superoperators acting on $n_v$ qubits, so $B$ is invertible. In turn, we have
  \begin{equation}
    \gamma = B^{-1}\,\mathrm{vec}(T).
  \end{equation}
  Both the channel basis and the Pauli basis factorize over the $n_v$ qubits, so $B = b^{\otimes n_v}$ for the single-qubit synthesis matrix $b\in \mathbb{R}^{16 \times 16}$, and hence $B^{-1} = (b^{-1})^{\otimes n_v}$.
  
  Recall that $T = \mathcal{I} + \mathcal{E}$, with $\mathcal{I}$ the Pauli transfer matrix of the identity channel.
  Since the identity channel is itself a basis element, $B^{-1}\,\mathrm{vec}(\mathcal{I})$ is the corresponding unit coordinate vector with $\ell^1$-norm $\|B^{-1}\,\mathrm{vec}(\mathcal{I})\|_1 = 1$.
  Therefore
  \begin{equation}
    |\gamma|_1 = \bigl\| B^{-1}\,\mathrm{vec}(\mathcal{I} + \mathcal{E}) \bigr\|_1 \le 1 + \|B^{-1}\|_{1\to 1}\,\|\mathrm{vec}(\mathcal{E})\|_1,
  \end{equation}
  where $\|B^{-1}\|_{1\to 1}$ is the \emph{maximum absolute column sum norm} of $B^{-1}$.
  This induced norm is multiplicative under tensor products, so
  \begin{equation}
    \|B^{-1}\|_{1\to 1}
    = \|b^{-1}\|_{1\to 1}^{\,n_v}
    = \mathrm{poly}(\chi),
  \end{equation}
  as $b$ (and $b^{-1}$) have entries and dimensions independent of $\chi$ and $q$.
  By~\eqref{eq:E-entry-bound} and since $\mathcal{E}$ has $16^{n_v}$ entries,
  \begin{equation}
    \|\mathrm{vec}(\mathcal{E})\|_1
    = \sum_{\alpha,\beta} |[\mathcal{E}]_{\alpha,\beta}|
    \le 16^{n_v}\,(2\eta+\eta^2)
    = \mathrm{poly}(\chi)\,(2\eta+\eta^2).
  \end{equation}
  
  Combining these bounds with $\eta = \Gamma'\,e^{-q/\xi}$ from~\cref{eq:Delta-op-bound},
  \begin{equation}
    |\gamma|_1
    \le 1 + \mathrm{poly}(\chi)\,(2\eta+\eta^2)
    = 1 + O\!\bigl(\textrm{poly}(\chi) \Gamma'\,e^{-q/\xi}\bigr),
  \end{equation}
  which was reported in~\cref{eq:bound_l1} in the main text.
  \end{proof}

\subsection{Exact preparation of the renormalised state for the all-at-once protocol of~\cref{subsec:exact-protocol}}\label{app:qpd_aao}
The variance of the estimator~\eqref{eq:qpd-estimator} is increased by a
factor $|\gamma|_1^2$ per application of a correction map.
Since the all-at-once protocol requires the implementation of $L' = L/q$ of such correction maps to prepare $\ket{P^{(q)}(L')}$ (\cref{protocol:exact-prep}), the total sampling overhead relative to a coherent implementation is given by
\begin{equation}
  |\gamma|_1^{2L'}
  = \Bigl(1 + O\!\bigl(\Gamma'\textrm{poly}(\chi)\,e^{-q/\xi}\bigr)\Bigr)^{2L/q}
  = 1 + O\!\Bigl(\Gamma'\textrm{poly}(\chi)\,\tfrac{L}{q}\,e^{-q/\xi}\Bigr).
\end{equation}
This shows that to make this sampling overhead constant, independent of $L$, we can choose a block size
\begin{equation}\label{eq:q_qpd}
    q  =O\!\bigl(\xi\log(\chi\Gamma L)\bigr).
\end{equation}
That is, for this choice of $q$ we achieve a circuit depth $D_C^{\rm QPD}=O(1)$ for the implementation of the correction maps, while maintaining constant sampling overhead.

\subsection{Sequential implementation of the isometries}\label{app:qpd_vq}
We now show that the isometries $V^{(q)}$ themselves are implemented using Bell measurements together with OAAI (\cref{thm:OAAI}) and that $V^{(q)}$ can be implemented with circuit depth $O(\chi^{3} q)$ at constant sampling overhead.
To see this, we begin by rewriting the decomposition
in~\cref{eq:Vq_decomp_chain} as
\begin{equation}
  \label{eq:Vq_Bell}
  \frac{1}{\sqrt{\chi}} V^{(q)} =
  \raisebox{-0.4\height}{\hspace{1em}\includegraphics[width=0.28\textwidth]{figs/Vq_Bell.pdf}\\[6pt] }.
\end{equation}
Here, we introduced a Bell measurement satisfying the equality
\begin{equation}
  \label{eq:bell}
  \raisebox{-0.4\height}{\hspace{1em}\includegraphics[width=0.08\textwidth]{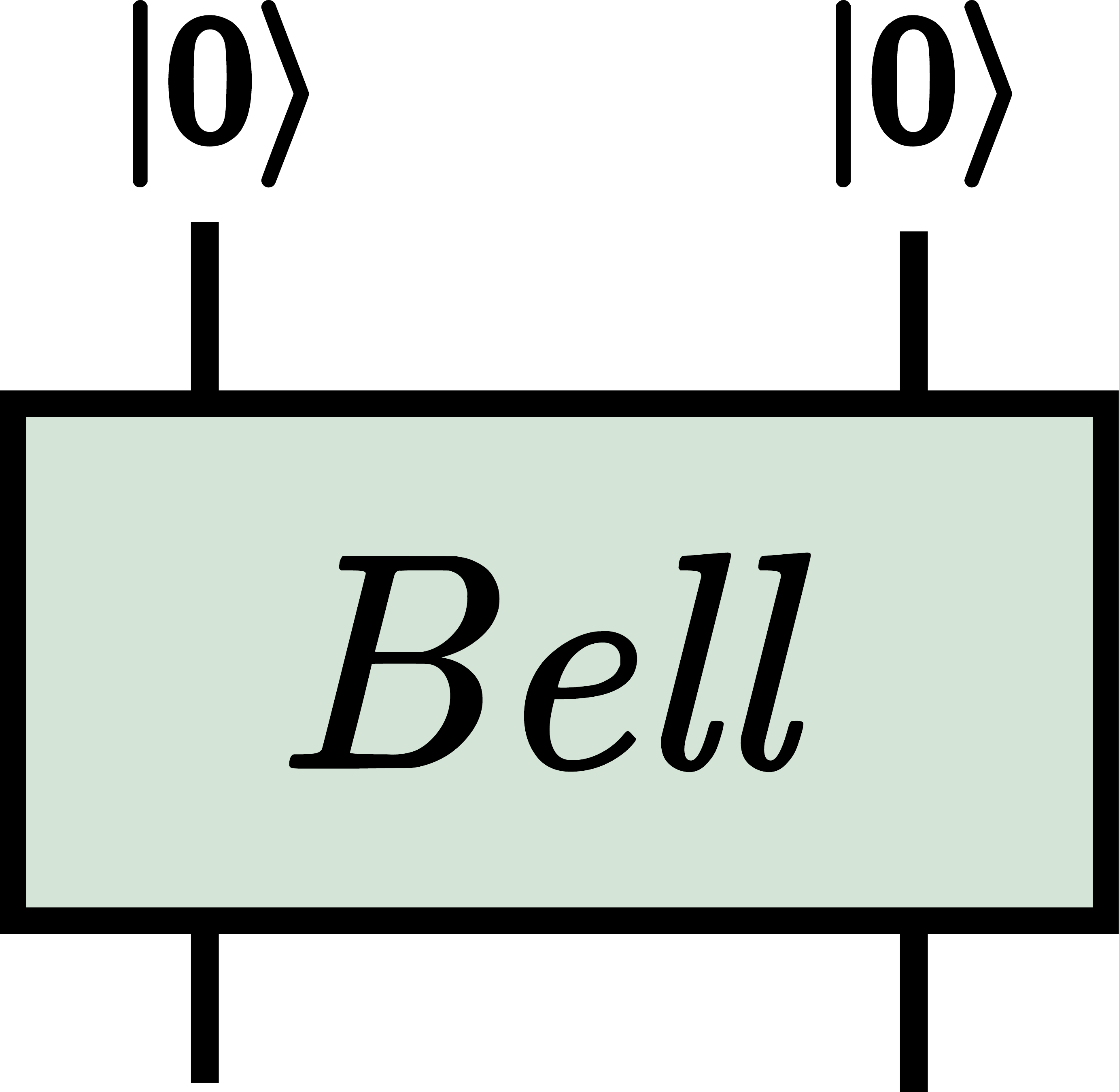}\\[6pt] }
  =
  \raisebox{-0.4\height}{\hspace{1em}\includegraphics[width=0.075\textwidth]{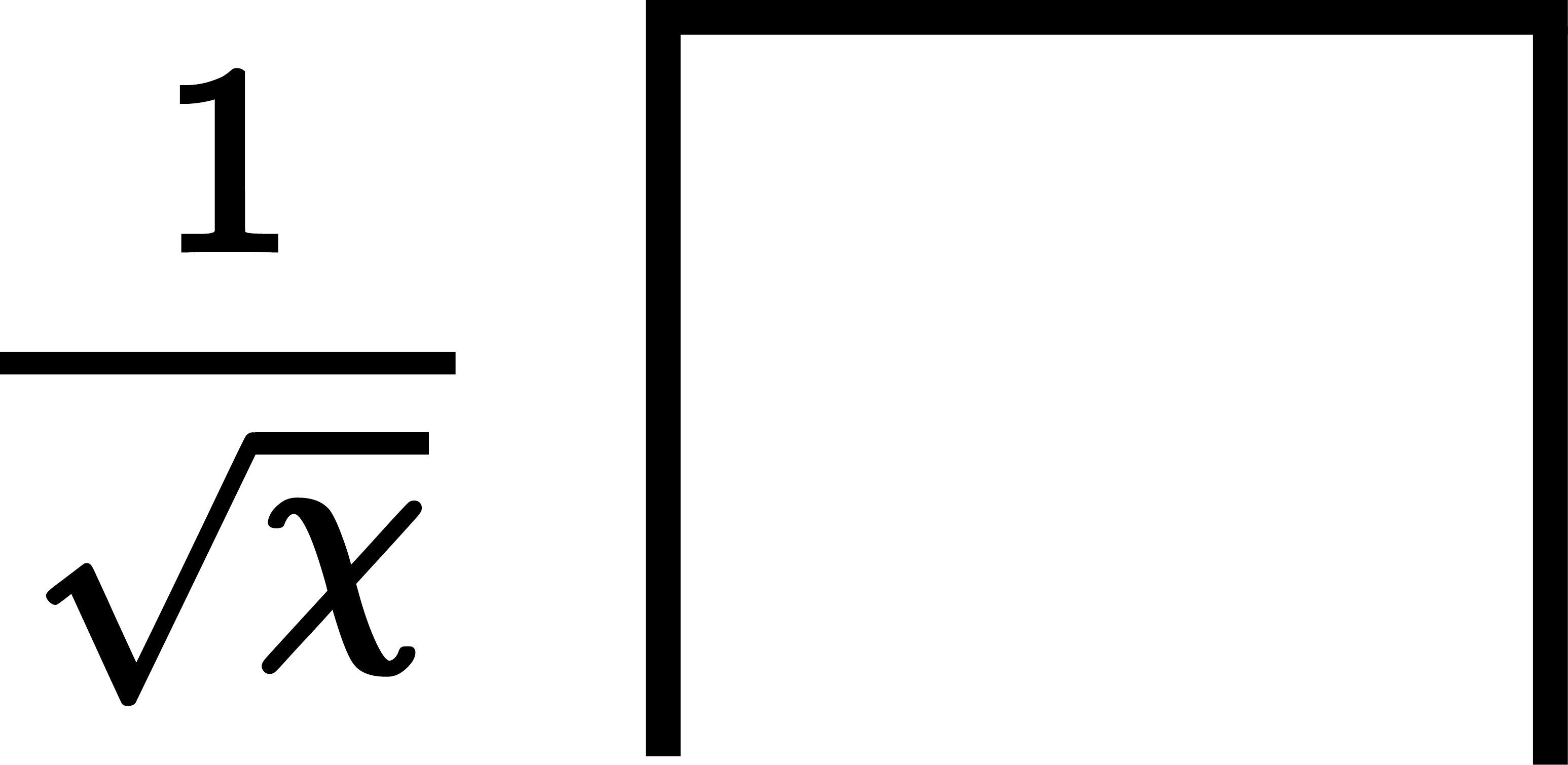}\\[6pt] }.
\end{equation}

Next, we apply a right-canonical decomposition to the length-$q$ MPS
including the $\sigma^{-1}$ term.
By repeatedly performing singular value decompositions from the right to the left, we obtain a right-canonical form in which the tensors satisfy isometric constraints.
As a result, the right-hand side of~\cref{eq:Vq_Bell} can be written as
\begin{equation}
  \label{eq:Vq_rightcanonical}
  \raisebox{-0.4\height}{\hspace{1em}\includegraphics[width=0.33\textwidth]{figs/Vq_rc.pdf}\\[6pt]} \; ,
\end{equation}
where $\{\tilde{A}_k\}_{k=1}^{q}$ are the resulting right-canonical tensors, ordered from left ($k=1$) to right ($k=q$).
For $k=2,\dots,q$, by construction each tensor $\tilde{A}_k$ is an exact isometry:
\begin{equation}
  \label{eq:tildeAk_iso}
  \sum_{i=1}^{d} \tilde{A}_k^{\,i}\bigl(\tilde{A}_k^{\,i}\bigr)^{\dag}
  \;=\; \mathbb{I}_{\chi}.
\end{equation}
The leftmost tensor $\tilde{A}_1$, however, requires more care, as it absorbs
the factor $\sigma^{-1}$ together with the overall scale of the chain.
Consider the entire length-$q$ chain $\sigma^{-1}A^{(q)}$ as
a single map from the left-virtual space to the joint physical and right-virtual
space, i.e.\ as a matrix in $\mathbb{C}^{d^q\chi\times\chi}$. Up to a correction
exponentially small in $q$, this map is proportional to an isometry.
To see that, recall that the tensors $A$ are left-canonical, with $\sum_i A^{i\dag}A^i=\mathbb{I}_\chi$,
so that $\mathbb{I}_\chi$ is the left fixed point of the transfer
matrix~\eqref{eq:transfer-left-ev}. The right-contracted chain is the $q$-fold
iterate of the MPS channel $\mathcal{E}(X)=\sum_i A^iXA^{i\dag}$ applied to the
identity,
\begin{equation*}
  \sum_{\mathbf{i}} [A^{(q)}]^{\mathbf{i}}\bigl([A^{(q)}]^{\mathbf{i}}\bigr)^{\dag}
  \;=\; \mathcal{E}^{q}(\mathbb{I}_\chi)
  \;=\; \chi\,\sigma^2 + O\!\bigl(e^{-q/\xi}\bigr),
\end{equation*}
where we used $E^{\infty}=\dket{\sigma^2}\dbra{\mathbb{I}_\chi}$ from~\cref{eq:Einfty-rank-one}
together with $\dbraket{\mathbb{I}_\chi}{\mathbb{I}_\chi}=\tr(\mathbb{I}_\chi)=\chi$.
The exponentially small remainder is governed by the transfer-matrix gap and decays at the correlation length~\eqref{eq:correll}, as established in~\cref{app:residual-bound}.
The factor $\chi$ comes from $\tr(\mathbb{I}_\chi)$ and originates from the left-isometric normalisation of $A$. In turn, conjugating by the edge
factor $\sigma^{-1}$ (with $\sigma=\sqrt{\rho}\succ0$)
yields 
\begin{equation*}
  \sum_{\mathbf{i}} \bigl(\sigma^{-1}[A^{(q)}]^{\mathbf{i}}\bigr)\bigl(\sigma^{-1}[A^{(q)}]^{\mathbf{i}}\bigr)^{\dag}
  \;=\; \chi\bigl(\mathbb{I}_\chi + O(e^{-q/\xi})\bigr).
\end{equation*}
Since the bulk right-canonical tensors $\tilde{A}_2,\dots,\tilde{A}_q$ are 
isometries~\cref{eq:tildeAk_iso}, their contraction telescopes to $\mathbb{I}_\chi$
and the scale is carried by the leftmost tensor. We thus find that
$\tilde{A}_1$ is isometric up to the overall scale $\chi$ and an
exponentially small corrections in $q$:
\begin{equation}
  \label{eq:tildeA1_almost_iso}
  \sum_{i=1}^{d} \tilde{A}_1^{\,i}\bigl(\tilde{A}_1^{\,i}\bigr)^{\dag}
  \;=\; \chi(\mathbb{I}_{\chi} + O\!\bigl(e^{-q/\xi}\bigr)).
\end{equation}

Since $\tilde{A}_1$ obeys~\cref{eq:tildeA1_almost_iso} rather than an exact
isometric constraint, it has spectral norm $\sqrt{\chi}$ (up to the small corrections). We therefore rescale it to
$\tilde{A}_1/\sqrt{\chi}$, which can be block-encoded into a unitary matrix (\cref{app:block-encoding}). 
Letting $U_q$ denote the resulting unitary block encoding of $V^{(q)}$, for an
arbitrary state $\ket{\phi}$ on the $\chi^2$-dimensional virtual space it acts as
\begin{equation}
  \label{eq:Uq_embeds_Vq}
  U_q\,\ket{\phi}\ket{0}_{a_{\mathrm{in}}}
  \;=\;
  \frac{1}{\chi\bigl(1+O(e^{-q/\xi})\bigr)}\,
  V^{(q)}\,\ket{\phi}\ket{0}_{a_{\mathrm{o}}}
  \;+\;
  \ket{\phi_\perp},
\end{equation}
where $a_{\mathrm{in}}$ and $a_{\mathrm{o}}$ indicate the input and output ancilla registers, respectively, and $\bra{\phi_\perp}(\ket{\cdot}\ket{0}_{a_{\mathrm{o}}})=0$.
The prefactor $1/\chi$ is the product of two contributions of $1/\sqrt{\chi}$ each: one from the postselected Bell measurement~\eqref{eq:bell}, seen in~\cref{eq:Vq_Bell}, and a further $1/\sqrt{\chi}$ 
due to the rescaling of $\tilde{A}_1$~\eqref{eq:tildeA1_almost_iso} 
needed for its block encoding.
Assuming $q$ to be large enough such that $O(e^{-q/\xi})$ is negligible, as would be the case for~\cref{eq:q_qpd}, we can use the exact OAAI of~\cref{app:proofs_OAAI} to exactly implement $V^{(q)}$.
The latter entails $O(\chi)$ calls to $U_q$ and $U_q^\dag$ for a total of $O(\chi)$ (inverse of the) correction maps within the circuit.

First, and following the argument of~\cref{app:qpd_aao}, we see that to implement $L'=O(L)$ isometries, a choice of $q$ scaling as~\cref{eq:q_qpd} is sufficient to ensure constant sampling overhead.
Second, given that the circuit depth of a single $U_q$ scales as $O(q\chi^{2})$ (due to the implementation of the $q$ isometries $\tilde{A}_i$), the overall circuit depth required to implement a single $V^{(q)}$ with unit success probability becomes $O(q\chi^3)$.

Overall this shows that through the use of QPD, and for a choice of block size as per~\cref{eq:q_qpd}, we can achieve both the all-at-once exact preparation (\cref{app:qpd_aao}) and subsequently implement the required isometries (\cref{app:qpd_vq}) with circuit depth scaling as $O(q \chi^3)$ and with constant sampling overhead. That is, we achieve
\begin{equation}
    D^{\rm QPD}=O\!\bigl(\chi^3\,\xi\,\log(\chi\Gamma L)\bigr),
\end{equation}
as reported in~\cref{tab:Vq-depth-summary}.

\end{document}